\documentclass[11pt]{article}

\usepackage{odonnell}
\usepackage{graphicx}
\usepackage{epstopdf}
\usepackage{relsize}
\usepackage{subcaption}
\usepackage{mleftright}
\usepackage{braket}

\newcommand{\calL}{\mathcal{L}}
\newcommand{\calM}{\mathcal{M}}

\newcommand{\rnote}[1]{\footnote{\color{blue}Ryan: {#1}}}
\newcommand{\jnote}[1]{\footnote{\color{red}John: {#1}}}

\newcommand{\unif}[1]{\mathsf{Unif}_{#1}}
\newcommand{\pan}[2]{\mathsf{P}^{#1}_{#2}}

\newcommand{\schurrat}[1]{\mathrm{Ratio}_{#1}}

\newcommand{\lds}{\mathrm{LDS}}

\newcommand{\trace}{\mathrm{tr}}
\newcommand{\unitary}[1]{U_{#1}}
\newcommand{\irrunitary}[1]{\widehat{U}_{#1}}
\newcommand{\symm}[1]{{\mathfrak S}_{#1}}
\newcommand{\irrsymm}[1]{\widehat{\mathfrak S}_{#1}}
\newcommand{\symmset}[1]{{\mathfrak S}(#1)}

\newcommand{\SWdens}[2]{\mathrm{SW}^{#1}_{#2}}
\newcommand{\SWdensProb}[3]{\mathrm{SW}^{#1}_{#2}(#3)}
\newcommand{\SWunif}[2]{\mathrm{SW}^{#1}_{#2}}
\newcommand{\SWunifProb}[3]{\mathrm{SW}^{#1}_{#2}(#3)}
\newcommand{\Planch}[1]{\mathrm{Planch}_{#1}}

\newcommand{\srep}{{\mathtt P}}                   
\newcommand{\urep}{{\mathtt Q}}

\newcommand{\sirrep}{{\mathtt p}}
\newcommand{\uirrep}{{\mathtt q}}

\newcommand{\schurbasis}{U_{\mathrm{Schur}}}

\newcommand{\rsk}{\mathrm{RSK}}
\newcommand{\word}[1]{\mathcal{W}_{#1}}

\newcommand{\allpartitions}{\textnormal{Par}}
\newcommand{\four}[1]{\wt{#1}}

\renewcommand{\r}{r}

\newcommand{\weight}{\mathrm{wt}}

\newcommand{\df}[1]{#1!!}       
\newcommand{\rising}[2]{#1^{\uparrow #2}}
\newcommand{\falling}[2]{#1^{\downarrow #2}}

\newcommand{\evencontents}[1]{[#1]_{\textnormal{even}}}
\newcommand{\oddcontents}[1]{[#1]_{\textnormal{odd}}}
\newcommand{\evenhooks}[1]{\calH^{#1}_{\textnormal{even}}}
\newcommand{\evenhandhooks}[2]{\calH^{#2}_{#1}}
\newcommand{\hand}{\mathrm{hand}}
\newcommand{\foot}{\mathrm{foot}}

\newcommand{\pstar}[1]{p^{*}_{#1}}
\newcommand{\pstarc}[2]{p^{*}_{#1,#2}}
\newcommand{\psharp}[1]{p^{\sharp}_{#1}}
\newcommand{\qsharp}[1]{q^{\sharp}_{#1}}

\newcommand{\cyctype}[1]{\rho(#1)}     

\newcommand{\dcurve}[2]{d_{1}(#1,#2)}

\theoremstyle{definition}
\newtheorem*{nameddefinition}{\theoremname}
\newenvironment{nameddef}[1]{ \renewcommand{\theoremname}{#1} \begin{nameddefinition}} {\end{nameddefinition}}

\begin{document}

\title{Quantum Spectrum Testing}

\author{Ryan O'Donnell$^*$
\and
John Wright\thanks{Department of Computer Science, Carnegie Mellon University.  Some of this work performed while the first-named author was at the Bo\u{g}azi\c{c}i University Computer Engineering Department, supported by Marie Curie International Incoming Fellowship project number 626373.  Supported also by NSF grants CCF-0747250 and CCF-1116594.
The second-named author is also supported by a Simons Fellowship in Theoretical Computer Science and completed some of this work while visiting Columbia University.
\texttt{\{odonnell,jswright\}@cs.cmu.edu}}}

\maketitle

\begin{abstract}
In this work, we study the problem of testing properties of the spectrum of a mixed quantum state.
Here one is given $n$ copies of a mixed state $\rho\in \C^{d\times d}$
and the goal is to distinguish (with high probability) whether $\rho$'s spectrum satisfies some property $\calP$
or whether it is at least $\eps$-far in $\ell_1$-distance from satisfying~$\calP$.
This problem was promoted in the survey of Montanaro and de Wolf~\cite{MdW13}
under the name of testing unitarily invariant properties of mixed states.  It is the natural quantum analogue of the classical problem of testing symmetric properties of probability distributions.

Unlike property testing probability distributions---where one generally hopes for algorithms with sample complexity that is sublinear in the domain size---here the hope is for algorithms with \emph{subquadratic} copy complexity in the dimension~$d$.  This is because the (frequently rediscovered) ``empirical Young diagram (EYD) algorithm''~\cite{ARS88,KW01,HM02,CM06} can estimate the spectrum of any mixed state up to $\eps$-accuracy using only $\wt{O}(d^2/\eps^2)$ copies.
\ignore{On the other hand,  the work of Childs et al.~\cite{CHW07}
 has shown that given a mixed state $\rho$,
$\Theta(r_1)$ copies are necessary and sufficient to distinguish between the cases
when $\rho$'s spectrum is uniform on either $r_1$
or $r_2$ eigenvalues, for any $r_2 \geq 2r_1$.
This can be viewed as a quantum version of the classical birthday paradox and provides linear lower bounds for a variety of testing problems.
}In this work, we show that given a mixed state $\rho \in \C^{d \times d}$:
\begin{itemize}
\item $\Theta(d/\eps^2)$ copies are necessary and sufficient to test whether $\rho$ is the maximally mixed state, i.e., has spectrum $(\frac1d, \dots, \frac1d)$.  This can be viewed as the quantum analogue of Paninski~\cite{Pan08}'s sharp bounds for classical uniformity-testing.
\item $\Theta(r^2/\eps)$ copies are necessary and sufficient to test with one-sided error whether $\rho$ has rank~$r$, i.e., has at most $r$ nonzero eigenvalues.  For two-sided error, a lower bound of $\Omega(r/\eps)$ copies holds.
\item $\wt{\Theta}(r^2)$ copies are necessary and sufficient to distinguish whether $\rho$ is maximally mixed on an $r$-dimensional or an $(r+1)$-dimensional subspace.  More generally, for $r$ vs.\ $r+\Delta$ (with $1 \leq \Delta \leq r$), $\wt{\Theta}(r^2/\Delta)$ copies are necessary and sufficient. \ignore{$O(r^2/\ell)$ copies are sufficient to distinguish between the cases when $\rho$'s spectrum is uniform on either $r-\ell$ or $r$ eigenvalues, and a nearly matching lower bound of $\Omega((r^2/\ell)^{1-\eps})$ copies holds for any  $\eps >0$.\rnote{John, how would you feel about: a)~calling the lower- and upper-bounds just $\wt{\Theta}(r^2/\ell)$?  Probably we can legitimately get this if we work a bit; a quick quick back-of-the-envelope suggested $r^{2 - 1/\sqrt{\log r}}$ or so when $\ell = 1$. b)~stating the $\ell = 1$ case first? At least in the main text, after pointing out to the reader that taking $\ell = r/2$ recovers Childs et al., I kind of want to tell the reader, ``okay, from now I we recommend you mentally think of $\ell$ as 1.}}
\item The EYD algorithm requires $\Omega(d^2/\eps^2)$ copies to estimate the spectrum of~$\rho$ up to $\eps$-accuracy, nearly matching the known upper bound. In addition, we simplify part of the proof of the $\widetilde{O}(d^2/\eps^2)$ upper bound.
\end{itemize}
Our techniques involve the asymptotic representation theory of the symmetric group; in particular Kerov's algebra of polynomial functions on Young diagrams.
\end{abstract}

\newpage

\setcounter{tocdepth}{2}
\tableofcontents

\newpage

\section{Introduction}

A common scenario in quantum mechanics involves an experimental apparatus which outputs a particle whose state is a random variable.
For example, in a version of the the famous Stern--Gerlach experiment by Phipps and Taylor~\cite{PT27}, the experimental apparatus produced a hydrogen atom whose electron was either in state~$\left|+\tfrac{1}{2}\right\rangle$ or~$\left|-\tfrac{1}{2}\right\rangle$, each with probability $\frac12$.
More generally, one can describe the output of such an apparatus as falling in an orthonormal set of states $|\Psi_1\rangle, \ldots, |\Psi_d\rangle \in \C^d$, distributed according to a probability distribution $\mathcal{D}=(p_1, \ldots, p_d)$.
Such an object is called a \emph{mixed state} and is often conveniently represented using the \emph{density matrix} $\rho = \sum p_i\cdot |\Psi_i\rangle\langle \Psi_i|$.  The numbers $p_1, \dots, p_d$ are called the \emph{spectrum} of~$\rho$.

Given such an apparatus, a fundamental task---known as \emph{quantum state tomography}---is to produce an estimate~$\widetilde{\rho}\in\C^{d\times d}$ which well-approximates~$\rho$ according to some distance measure (typically, the trace distance).
To do this, one repeatedly runs the apparatus to produce many (say, $n$) independent copies of~$\rho$ and then one processes some measurement of $\rho^{\otimes n}$ to produce an estimate~$\widetilde{\rho}$.
It is known~\cite[Footnote 2]{FGLE12} that $O(d^4 \log(d)/\eps^2)$ copies of~$\rho$ are sufficient to output an estimate which is $\eps$-close to~$\rho$ in the trace distance. Unfortunately, the quartic dependence on~$d$ can be prohibitively large, even for quite reasonable values of~$d$;
further exacerbating this is the fact that many quantum systems are formed as the tensor product of many smaller subsystems, in which case~$d$ is exponential in the number of subsystems.

One potential way around this problem is to note that if our actual goal in producing~$\widetilde{\rho}$ is to determine whether $\rho$ satisfies some property (e.g.,  is maximally mixed, has low rank, etc.), then our estimate~$\widetilde{\rho}$ may be giving us far more information than we need.
Thus, we can possibly test whether $\rho$ has the property in question using a much smaller number of copies.
This is the motivation behind the model of \emph{property testing of mixed states}, as promoted in the recent survey of Montanaro and de~Wolf~\cite{MdW13}.  Formally, we have following definition:
\begin{definition}\label{def:the-general-model}
A property of mixed states $\calP$ is testable with $f(d, \eps)$ copies if for every $d \geq 2, \eps >0$ there is an algorithm $\calT$ which, when given $f(d, \eps)$ copies of a mixed state $\rho \in \C^{d\times d}$, behaves as follows:
\begin{itemize}
\item If $\rho$ satisfies $\calP$, then $\Pr[\text{$\calT$ accepts}] \geq 2/3$. (``Completeness'')
\item If $\rho$ is $\eps$-far in trace distance from all $\rho'$ satisfying $\calP$, then $\Pr[\text{$\calT$ rejects}] \geq 2/3$. (``Soundness'')
\end{itemize}
\end{definition}
\noindent The choice of probability $2/3$ here is essentially arbitrary, and it can be amplified to $1-\delta$ at the expense of increasing the number of copies by a factor of~$O(\log(1/\delta))$.

As mixed states are the quantum analogue of probability distributions, this model can be seen as the quantum analogue of the model of testing properties of probability distributions.  We note that the problem of testing properties of mixed states has also appeared in the area of quantum algorithms. For example, the work of~\cite{CHW07} considers Graph Isomorphism algorithms which output a mixed state $\rho$ satisfying a certain property if and only if the input graphs are isomorphic.

In this work, we focus on the problem of testing so-called \emph{unitarily invariant} properties.  These are properties $\calP$ for which $\rho$ satisfies $\calP$ if and only if $U \rho U^\dagger$ satisfies $\calP$ for every unitary matrix~$U$.  It is easy to see that whether a mixed state~$\rho$ has  such a property depends only on $\rho$'s spectrum (hence the name \emph{quantum spectrum testing}).  Many natural properties of mixed states are unitarily invariant, such as being the maximally mixed state, having low rank, or having low von~Neumann entropy. (An example of a natural property which is \emph{not} unitarily invariant is the property of being equal to a fixed mixed state~$\sigma$, so long as $\sigma$ is not the maximally mixed state.)
Though it is not immediately apparent from the definitions (we will show this in Section~\ref{sec:testing}), the model of testing properties of mixed states from Definition~\ref{def:the-general-model} is equivalent to the following definition in the case that the property in question is unitarily invariant.

\begin{definition}\label{def:the-model}
A property of spectra $\calP$ is testable with $f(d, \eps)$ copies if for every $d \geq 2, \eps >0$ there is an algorithm $\calT$ which, when given $f(d, \eps)$ copies of a mixed state $\rho \in \C^{d\times d}$ with spectrum $\eta=(\eta_1, \ldots, \eta_d)$, behaves as follows:
\begin{itemize}
\item If $\eta$ satisfies $\calP$, then $\Pr[\text{$\calT$ accepts}] \geq 2/3$.
\item If $\eta$ is $\eps$-far in total variation distance from every $\eta'$ satisfying $\calP$, then $\Pr[\text{$\calT$ rejects}] \geq 2/3$.
\end{itemize}
\end{definition}
The main gain in using Definition~\ref{def:the-model} over Definition~\ref{def:the-general-model} is that we only have to reason about a total variation distance involving $\eta$ rather than a trace distance involving~$\rho$, which is in general a more complicated distance measure.
We note that the spectrum of a matrix is more properly thought of as an unordered multiset of eigenvalues rather than an ordered tuple, and therefore any property of spectra~$\calP$ by necessity depends only on the multiset of values $\{\eta_1, \ldots, \eta_d\}$ and not on their ordering.
Hence, quantum spectrum testing corresponds in the classical world to the model of testing \emph{symmetric} properties of probability distributions.
As we will soon see, Definition~\ref{def:the-model} allows us to show a formal correspondence between these two models.

\subsection{Classical property testing of probability distributions}

The topic of property testing was introduced by Rubinfeld and Sudan in~\cite{RS92, RS96} in the context of testing algebraic properties of polynomials over finite fields.
Since then, it has found applications in a wide variety of areas, including testing properties of graphs and of Boolean functions.
 Over the past fifteen years, an extremely successful branch of property testing, first explicitly defined in~\cite{BFR+00, BFR+13}, has focused on testing properties of discrete probability distributions.
In the model of testing properties of probability distributions, there is an unknown distribution~$\calD$ on the set~$\{1, \ldots, d\}$, and the tester may draw a \emph{random word} of length~$n$ from~$\calD^{\otimes n}$; i.e., obtain a sequence of $n$ i.i.d.\ samples from~$\calD$.  Its goal is to decide whether~$\calD$ has some property~$\calP$ or is $\eps$-far from~$\calP$ in total variation distance, while minimizing~$n$.

\ignore{
\begin{definition}\label{def:prob-dist-model}
A property of probability distributions $\calP$ is testable with $f(d, \eps)$ samples if for every $d \geq 2, \eps >0$ there is an algorithm $\calT$ which, when given $f(d, \eps)$ samples of a probability distribution $\calD$ supported on $\{1, \ldots, d\}$, behaves as follows:
\begin{itemize}
\item If $\calD \in \calP$, then $\Pr[\text{$\calT$ accepts}] \geq 2/3$.
\item If $\calD$ is $\eps$-far from $\calP$ in total variation distance, then $\Pr[\text{$\calT$ rejects}] \geq 2/3$.
\end{itemize}
\end{definition}
}

It is well known~\cite[pages 10 and 31]{DL01} (cf.~\cite[Slide 6]{Dia14})
that after taking $n = \Theta(d/\eps^2)$ samples from $\calD$, the empirical distribution is $\eps$-close to $\calD$ with high probability.  As a result, any property of probability distributions is testable with a linear (in~$d$) number of samples; thus research in this area is directed at finding algorithms of \emph{sublinear} sample complexity for various properties.  That such algorithms could exist is suggested by the following Birthday Paradox-based fact:
\begin{fact}\label{fact:b-day-p-dox}
$\Theta(\sqrt{r})$ samples are necessary and sufficient to distinguish between the cases when the distribution is uniform on either~$r$ or~$2r$ values.
(The bound also holds for~$r$ vs.\ $r'$ when $r' > 2r$.)
\end{fact}
Setting $r = \frac{d}{2}$, we see that this fact gives a sublinear algorithm for distinguishing between the uniform distribution and a distribution that is uniform on exactly half of the elements of $\{1, \ldots, d\}$.
This fact is also important as it immediately gives a lower bound of~$\Omega(\sqrt{d})$ for testing a variety of natural problems, those for which Fact~\ref{fact:b-day-p-dox} appears as a special case.

Perhaps the most basic property of probability distributions one can test for is the property of being equal to the uniform distribution, $\unif{d}$.
A $\Omega(\sqrt{d})$ lower bound follows directly from Fact~\ref{fact:b-day-p-dox}.
On the other hand, a $O(\sqrt{d}/\eps^4)$ upper bound was shown in the early work of~\cite{BFR+00, BFR+13} using techniques of~\cite{GR11}. The correct sample complexity was finally pinned down by Paninski in~\cite{Pan08}, who showed matching upper and lower bounds:
\begin{theorem}[\cite{Pan08}]\label{thm:pan-classical}
$\Theta(\sqrt{d}/\eps^2)$ samples are necessary and sufficient to test whether~$\calD$ is the uniform distribution $\unif{d}$.
\end{theorem}
This result was recently extended~\cite{VV14} to an $O(\sqrt{d}/\eps^2)$ upper bound for testing equality to \emph{any} fixed distribution, improving on the previously known~\cite{BFF+01} upper bound of $\widetilde{O}(\sqrt{d}/\eps^4)$.  More precisely, \cite{VV14} upper-bounds the sample complexity of testing equality to a fixed distribution~$\calD$ by $O(f(\calD)/\eps^2)$, where $f(\calD)$ is a certain norm which is maximized when $\calD$ is the uniform distribution.  Thus the uniform distribution is the hardest fixed distribution to test equality to.

The property of being the uniform distribution falls within the class of \emph{symmetric} properties of probability distributions.
These are the properties~$\calP$ for which $\calD=(p_1, \ldots, p_d) \in \calP$ if and only if $(p_{\pi(1)}, \ldots, p_{\pi(d)}) \in \calP$ for every permutation $\pi$.  Other interesting symmetric properties include having small entropy or small support size.  
Testing for small support size does not appear to have been precisely addressed in the literature; however the following is easy to derive from known results (in particular, the lower bound follows from the work of~\cite{VV11a}):
\begin{theorem}\label{thm:small-support-size}
To test (with $\eps$ a constant) whether a probability distribution has support size~$r$, $O(r)$ samples are sufficient and~$\Omega(r/\log(r))$ samples are necessary.
\end{theorem}
\ignore{For the latter of these, the following bounds on the sample complexity can easily be derived from the literature:
\begin{theorem}\label{thm:small-support-size}
To test whether a probability distribution has support size~$r$, XXX~samples are sufficient and~$\Omega(r/\log(r))$ samples are necessary.
\end{theorem}
The upper bound is shown in~XXX and the lower bound follows from the work of~\cite{VV11a}\jnote{Do the Valiants get a dependence on~$\eps$?  They mention one in their Corollary 10}.}

Let us now relate this section back to the main topic of this paper.
As we saw earlier, the spectrum of a mixed state can be thought of as a probability distribution on the numbers $\{1, \ldots, d\}$ (indexing the associated eigenvectors); thus any property of mixed state spectra is simply a symmetric property of probability distributions.
This correspondence allows us to directly compare the difficulty of testing properties of mixed state spectra and of probability distributions.  In fact, the quantum case is always at least as difficult as the classical case; the reason is that the classical problem is equivalent to the quantum problem under the promise that the $n$ ``samples'' provided are known orthogonal pure states, $\left| 1 \right\rangle, \dots, \left| d \right\rangle$.  Alternatively, in Sections~\ref{sec:RSK} and Section~\ref{sec:weak-schur} we will observe the following purely classical characterization of quantum spectrum testing:
\begin{fact}\label{fact:quantum-classically}
    Let $\calP$ be a symmetric property of probability distributions on $\{1, \dots, d\}$.
    Testing whether the spectrum of a $d$-dimensional quantum mixed state satisfies~$\calP$ is equivalent to the following classical testing problem:  Test whether a probability distribution~$\calD$ satisfies~$\calP$ when one is not allowed to see the whole random word $\bw \sim \calD^{\otimes n}$, but only the following $d$ statistics: the length of the longest $k$-increasing subsequence of~$\bw$, for each $1 \leq k \leq d$.  Here a \emph{$k$-increasing subsequence} means a disjoint union of~$k$ weakly increasing subsequences.
\end{fact}
In light of the above remarks we record the following fact:
\begin{fact}\label{fact:classical-vs-quantum}
Let $\calP$ be a symmetric property of probability distributions which requires $f(d, \eps)$ samples to test classically.
Then testing whether a mixed state's spectrum satisfies $\calP$ also requires at least $f(d, \eps)$ copies of the mixed state.
\end{fact}

Although quantum spectrum testing is at least as hard as testing symmetric properties of probability distributions, there are some interesting nontrivial properties which have the same complexity in both models (up to constant factors).
For example, if $\calP$ is the property of having support size~$1$, then $\Theta(1/\eps)$ samples/copies are necessary and sufficient to test $\calP$ in both models (see~\cite{MdW13} for the $O(1/\eps)$ quantum spectrum testing upper bound).
In general, however, it is known that spectrum testing can require an asymptotically higher complexity (at least in terms of the parameter~$d$).

We end this section by pointing out that a large portion of the property testing literature concerning entropy and support size actually considers the problems of either computing these values~\cite{Pan04,BDKR05, VV11a, VV11b} (within some tolerance) or distinguishing between the cases when these values are either large or small~\cite{Val08} (often these problems have some added guarantee on the probability distribution, such as all of its nonzero probabilities being sufficiently large).
These problems, strictly speaking, do not fit within the above property testing framework.
In this work, when we consider the problem of testing a mixed state's rank (the quantum analogue of support size) we will be doing so explicitly within the property testing framework.

\subsection{Related work}

Returning to quantum spectrum testing, we would like to mention two prior lines of research that are directly relevant.  The first is an algorithm---which we call the \emph{empirical Young diagram} (EYD) algorithm---for learning the spectrum of an unknown mixed state.  This algorithm is naturally suggested by the early work of Alicki, Rudnicki, and Sadowski~\cite{ARS88} and was explicitly proposed by Keyl and Werner~\cite{KW01}. Regarding its performance guarantee, Hayashi and Matsumoto~\cite{HM02} gave explicit error bounds and a short proof, but their work contained some small calculational errors, subsequently corrected by Christandl and Mitchison~\cite{CM06}.  From the last of these it is easy to deduce the following:
\begin{theorem}\label{thm:kw}
The empirical Young diagram algorithm, when given $O(d^2/\eps^2 \cdot \ln(d/\eps))$ copies of a mixed state~$\rho$ with spectrum~$\eta$, outputs with high probability an estimate of~$\eta$ that is $\eps$-close in total variation distance.
\end{theorem}
We will give a description of this algorithm later in the paper; for now, suffice it to say that it can be viewed as the  quantum version of the natural classical algorithm for learning an unknown distribution, viz., outputting the empirical distribution.  The EYD algorithm gives a near-quadratic improvement over known quantum state tomography algorithms for the problem of estimating a mixed state's spectrum.\footnote{One may note that the dependence on~$\eps$ in Theorem~\ref{thm:kw} is slightly \emph{worse} than that for full tomography; however, we speculate that this is an artifact of the analysis and that $O(d^2/\eps^2)$ copies suffice for the EYD algorithm.} As a result, testing properties of quantum spectra is easy with a quadratic number of copies, and so we hope for \emph{subquadratic} algorithms.

The second result comes from the work of Childs et al.~\cite{CHW07}. It can be thought of as a quantum analogue of Fact~\ref{fact:b-day-p-dox}:
\begin{theorem}\label{thm:q-bday}
$\Theta(r)$ copies of a state $\rho$ are necessary and sufficient to distinguish between the cases when
$\rho$'s spectrum is uniform on either~$r$ or~$2r$ values.
(The bound also holds for~$r$ vs.\ $cr$ when $c > 2$ is an integer.)
\end{theorem}
\noindent Setting $r = \frac{d}{2}$, Theorem~\ref{thm:q-bday} gives a  linear lower bound of $\Omega(d)$ for various properties of spectra.
This is in contrast with property testing of probability distributions, in which sublinear algorithms are the main goal, with the  Birthday Paradox typically precluding sub-$O(\sqrt{d})$-sample algorithms.

Finally, we mention that we may also obtain relevant results by applying Fact~\ref{fact:classical-vs-quantum} to known lower bounds for classical property testing of probability distributions. Though in general these lower bounds are not tight, prior to our work this was (to our knowledge) the only way to produce lower bounds for testing spectra with a dependence on~$\eps$.

\subsection{Our results}

We have four main results.
The first concerns the property that Montanaro and de Wolf refer to as \textbf{Mixedness}:
\begin{theorem}\label{thm:paninski-intro}
$\Theta(d/\eps^2)$ copies are necessary and sufficient to test whether $\rho \in \C^{d \times d}$ is the maximally mixed state; i.e., whether its spectrum is $\eta = (1/d, \dots, 1/d)$.
\end{theorem}
\noindent
This is the quantum analogue of Paninski's Theorem~\ref{thm:pan-classical}.  We also remark that given the way we prove Theorem~\ref{thm:paninski-intro}, Childs et al.'s Theorem~\ref{thm:q-bday} can be obtained as a very special case.

Our second result gives new bounds for testing whether a state has low rank.
\begin{theorem}\label{thm:rank-intro}
$\Theta(r^2/\eps)$ copies are necessary and sufficient to test whether $\rho \in \C^{d \times d}$ has rank~$r$ with one-sided error.
With two-sided error, a lower bound of $\Omega(r/\eps)$ holds.
\end{theorem}
\noindent
We note that the copy complexity is independent of the ambient dimension~$d$.
Knowing that a state is low rank can often make solving a given problem much simpler.
For example, quantum state tomography can be made more efficient when the state is known to be low-rank~\cite{FGLE12}.
Compare this to Theorem~\ref{thm:small-support-size}.

Next, we extend Childs et al.'s Theorem~\ref{thm:q-bday} to $r$ vs.\ $r'$ for \emph{any} $r+1 \leq r' \leq 2r$.  A qualitative difference is seen when $r' = r+1$; namely, nearly quadratically many copies are necessary.
\begin{theorem}\label{thm:unif-dist}
Let $1 \leq \Delta \leq r$.  Then $O(r^2/\Delta)$ copies are sufficient to distinguish between the cases
when $\rho$'s spectrum is uniform on either $r$
or $r+\Delta$ eigenvalues; further, a nearly matching lower bound of $\wt{\Omega}(r^2/\Delta)$ copies holds.
\end{theorem}
\noindent
As above, we note that these bounds are independent of the ambient dimension~$d$.

Our final results concern the EYD algorithm from Theorem~\ref{thm:kw}.
First, we give an arguably simpler proof of Theorem~\ref{thm:kw}.
Next, we complement this with a lower bound showing that the analysis of the EYD algorithm from Theorem~\ref{thm:kw} is tight up to logarithmic factors.
\begin{theorem}
If $\rho \in \C^{d \times d}$ is the maximally mixed state, the algorithm from Theorem~\ref{thm:kw} \emph{fails} to give an $\eps$-accurate estimate (with high probability) unless $\Omega(d^2/\eps^2)$ copies are used.
\end{theorem}
\noindent
To our knowledge, no such lower bound was known previously.  We remark that it is an interesting open question whether some \emph{other} algorithm can estimate an unknown state's spectrum from a subquadratic number of copies.

\subsection{Overview of our techniques}

Following~\cite{ARS88,Har05,CM06,CHW07}, we use techniques from representation theory of the symmetric group~$\symm{n}$. A basic tool is \emph{Schur--Weyl duality}, which decomposes the space $(\C^d)^{\otimes n}$ as
\begin{equation}
(\C^d)^{\otimes n} \stackrel{\symm{n} \times \unitary{d}}{\cong} \bigoplus_{\lambda \vdash n} \srep_\lambda \otimes \urep_\lambda^d,
\label{eq:schur-weyl}
\end{equation}
where the subspace $\srep_\lambda$ corresponds to the symmetric group, the subspace $\urep_\lambda^d$ corresponds to the unitary group, and $\lambda$ is a partition of~$n$, thought of as a Young diagram. (Recall that a partition of~$n$ is a tuple $\lambda = (\lambda_1, \ldots, \lambda_\ell)$ satisfying $\lambda_1 \geq \ldots \geq \lambda_\ell \geq 0$ and $\lambda_1 + \ldots + \lambda_\ell = n$.) In our testing problem, the tester is provided with $\rho^{\otimes n}$, which is invariant under any permutation of the~$n$ coordinates, and whether the tester accepts or rejects should be invariant under any unitary transformation of~$\rho$.
This means that if we measure $\rho^{\otimes n}$ in the \emph{Schur basis} described in equation~\eqref{eq:schur-weyl} below, we can throw away the information from the permutation and unitary registers without losing any relevant information.  What is left is only the ``irrep'' label~$\lambda$.

The end result is this: there is a sampling algorithm---referred to in~\cite{CHW07} as \emph{weak Schur sampling}---which, on input a mixed state $\rho^{\otimes n}$, outputs a random partition~$\blambda$ whose distribution depends only on the spectrum of $\rho$.  We will denote this distribution by $\SWdens{n}{\rho}$.  Furthermore, an argument which is essentially from~\cite{CHW07} (though see~\cite[Lemma~$19$]{MdW13} for a full statement)
shows that for any spectrum property~$\calP$, there is an \emph{optimal} tester in the model of Definition~\ref{def:the-model} whose operation is as follows: 1.~Sample $\blambda\sim\SWdens{n}{\rho}$. 2.~Accept or reject based only on~$\blambda$.
We may therefore proceed without loss of generality by analyzing only algorithms of this form.  In particular, this means we need not study study quantum measurements or algorithms per se; in principle it suffices simply to understand the distribution $\SWdens{n}{\rho}$ (which is equivalent to the distribution on $k$-increasing subsequence lengths described in Fact~\ref{fact:quantum-classically}).

In case $\rho$ is the maximally mixed state, the distribution $\SWdens{n}{\rho}$ has been fairly well studied, starting with the works~\cite{TW01,Joh01,Bia01,Kup02} (see~\cite{Mel10a} for a recent, comprehensive treatment). It is known as the \emph{Schur--Weyl} distribution, and we denote it by $\SWunif{n}{d}$.  (In the limit as $d \rightarrow \infty$, it approaches the well-known \emph{Plancherel} distribution.)
The exact distribution on partitions given by $\SWunif{n}{d}$ is somewhat complicated and difficult to work with, and so various works have instead sought to describe large-scale features of a ``typical'' $\blambda \sim \SWunif{n}{d}$.
For example, Biane~\cite{Bia01} showed that, up to small fluctuations,
the ``shape'' of the random Young diagram $\blambda \sim \SWunif{n}{d}$ tends toward a certain limiting shape $\Omega$
which depends only on the ratio~$\frac{\sqrt{n}}{d}$.
Furthermore, Meliot~\cite{Mel10a} has characterized these small fluctuations as being distributed according to a certain Gaussian process.
The second of these results borrows heavily from a proof of the analogous result by Kerov (see~\cite{IO02}) for the Plancherel distribution, and we will give an overview his techniques below.

Kerov's approach involves studying a certain space of symmetric polynomial functions on Young diagrams.
For example, if one is interested in showing that a random $\blambda \sim \SWunif{n}{d}$ tends to have some coordinates which are much larger than the rest, then it would be natural to study ``moments'' of the form $\sum_i \blambda_i^k$.  However, the approach of Kerov would suggest studying the following ``moments'' instead:
\begin{equation*}
\pstar{k}(\blambda) \coloneqq \sum_{i=1}^\infty [(\blambda_i - i +\tfrac{1}{2})^k - (-i + \tfrac{1}{2})^k],\text{ for $k \geq 1$}.
\end{equation*}
The polynomial family $(\pstar{k})$ inhabits (in fact, generates) the so-called \emph{algebra of polynomial functions on the set of Young diagrams} $\Lambda^*$ (also known as Kerov's \emph{algebra of observables}).
There are other important polynomial families within $\Lambda^*$---in addition to the $\pstar{k}$ polynomials, our work involves the $\widetilde{p}_k$, $c_k$,
 $\psharp{\mu}$, and $s^*_\mu$ polynomials---and each of these families sheds light on a different aspect of the input partition~$\lambda$.
For example, though the $\psharp{\mu}(\lambda)$ polynomials don't give any obvious information regarding the ``shape'' of $\lambda$, they are unique in that we can easily compute the expectation $\E_{\blambda \sim \SWdens{n}{\rho}} [\psharp{\mu}(\blambda)]$ for any mixed state $\rho$.
There exist some methods for passing from one polynomial family to another, and it is often the case that a problem most easily stated in terms of one polynomial family is most easily solved in terms of another.

The main component of our work is lower bounds for quantum spectrum testing, and these lower bounds generally have the following outline: 1.~Reduce the problem to showing that a certain expression within the algebra of observables is small with high probability. 2.~Use various polynomial-estimation techniques developed by Kerov and others for proving concentration of said expression.  For example, roughly speaking the key component of the lower bound in Theorem~\ref{thm:unif-dist} is showing that for $n \ll r^2$, the expression
\[
    \sum_{k=2}^\infty\frac{(-1)^k \pstar{k}(\blambda)}{k(\r+\tfrac12)^k}
\]
is typically very close to~$0$ when $\blambda \sim \SWunif{n}{r}$.  As another example, proving the lower bound in Theorem~\ref{thm:paninski-intro} reduces to showing that when $n \ll d/\eps^2$, the expression
\[
    \sum_{\substack{\text{partitions $\mu$ of~$n$} \\ \text{with at most $d$ nonzeros}}} \frac{s^*_\mu(\blambda) s_\mu(+2\eps, -2\eps, \ldots, +2\eps, -2\eps)}{\prod_{i=1}^d \prod_{j=1}^{\mu_i} (d+(j-i))}
\]
is typically very close to~$1$ when $\blambda \sim \SWunif{n}{d}$. Our upper bounds generally involve analyzing algorithms which accept or reject based on simple statistics of the sampled $\blambda \sim \SWunif{n}{d}$.
For example, the rank tester of Theorem~\ref{thm:rank-intro} accepts if and only if the sampled $\blambda$ has at most~$r$ nonzero parts, and the uniformity tester of Theorem~\ref{thm:paninski-intro} accepts if and only if the ``content polynomial'' $c_1(\blambda)$ is sufficiently small.
As in the lower bounds, analyzing these algorithms uses techniques from the algebra of observables, and we sometimes also require certain combinatorial interpretations of the weak Schur sampling algorithm; e.g., its relationship with the Robinson--Schensted--Knuth ``bumping'' algorithm.

\subsection{Acknowledgments}
We thank Ilias Diakonikolas, Rocco Servedio, Greg Valiant, and Paul Valiant for helpful discussions regarding classical testing and learning of probability distributions.  We thank Ashley Montanaro for helpful discussions regarding quantum property testing and for suggesting the proof of Proposition~\ref{prop:eq-model}.

\section{Preliminaries}

\subsection{Probabilistic distances}
Given two discrete probability distributions $\calD_1$ and $\calD_2$ on a finite set~$\Omega$, the \emph{total variation distance} between them is
\begin{equation*}
\dtv{\calD_1}{\calD_2} \coloneqq \frac{1}{2}\cdot \sum_{\omega \in \Omega} \left|\calD_1(\omega)-\calD_2(\omega)\right|.
\end{equation*}
We will also require some nonsymmetric ``distances'' between probability distributions.  The \emph{chi-squared distance} is
\begin{equation*}
\dchi{\calD_1}{\calD_2} \coloneqq \E_{\bomega \sim \calD_2}\left[\left(\frac{\calD_1(\bomega)}{\calD_2(\bomega)}-1\right)^2\right].
\end{equation*}
Further, if $\supp(\calD_1)\subseteq \supp(\calD_2)$, then the \emph{Kullback--Leibler divergence} is
\begin{equation*}
\dkl{\calD_1}{\calD_2} \coloneqq \E_{\bomega \sim \calD_1}\left[\ln\left(\frac{\calD_1(\bomega)}{\calD_2(\bomega)}\right)\right].
\end{equation*}
To relate these quantities, Cauchy--Schwarz implies that $\dtv{\calD_1}{\calD_2} \leq \frac{1}{2}\sqrt{\dchi{\calD_1}{\calD_2}}$, and Pinsker's inequality states that $\dtv{\calD_1}{\calD_2} \leq \frac{1}{\sqrt{2}}\sqrt{\dkl{\calD_1}{\calD_2}}$.

We would also like to introduce a ``permutation-invariant'' notion of total variation distance.  Suppose that the set~$\Omega$ is naturally ordered; say, $\Omega = [d] \coloneqq \{1, 2, \dots, d\}$.  We define
\[
    \dtvsymm{\calD_1}{\calD_2} \coloneqq \dtv{\calD_1^\downarrow}{\calD_2^\downarrow} = \min_{\pi \in \symm{d}}\{\dtv{\calD_1}{\calD_2 \circ \pi}\}.
\]
Here $\calD_i^\downarrow$ denotes the probability distribution on~$[d]$ given by rearranging $\calD_i$'s probabilities in nonincreasing order, so $\calD_i^\downarrow(1) \geq \cdots \geq \calD_i^\downarrow(d)$.  By virtue of the permutation-invariance, we may also naturally extend this notation to the case when $\calD_1$ and $\calD_2$ are simply unordered multisets of nonnegative numbers summing to~$1$.

A $d$-dimensional mixed quantum state is represented as a \emph{density matrix} $\rho \in \C^{d \times d}$; i.e., a positive semidefinite matrix with trace~$1$.
We may write $\rho$ using its spectral decomposition as
\begin{equation*}
\rho = \sum_{i = 1}^d \eta_i \cdot\vert \Psi_i \rangle \langle \Psi_i |,
\end{equation*}
where the $|\Psi_i\rangle$'s are orthornormal  and the $\eta_i$'s are nonnegative reals satisfying $\eta_1 + \dots + \eta_d = 1$.
Equivalently, $\rho$ describes a probability distribution on pure states in which $\left|\Psi_i\right\rangle$ has probability~$\eta_i$.  If $\sigma$ is another $d$-dimensional mixed state with eigenvalues $\{\lambda_1, \dots, \lambda_d\}$ (thought of as a multiset), we will use the notation
\[
    \dtvsymm{\rho}{\sigma} \coloneqq \dtvsymm{\{\eta_1, \dots, \eta_d\}}{\{\lambda_1, \dots, \lambda_d\}}.
\]

We will now define \emph{trace distance}, which is the standard notion of distance between two density matrices.  (The above nonstandard notion of distance will be related to the trace distance in Proposition~\ref{prop:eq-model} below.)  If $M \in \C^{d \times d}$ is any Hermitian matrix with eigenvalues $\eta_1, \ldots, \eta_d$, the \emph{trace norm} of $M$ is
\begin{equation*}
\Vert M \Vert_{\mathrm{tr}}\coloneqq  \trace\mleft(\sqrt{M^{\dagger}M}\mright) = \sum_{i=1}^d |\eta_i|.
\end{equation*}
Given two density matrices~$\rho$ and~$\sigma$, the \emph{trace distance} between them is
$\dtr{\rho}{\sigma} \coloneqq  \frac{1}{2} \Vert \rho - \sigma \Vert_{\mathrm{tr}}$. The trace distance is the standard generalization of the total variation distance to mixed states; for example, it represents the maximum probability with which two mixed states can be distinguished by a measurement~\cite[equation~(9.22)]{NC10}.  This property makes it the natural choice of distance for property testing of quantum states.  We also have the following simple fact:
\begin{fact}\label{fact:trace-tv}
Suppose $\rho$ and $\sigma$ are diagonal density matrices with diagonal entries $\eta = (\eta_1, \ldots, \eta_d)$ and $\lambda = (\lambda_1, \ldots, \lambda_d)$, respectively.
Then $\dtr{\rho}{\sigma} = \dtv{\eta}{\lambda}$.
\end{fact}

\subsection{Property testing}\label{sec:testing}

In the model of property testing, there is a set of objects $\calO$ along with a distance measure $\dist :\calO\times\calO \rightarrow \R$.
A property $\calP$ is a subset of $\calO$, and for an object $o \in \calO$, we define the distance of $o$~to~$\calP$ to be\footnote{Formally, our sets $\calO$ will always lie within some $\R^N$ or $\C^N$, and we always require that $\calP$ be a \emph{closed} set.  Thus the ``$\min$'' here is well-defined.}
\begin{equation*}
\dist(o, \calP) \coloneqq  \min_{o' \in \calP} \{\dist(o, o')\}.
\end{equation*}
If $\dist(o, \calP) \geq \eps$, then we say that~$o$ is $\eps$-far from $\calP$.
A testing algorithm $\calT$ tests~$\calP$ if, given some sort of ``access'' to $o \in \calO$ (e.g., independent samples or queries), $\calT$ accepts if $o \in \calP$ and rejects if~$o$ is $\eps$-far from $\calP$.
Generally, the aim is for $\calT$ to be efficient according some measure, most typically the number of accesses made to~$o$.
(On the other hand, $\calT$ is generally allowed unlimited computational power.  Nevertheless, as we will see, all of the testers considered in this paper can be implemented efficiently.)

We will instantiate property testing in the following natural settings:
\begin{enumerate}[label=(\roman*)]
\item \textbf{Properties of mixed states:} $\calO$ is the set of $d$-dimensional mixed states~$\rho$, the tester gets access to (unentangled) copies of~$\rho$, and $\dist = \dtrbare$.\label{item:general-model}
\item \textbf{Unitarily invariant properties of mixed states:} As above, but $\calP$ must be unitarily invariant; equivalently, whether or not $\rho \in \calP$ only depends on the multiset of $\rho$'s eigenvalues.\label{item:unitarily}
\item \textbf{Quantum spectrum testing:} $\calO$ is the set of $d$-dimensional mixed states, $\calP$ must be unitarily invariant, and $\dist(\rho, \sigma) = \dtvsymm{\rho}{\sigma}$.\label{item:spectrum}
\item \textbf{Symmetric properties of probability distributions:} $\calO$ is the set of probability distributions $\calD$ on~$[d]$, the tester gets i.i.d.\ draws from~$\calD$, $\calP$~is any symmetric property, and $\dist = \dtvbare$.\label{item:symm}
\end{enumerate}

Let us now establish some basic facts about these models.  The simplest fact is that Model~\ref{item:unitarily} is a special case of Model~\ref{item:general-model}.  Next, in Model~\ref{item:symm} it would be equivalent if we had chosen $\dist = \dtvsymmbare$; this is by virtue of the assumption that $\calP$ is a symmetric (permutation-invariant) property of distributions on~$[d]$.  Finally, we have the following important simplifying fact, whose proof is not trivial:
\begin{proposition}\label{prop:eq-model}
Models~\ref{item:unitarily} and~\ref{item:spectrum} are equivalent.
\end{proposition}
\begin{proof}
We need to show that if $\calP$ is a unitarily invariant property of $d$-dimensional mixed states then $\dtr{\rho}{\calP} = \dtvsymm{\rho}{\calP}$ holds for all mixed states~$\rho$.  By performing a unitary transformation, we may assume without loss of generality that~$\rho$ is a diagonal matrix with nonincreasing diagonal entries (spectrum).

    The easy direction of the proof is showing that $\dtr{\rho}{\calP} \leq \dtvsymm{\rho}{\calP}$.  To see this, suppose $\sigma \in \calP$ achieves $\dtvsymm{\rho}{\sigma} = \eps$.  Let $\sigma'$ denote the diagonal density matrix whose diagonal entries are the eigenvalues of~$\sigma$ arranged in nonincreasing order.  Now $\sigma'$ is unitarily equivalent to~$\sigma$, and hence $\sigma' \in \calP$ as well.  But $\dtr{\rho}{\sigma'} = \eps$ by Fact~\ref{fact:trace-tv} and we therefore conclude $\dtr{\rho}{\calP} \leq \eps$, as needed.

    The more interesting direction is showing  that $\dtvsymm{\rho}{\calP} \leq \dtr{\rho}{\calP}$.  The authors learned the proof of this fact from Ashley Montanaro~\cite{AM14}.  Suppose that $\sigma \in \calP$ achieves $\dtr{\rho}{\sigma} = \eps$.  Since $\| \cdot \|_{\mathrm{tr}}$ is a unitarily invariant norm, a theorem of Mirsky (see~\cite[Corollary~7.4.9.3]{HJ13}) states that
    \begin{equation}    \label{eqn:mirsky}
        \|\rho - \sigma\|_{\mathrm{tr}} \geq \|\rho' - \sigma'\|_{\mathrm{tr}},
    \end{equation}
    where $\sigma'$ (respectively, $\rho'$) denotes the diagonal density matrix whose entries are the eigenvalues of $\sigma$ (respectively, $\rho$) arranged in nonincreasing order.  We have $\rho' = \rho$, and $\sigma'$ is again unitarily equivalent to~$\sigma$, implying $\sigma' \in \calP$.  But the left-hand side of~\eqref{eqn:mirsky} is~$2\eps$, and the right-hand side is $2\dtv{\rho'}{\sigma'}$ (by Fact~\ref{fact:trace-tv}), which in turn equals $2\dtvsymm{\rho}{\sigma'}$.   Thus $\dtvsymm{\rho}{\calP} \leq \eps$, as needed. \ignore{

The authors learned this proof from Ashley Montanaro~\cite{AM14}.
Let $\calP$ be any unitarily invariant property of mixed states.
Given two states $\rho$ and $\sigma$ with spectra $\eta = (\eta_1, \ldots, \eta_d)$ and $\lambda = (\lambda_1, \ldots, \lambda_d)$, respectively, write
$
\dtv{\rho}{\sigma}\coloneqq  \dtv{\eta}{\lambda}.
$
Our goal now is to show that for any mixed state $\rho$, $\dtr{\rho}{\calP} = \dtv{\rho}{\calP}$, from which the proposition will follow as models~\ref{item:unitarily} and~\ref{item:spectrum} differ only in the distance measure used.  Write $\eta= (\eta_1, \ldots, \eta_d)$ for $\rho$'s spectrum.

By performing a unitary change of basis, we may assume that~$\rho$ is a diagonal matrix with decreasing diagonal elements (i.e.\ $\rho_{1, 1} \geq \ldots\geq \rho_{d, d}$).
To show that $\dtr{\rho}{\calP} \leq \dtv{\rho}{\calP}$, let $\sigma \in \calP$ satisfy $\dtv{\rho}{\sigma} = \dtv{\rho}{\calP}$\jnote{Do we have to worry about existence issues here? Such a $\sigma$ should always exist if $\calP$ is closed, but if $\calP$ is open it might not.  On the other hand, this should still work if we take a sequence of mixed states $\sigma_1, \sigma_2, \ldots$ where $\dtv{\rho}{\sigma_i}$ approaches $\dtv{\rho}{\calP}$ in the limit, but is this worth going into?}.  Because $\calP$ is unitarily invariant and $d_{\mathrm{TV}}$ only depends on the spectra of the states, we may assume that $\sigma$ is a diagonal matrix.  Write $\lambda=(\lambda_1, \ldots, \lambda_d)$ for $\sigma$'s spectrum.  Then because $\rho$ and $\sigma$ are simultaneously diagonalizable, $\dtr{\rho}{\sigma} = \dtv{\eta}{\lambda}$ by Fact~\ref{fact:trace-tv}.  This concludes this case, as $\dtr{\rho}{\calP} \leq \dtr{\rho}{\sigma}$.

The more interesting case is showing that $\dtr{\rho}{\calP} \geq \dtv{\rho}{\calP}$.  Write $\sigma \in \calP$ for the state that satisfies $\dtr{\rho}{\sigma} = \dtr{\rho}{\calP}$, and let $\lambda = (\lambda_1, \ldots, \lambda_d)$ be $\sigma$'s spectrum, arranged in decreasing order.  By definition, $\dtr{\rho}{\sigma} = \frac{1}{2} \Vert \rho - \sigma \Vert_{\mathrm{tr}}$. Since $\Vert \cdot \Vert_{\mathrm{tr}}$ is a unitarily invariant norm, a result of Mirsky (proved in \cite[Corollary 7.4.9.3]{HJ13}) says that
\begin{equation*}
\Vert \rho - \sigma \Vert_{\mathrm{tr}} \geq \Vert \mathrm{diag}^{\downarrow}(\rho) - \mathrm{diag}^{\downarrow}(\sigma) \Vert_{\mathrm{tr}},
\end{equation*}
where for a matrix $A$,  $\mathrm{diag}^{\downarrow}(A)$ is the diagonal matrix containing $A$'s eigenvalues arranged in descending order.
By assumption, $\rho = \mathrm{diag}^{\downarrow}(\rho)$, and because $\calP$ is unitarily invariant, $\mathrm{diag}^{\downarrow}(\sigma) \in \calP$.
Thus, $\dtr{\rho}{\sigma} \geq \dtr{\rho}{\mathrm{diag}^{\downarrow}(\sigma)} = \dtv{\rho}{\sigma}$, where the last equality follows from Fact~\ref{fact:trace-tv}.
As $\dtv{\rho}{\calP} \leq \dtv{\rho}{\sigma}$, we are done.}
\end{proof}

Finally, we remind the reader of Fact~\ref{fact:classical-vs-quantum}, which says that any quantum spectrum testing problem (in either of the equivalent Models~\ref{item:unitarily} and~\ref{item:spectrum}) is at least as hard as the corresponding classical problem in Model~\ref{item:symm}.\ignore{
\rnote{The following is basically identical to Fact~\ref{fact:classical-vs-quantum}, so I'm not sure why we're repeating it here. Is it not enough to point back to Fact~\ref{fact:classical-vs-quantum}?}
\begin{proposition}
Let $\calP$ be any symmetric property of probability distributions, and let $\calP_Q$ be the set of mixed states whose spectra fall in~$\calP$.
If $f(d, \eps)$ is a lower bound on the number of samples needed to test~$\calP$ in Model~\ref{item:symm}, then $f(d, \eps)$ is also a lower bound on the number of copies needed to test $\calP_Q$ in models~\ref{item:unitarily} and~\ref{item:spectrum}.
\end{proposition}}

\subsection{Partitions and Young diagrams}\label{sec:partitions}

A \emph{partition} of $n \geq 1$, denoted $\lambda\vdash n$, is a list of nonnegative integers $\lambda = (\lambda_1, \lambda_2, \ldots, \lambda_k)$ satisfying $\lambda_1 \geq \lambda_2 \geq \ldots \geq \lambda_k$ and $\lambda_1 + \lambda_2 + \ldots + \lambda_k = n$. The \emph{length} of the partition, denoted $\ell(\lambda)$, is the number of nonzero $\lambda_i$'s in $\lambda$. The partition's \emph{size} is~$n$, and is also written as $|\lambda|$. Two partitions are considered to be equivalent if they only differ in trailing zeros. For example, $(4,2)$ and $(4,2,0,0)$ are equivalent.  We write $\allpartitions$ to denote the set of all partitions, of any size.  For $w \in \N^+$ we will use the notation $m_w(\lambda)$ to denote the number of parts~$i$ with $\lambda_i = w$.  Finally, at one point we will require the fairly elementary fact (see e.g.~\cite[(1.15)]{Rom14}) that the number of partitions of~$n$ is~$2^{O(\sqrt{n})}$ (much more precise asymptotics are known~\cite{HR18}).

One way in which partitions arise is as \emph{cycle types} of permutations $\pi \in \symm{n}$.  We say that $\pi$ has cycle type $\lambda = (\lambda_1, \dots, \lambda_k) \vdash n$ if $\pi$ is the product of disjoint cycles of size $\lambda_1, \lambda_2, \ldots, \lambda_k$.  (Note that $\pi$'s length-$1$ cycles are included.)  The standard notation for this is $\cyctype{\pi} = \lambda$. However we will use this notation extremely sparingly (and with warning) so as to preserve the symbol ``$\rho$'' for density matrices.  In aid of this, we adopt the following convention: \emph{whenever a permutation $\pi$ appears in a place where a partition $\lambda$ is expected, the meaning is that $\lambda$ should be the cycle type of~$\pi$}. We also use the following standard notation:
\begin{equation*}
    z_\lambda \coloneqq  \prod_{w \geq 1} (w^{m_w(\lambda)}\cdot m_w(\lambda)!).
\end{equation*}
When $\lambda \vdash n$, the quantity $n!/z_\lambda$ is the number of permutations in $\symm{n}$ of cycle type~$\lambda$,
so $z_\lambda^{-1}$ represents the probability that a uniformly random permutation in $\symm{n}$ has cycle type~$\lambda$.\\

It is standard to represent a partition $\lambda \vdash n$ pictorially with a \emph{Young diagram}; i.e., a certain arrangement of~$n$ squares, called \emph{cells} or \emph{boxes}. There are several conventions for how to draw Young diagrams:  we will define the \emph{French notation}, the \emph{Russian notation}, and the \emph{Maya notation}. \footnote{We will not require the \emph{English notation}, which is the reflection of the French notation across the horizontal axis.}

\begin{figure}
\centering
\begin{subfigure}[t]{.4\textwidth}
	\centering
	\includegraphics[height=100pt]{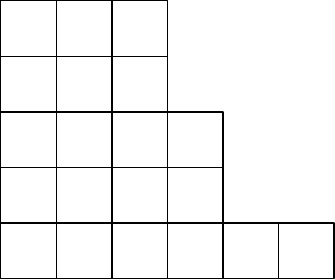}
	\caption{French notation.}
	\label{fig:french-notation}
\end{subfigure}%
\begin{subfigure}[t]{.6\textwidth}
	\centering
	\includegraphics[height=150pt]{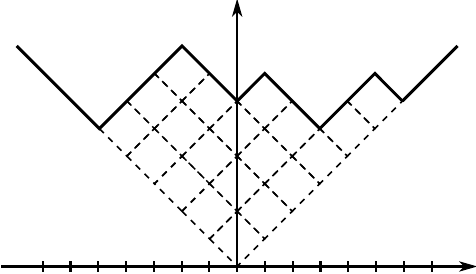}
	\caption{Russian notation (in dashed lines).  The marks on the horizontal axis are integral $x$-values,
			and the heavy black line is the curve $\lambda(x)$.}
	\label{fig:russian-notation}
\end{subfigure}
\caption{Two ways of drawing the partition $\lambda = (6,4,4,3,3)$.}
\end{figure}

In the \emph{French notation}, the Young diagram for $\lambda = (\lambda_1, \ldots, \lambda_k)$ is drawn with left-justified rows of cells: $\lambda_1$ cells in the bottom row, $\lambda_2$ cells on top of this, $\lambda_3$ cells on top of this, etc. As an example, the French notation for $(6,4,4,3,3)$ is pictured in Figure~\ref{fig:french-notation}.  We think of the French diagram as consisting of unit squares sitting in~$\R_+^2$, with bottom-left corner at the origin.

Given the French diagram, it's natural to define the \emph{width} of $\lambda$ as $\lambda_1$, and to refer to $\ell(\lambda)$ as its \emph{height}.  We can also define the \emph{conjugate partition} of $\lambda$ to be the partition $\lambda' \vdash n$ obtained by reflecting the French diagram through the line $y = x$; i.e., exchanging rows and columns.  For example, the conjugate of $\lambda = (6,4,4,3,3)$ is $\lambda' = (5,5,5,3,1,1)$. Note that the height of~$\lambda$ is the width of $\lambda'$, and vice versa; in particular, we sometimes prefer the notation $\lambda'_1$ to  $\ell(\lambda)$.

We now define the \emph{Russian notation} for $\lambda$.  This is obtained from the French notation by first rotating the diagram $45^\circ$ counterclockwise about the origin, and then dilating by a factor of~$\sqrt{2}$; see Figure~\ref{fig:russian-notation}.  The purpose of the dilation is so that the corners of the boxes will have integer $x$- and $y$-coordinates.  The purpose of the rotation is so that conjugation corresponds to reflection in the $y$-axis and so that the boundary of the diagram forms the graph of a function:
\begin{definition}                              \label{def:russian-function}
Given a partition~$\lambda$ drawn in Russian notation, its upper boundary forms the graph of a function with domain $[-\lambda'_1, \lambda_1] \subseteq \R$.  We extend this function to have domain all of $\R$ according to the function $x \mapsto |x|$.  We will use the notation $\lambda : \R \to \R_+$ for this function, which we remark is a continuous and piecewise linear curve. Any time we write $\lambda(x)$, where $\lambda$ is a partition and $x \in \R$, we are referring to this curve. See Figure~\ref{fig:russian-notation} for an example.
\end{definition}

Finally, we define the \emph{Maya notation}. It contains no boxes; just a sequence of black and white pebbles.  However the Maya notation is typically drawn in conjunction with the Russian notation, with the pebbles being located on the half-integer points $\Z + \frac12$ of the $x$-axis.  In the Maya notation, a black pebble is placed at all points directly below a ``downward-sloping'' segment in $\lambda$'s graph, and a white pebble is placed at all points directly below an ``upward-sloping'' segment.  (Thus all sufficiently negative half-integer points have a black pebble and all sufficiently positive half-integer points have a white pebble.)  The notation also includes a vertical tick mark to denote the location of the origin.  A picture of the Russian and Maya notation for $\lambda = (6,4,4,3,3)$ appears later in Figure~\ref{fig:partition-diagram-example} (the reader consulting it now should ignore the red and green coloring, the dashed lines, and the box labeled ``$d$'').  One can check that the sequence of pebbles uniquely identifies the partition~$\lambda$.  It also uniquely determines the position of the origin mark, in that the number of black pebbles to the right of the origin mark always equals the number of white pebbles to the left of the origin mark.  These numbers are both equal to $d(\lambda)$, defined to be the number of cells touching the $y$-axis in the Russian diagram.  We make one more definition:
\begin{definition}                                  \label{def:mod-frob-coords}
    Given the Maya diagram of a partition $\lambda$, we may define its \emph{modified Frobenius coordinates} to be the half-integer values $a^*_1 > a^*_2 > \cdots a^*_d > 0$ and $b^*_1 > b^*_2 > \cdots > b^*_d > 0$ (for $d = d(\lambda)$), where $a^*_i$ is the position of the $i$th rightmost black pebble and $b^*_i$ is the negative of the position of the $i$th leftmost white pebble.  One may check that, equivalently, $a^*_i = \lambda_i - i + \frac12$ and $b^*_i = \lambda'_i - i + \frac12$.  For example, if $\lambda = (6,4,4,3,3)$, then $a^* = (\frac{11}{2}, \frac{5}{2}, \frac{3}{2})$ and $b^* = (\frac92, \frac72, \frac52)$.  The coordinates have the property that $\sum_i (a^*_i + b^*_i) = |\lambda|$.
\end{definition}

For a partition $\lambda$ (drawn either in the French or Russian notation), we often use the symbol ``$\square$'' to denote a box in $\lambda$'s Young diagram.  We write $[\lambda]$ for the set of all boxes in the diagram.  Each box $\square \in [\lambda]$ is indexed by an ordered pair $(i, j)$, where~$i$ is $\square$'s row and~$j$ is $\square$'s column.  Note that this indexing is slightly peculiar vis-a-vis the French notation, in which the center of $\square$ has Cartesian coordinates $(j - \frac12,i - \frac12)$. We define the \emph{content} of cell $\square$ to be $c(\square)\coloneqq j-i$.  Note that in the Russian diagram, the content of~$\square$ is the $x$-coordinate of its center.  We also define the \emph{hook length} $h(\square)$ of~$\square$ via the French notation: it is the number of cells directly to the right or above~$\square$, including~$\square$ itself; equivalently, it is $(\lambda_i - j) + (\lambda'_j - i) + 1$.

Having defined ``content'' for cells in a Young diagram, we may introduce some convenient notation (essentially from~\cite{OO98b}) that generalizes the standard notions of ``falling factorial power'' and ``rising factorial power''.  First, for $z \in \R$ and $m \in \N$, recall the \emph{falling factorial power}\footnote{Or Pochhammer symbol, sometimes denoted $(z)_m$ or $z^{\underline{m}}$.}
\[
    \falling{z}{m} \coloneqq z(z-1)(z-2)\cdots(z-m+1)
\]
and \emph{rising factorial power}
\[
    \rising{z}{m} \coloneqq z(z+1)(z+2)\cdots(z+m-1).
\]
We generalize this notation to the case of an arbitrary partition $\lambda \vdash m$:
\begin{equation*}
    \falling{z}{\lambda} \coloneqq  \prod_{\square \in [\lambda]} (z - c(\square))
\quad \text{and} \quad
    \rising{z}{\lambda} \coloneqq  \prod_{\square \in [\lambda]} (z + c(\square)).
\end{equation*}

\subsubsection{Random words and Young diagrams, and symmetric polynomials}
\begin{definition}
    Let $\calA$ be an \emph{alphabet}; i.e., a totally ordered set.  Most often we consider $\calA = [d]$.  A \emph{word} is a finite sequence $(a_1, \dots, a_n)$ of elements from~$\calA$.  It is \emph{weakly increasing} if $a_1 \leq a_2 \leq \cdots \leq a_n$ and \emph{strongly (or strictly) increasing} if $a_1 < a_2 < \cdots < a_n$.  If $\calD$ is a probability distribution on~$\calA$ we write $\calD^{\otimes n}$ to denote the probability distribution on words of length~$n$ given by drawing the letters independently from~$\calD$.
\end{definition}
\begin{definition}                                  \label{def:sorted-type}
    Given a word $a \in [d]^n$, there is an associated partition $\lambda \vdash n$ of length at most~$d$ called the \emph{sorted type (or histogram)}. It is defined as follows: $\lambda_i$ is the frequency of the $i$th-most frequent letter in~$a$, for $1 \leq i \leq d$.  In other words, $\lambda$ is the histogram of letter frequencies, sorted into nonincreasing order.  For example, the sorted type of $(4,1,3,4,4,4,1,4) \in [4]^8$ is $(5,2,1,0) \vdash 8$.
\end{definition}
\begin{definition}                                  \label{def:power-sum}
    Let $x_1, \dots, x_d$ be indeterminates, typically standing for real numbers.  For $m \in \N$, the \emph{$m$th power sum symmetric polynomial} is $p_m(x) = \sum_{i = 1}^d x_i^m$.  More generally, for a partition~$\lambda$ we define $p_\lambda(x) = \prod_{i=1}^{\ell(\lambda)} p_{\lambda_i}(x)$.  By our conventions, if $\pi \in \symm{n}$ then $p_\pi(x)$ denotes $p_\lambda(x)$, where $\lambda$ is the cycle type of~$\pi$.  If $\calD = (\eta_1, \dots, \eta_d)$ is a probability distribution on~$[d]$, there is a natural interpretation of $p_\pi(\eta_d, \dots, \eta_d)$: it is the probability that a random word $\ba \sim \calD^{\otimes n}$ is invariant under the permutation~$\pi$.
\end{definition}

\begin{figure}
\centering
	\includegraphics[height=100pt]{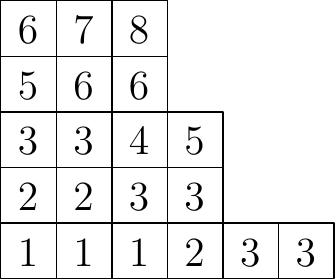}
	\caption{A semistandard tableau of shape $\lambda = (6,4,4,3,3)$ with alphabet $[8]$.}
	\label{fig:semistandard}
\end{figure}
\begin{definition}
    Let $\lambda \vdash n$, and think of its Young diagram in the French notation.  If each cell is filled with an element from some alphabet~$\calA$, we call the result a \emph{Young tableau of shape~$\lambda$}.  The Young tableau is said to be \emph{semistandard} if its entries are weakly increasing from left-to-right along rows and are strongly increasing from bottom-to-top along columns.   Figure~\ref{fig:semistandard} gives an example semistandard tableau of shape $(6,4,4,3,3)$. If the rows are in fact strongly increasing, the Young tableau is called \emph{standard}.
\end{definition}
\begin{definition}                                      \label{def:hook-formula}
    For reasons we will see later, the number of standard Young tableaus\footnote{Often spelled ``tableaux''.} of shape $\lambda \vdash n$ over alphabet~$[n]$ is denoted $\dim(\lambda)$.  It can be computed via the \emph{Hook-Length Formula} of Frame, Robinson, and Thrall~\cite{FRT54} (see also~\cite[Corollary~7.21.6]{Sta99}):
    \[
        \dim(\lambda) = \frac{n!}{\prod_{\square \in [\lambda]} h(\square)}.
    \]
\end{definition}
We will also consider counting semistandard tableaus,  via the following definition:
\begin{definition}
    Let $x_1, \dots, x_d$ be indeterminates, typically standing for real numbers.  Given $\lambda \vdash n$, the \emph{Schur polynomial} $s_\lambda(x_1, \dots, x_d)$ is the degree-$n$ homogeneous polynomial defined by  $\sum_{T} x^T$, where the sum is over all semistandard tableaus of shape~$\lambda$ over alphabet~$[d]$, and where
    \[
        x^T \coloneqq \prod_{i=1}^d x_i^{\text{(\# of occurrences of letter~$i$ in~$T$)}}.
    \]
    The following formula from~\cite[Corollary~7.21.4]{Sta99} thereby lets us count the number of such tableaus:
    \[
        s_\lambda(\underbrace{1, 1, \dots, 1}_{d \text{ entries}}) = \frac{\rising{d}{\lambda}}{\prod_{\square \in [\lambda]} h(\square)}.
    \]
\end{definition}
We record here a consequence of the above two formulas:
\begin{proposition}                 \label{prop:s111}
    Let $\lambda$ be a partition and let $d \in \Z^+$.  Then
    $\displaystyle
        s_\lambda(\underbrace{1, \dots, 1}_{d\ \textnormal{entries}}) = \frac{(\dim \lambda)\rising{d}{\lambda}}{|\lambda|!}.
    $
\end{proposition}
When $\ell(\lambda) > d$, there are no semistandard tableaus of shape~$\lambda$ over alphabet~$[d]$.
Thus, the sum $\sum_{T} x^T$ is the empty sum.
This gives us the following fact about Schur polynomials:
\begin{proposition}\label{prop:identically-zero}
Consider the Schur polynomial $s_\lambda(x_1, \ldots, x_d)$.
If $\ell(\lambda) > d$ then $s_\lambda \equiv 0$.
\end{proposition}
Though it is not at all obvious from the definition, the Schur polynomials are symmetric.  This can be inferred from the following classical fact (see e.g.~\cite[Theorem~7.15.1]{Sta99}), which expresses them as the ratio of a skew-symmetric polynomial and the Vandermonde determinant:
\begin{theorem}                                     \label{thm:schur-determinant}
$\displaystyle s_\lambda(x_1, \ldots, x_d) = \frac{\det\Bigl(x_i^{d + \lambda_j -j}\Bigr)_{ij}}{\det\Bigl(x_i^{d-j}\Bigr)_{ij}}.$
\end{theorem}

We will actually not need this formula.  Instead, we will next describe a combinatorial algorithm which gives an interpretation for $s_\lambda(\eta_1, \dots, \eta_d)$ when $\calD = (\eta_1, \dots, \eta_d)$ is a probability distribution.

\subsubsection{The RSK algorithm}                       \label{sec:RSK}
We now describe the \emph{Robinson--Schensted--Knuth (RSK)} algorithm $\rsk(\cdot)$, which takes as input a word $a \in \calA^n$ and outputs a partition $\lambda = \rsk(a) \vdash n$. The relevance of RSK to quantum spectrum testing is described at the end of this section.  As there are many descriptions of the RSK algorithm in the literature (see, e.g.,~\cite{Knu70,Bay02,Dor05,Rom14}), we will be brief.
\begin{nameddef}{The RSK algorithm}
Given as input a word $a = (a_1, \ldots, a_n)$ over (ordered) alphabet~$\calA$, the RSK algorithm produces a sequence $T_0, \dots, T_n$ of semistandard tableaus over~$\calA$, with $T_i$ having size~$i$ (and being thought of in French notation).  Tableau~$T_{i+1}$ is produced from tableau~$T_{i}$ via the ``insertion'' of letter~$a_i$ into the $1$st row.  The insertion algorithm for letter~$b$ into row~$j$ of tableau~$T$ is as follows: Find the rightmost position in the $j$th row such that if $b$~were placed there, weak-increasingness along row~$j$ would be maintained.  If this position is at the end of the row, the insertion of~$b$ is complete.  If instead it is at a cell that already contains some letter~$c$ (which will in fact be the least~$c$ in row~$j$ with $c > b$) then~$c$ is ``bumped up''.  By this we mean that the insertion algorithm is recursively applied to letter~$c$ and row $j+1$ of~$T$ (which may be a newly created row, in which the insertion will immediately terminate with~$c$ in its own row at the top of~$T$). In the end, the output of the RSK algorithm is the Young diagram $\lambda \vdash n$ given by the \emph{shape} of~$T_n$; i.e., $\rsk(a)$ is $T_n$ with its cell entries erased.
\end{nameddef}

To get some feel for this algorithm, note that if the inserted word~$a$ is weakly increasing then $\rsk(a) = (n) \vdash n$.  On the other hand, if $a$ is strongly decreasing, the output will be $\rsk(a) = (1, 1, \dots, 1) \vdash n$.  More generally, it is not hard to show that when $\rsk(a) = \lambda$, the value $\lambda_1$ is the length of the longest weakly increasing subsequence of~$a$, and $\ell(\lambda) = \lambda'_1$ is the length of the longest strongly decreasing subsequence of $a$.  Even more generally, we have the following theorem of Greene~\cite{Gre74}, completely characterizing the partition $\rsk(a)$ in terms of increasing subsequences:
\begin{theorem}                                     \label{thm:greene}
    Let $\rsk(a) = \lambda$.  Then for each $k \geq 1$, the value $\lambda_1 + \ldots + \lambda_k$ is the length of the longest $k$-increasing subsequence in~$a$ (as defined in Fact~\ref{fact:quantum-classically}).
\end{theorem}
\noindent Indeed, the RSK algorithm is most commonly  used in the literature to study the length of the longest increasing subsequence of a random permutation (equivalently, of a random word $\ba \sim \calX^{\otimes n}$, where $\calX$ denotes the uniform distribution on the alphabet $\calA = [0,1]$).

Let us note one immediate consequence of Greene's theorem.
(This consequence may also be derived directly from the description of the RSK algorithm.)
\begin{proposition}\label{prop:majorizer}
Given $a \in [d]^n$,
let $\rsk(a) = \lambda$.
Write $c_i(a)$ for the number of letter~$i$'s in~$a$.  Then $\lambda$ majorizes $c(a) := (c_1(a), \ldots, c_d(a))$.
\end{proposition}
\noindent
To see why this is true, note that for each $k \in [d]$, the all one's, all two's, \ldots, and all $k$'s subsequences together form a $k$-increasing subsequence of size $c_1(a) + \ldots + c_k(a)$, which by Theorem~\ref{thm:greene} is at most $\lambda_1 + \ldots + \lambda_k$, giving the proposition.
As $c(a)$ is the histogram of $a$, this shows that we can view $\rsk(a)$ as a ``shifted histogram'' of $a$ in which cells are shifted towards the lower numbers.

Although Greene's Theorem succinctly characterizes the output by the RSK algorithm, it is important to retain the algorithm itself and even to consider an extension of it.  Suppose that when the RSK algorithm is applied to~$a$ we also form a standard tableau~$T'$ over alphabet~$[n]$, where $T'$ has the same shape as $T_n$ and each cell~$\square$ in~$T'$ is labeled by the ``time'' at which $\square$ was created in~$T_n$.  As noted by Knuth~\cite{Knu70}, the word~$a$ is uniquely determined by the pair $(T_n, T')$.  As a consequence of this and of previous formulas, it is not hard to verify the following important fact, perhaps first observed by Its, Tracy, and Widom~\cite[equation~(2-1)]{ITW01}:
\begin{proposition}\label{prop:rsk-correspondence}
    Let $\ba \sim \calD^{\otimes n}$, where $\calD = (\eta_1, \dots, \eta_d)$ is a probability distribution on~$[d]$.  Then for each $\lambda \vdash n$,
    \[
        \Pr[\rsk(\ba) = \lambda] = \dim(\lambda) \cdot s_\lambda(\eta_1, \dots, \eta_d).
    \]
\end{proposition}
\noindent By the symmetry of the Schur polynomials, this implies the surprising fact that the distribution of~$\rsk(\ba)$ is invariant to permutations of~$\calD$.

Finally, we mention the connection between the RSK algorithm and quantum spectrum testing.  As we will eventually see in Section~\ref{sec:understanding} (Remark~\ref{rem:quantum-free}), all of quantum spectrum testing can be boiled down to classical testing of symmetric probability distributions~$\calD$, with the following twist:  Rather than getting to see a random word~$\ba$ sampled from $\calD^{\otimes n}$, the tester only gets to see the partition~$\blambda = \rsk(\ba)$.  In light of Greene's Theorem~\ref{thm:greene}, this statement is equivalent to Fact~\ref{fact:quantum-classically}.

\ignore{
A combinatorial viewpoint is given by the \emph{Robinson--Schensted--Knuth (RSK)} algorithm.
The main subroutine of the RSK algorithm is \emph{Schensted's bumping algorithm}.\rnote{I'd also like to throw in the O'Connell version where you take infinitely many samples and condition on always having a decreasing fingerprint\dots}
We first describe this subroutine, using almost verbatim the description given in~\cite{Bay02}:\rnote{I like Dan Romik's description too}
\begin{nameddef}{Schensted's bumping algorithm}
Given a Young tableau $T$ over an alphabet $\calA$ and a letter $a \in \calA$, Schensted's bumping algorithm produces a new Young tableau as follows:
\begin{enumerate}
\item Look at the first row of $T$ and find the smallest letter that is larger than $a$. Replace this letter with $a$.   If the smallest letter larger than $a$ occurs more than once in the row then choose the one furthest to the left.  If no such letter is larger than $a$, simply place $a$ at the end of the first row.
\item If a letter $b$ was replaced by $a$ in the first row then bump $b$ into the second row using the same method as above.  If there is no row to add $b$ to, then $b$ has been bumped out of the bottom, in which case it forms a new row with one entry.
\item Repeat the process on each row of the tableau until either some letter gets added to the end of a row or until it is bumped out of the bottom.
\end{enumerate}
\end{nameddef}

The RSK simply iterates the bumping algorithm to produce a Young diagram from a word $w$.
\begin{nameddef}{The RSK algorithm}
Given a word $w = (w_1, \ldots, w_n)$, the RSK algorithm works as follows:
\begin{enumerate}
\item Let $T_0$ be the empty tableau.
\item Bump $w_1$ into $T_0$ to form $T_1$.
\item Bump $w_2$ into $T_1$ to form $T_2$.
\item Continue in this manner until every letter of $w$ has been bumped into the tableau, producing the final tableau $T_n$. Output the Young diagram with the same shape as $T_n$.
\end{enumerate}
We write $\rsk(w)$ for the Young diagram output by the RSK algorithm.
\end{nameddef}

Given a word $w$, the shape of the Young diagram $\rsk(w)$ is known to depend on various natural combinatorial properties of $w$.
For example, if we define
\begin{definition}
Given a word $w$, a \emph{weakly increasing subsequence} of $w$ is a set of indices $i_1< \ldots< i_m$ such that $w_{i_1} \leq \ldots \leq w_{i_m}$.
Similarly, a \emph{strongly decreasing subsequence} of $w$ is a set of indices $i_1 < \ldots < i_m$ such that $w_{i_1} > \ldots > w_{i_m}$.
\end{definition}
\noindent
Then these quantities characterize the width and height of $\rsk(w)$:
\begin{fact}
Given a word $w$, let $\lambda = \rsk(w)$.
Then $\lambda_1$ is the length of the longest weakly increasing subsequence of $w$,
and $\ell(\lambda)$ is the length of the longest strongly decreasing subsequence of $w$.
\end{fact}
\noindent
For a proof of this, consult~XXX.

Now we can describe the correspondence between the RSK algorithm and weak Schur sampling.
First, we will define our inhomogeneous random word model:
\begin{definition}
Let $\calA$ be an alphabet, and let $\mu$ be a probability distribution over $\calA$.
Then $\word{n}(\calA, \mu)$ outputs $w= (w_1, \ldots, w_n)$, where each $w_i$ is drawn independently from $\mu$.
Let $\eta_1, \ldots, \eta_d$ be nonnegative numbers with $\eta_1+\ldots+\eta_d = 1$.
Then we define $\word{n}(\eta_1, \ldots, \eta_d) \coloneqq  \word{n}([d], (\eta_1, \ldots, \eta_d))$.
\end{definition}
\begin{fact}\label{fact:rsk-correspondence}
Let $\rho$ be a density matrix with eigenvalues $\eta_1, \ldots, \eta_d$.
Consider a random partition $\lambda\vdash n$
produced by sampling $w \sim \word{n}(\eta_1, \ldots, \eta_d)$ and outputting $\rsk(w)$.
Then $\lambda$ is distributed according to $\SWdens{n}{\rho}$.
\end{fact}
\noindent
This fact is proved in~XXX.\jnote{is this proved anywhere in like a survey or something?  I see ``proved'' it in the Its-Tracy-Widom paper and in the Xu thesis, but it would be nice to be able to point to a survey or a textbook with this result.}\rnote{See the bottom of page 23 of the Feray--Meliot paper, or Sec 3 of O'Connell.}

One property that Fact~\ref{fact:rsk-correspondence} makes clear is that weak Schur sampling is never affected by the nonzero parts of the spectum of a state.
This means that, for instance, the state $\rho_1$ with eigenvalues $(.5, .5, 0)$ which sits in $\C^3$ and the state $\rho_2$ with eigenvalues $(.5, .5, 0, 0, 0)$ which sits in $\C^5$ will both produce identically distributed partitions when run through the weak Schur sampling algorithm.
In other words, we have the following strengthening of Fact~\ref{fact:only-spectrum}:
\begin{fact}\label{fact:only-nonzero-spectrum}
Let $\rho$ be a density matrix. The output of weak Schur sampling on $\rho^{\otimes n}$ depends only on the \emph{nonzero} spectrum of $\rho$ (and not, for example, on the ambient dimension of $\rho$).
\end{fact}
}

\subsection{Representation theory, and the symmetric group}                          \label{sec:rep-theory}

Herein we recall some basics of representation theory.  We will mainly focus on $\C$-representations of finite groups~$G$ (though at one point we will want to consider representations of the unitary group).  We may therefore define a \emph{representation} $\mu$ of $G$ to be a group homomorphism from $G$ into $\unitary{d}$, for some $d \in \Z^+$.  Here $\unitary{d}$ denotes the group of $d \times d$ unitary matrices.  The number~$d$ is also called the \emph{dimension} of the representation~$\mu$ and is denoted $\dim(\mu)$.

Two representations $\mu_1$ and $\mu_2$ are said to be \emph{isomorphic} if there is some unitary matrix $U$ such that $U \mu_1 U^\dagger = \mu_2$.  In this case we write  $\mu_1 \cong \mu_2$.  The \emph{direct sum} of $k$ representations $\mu_1, \ldots, \mu_k$ produces the representation $\mu$ given by block-diagonal matrices:
\begin{equation}
\mu(g) \coloneqq
\left[
\begin{array}{cccc}
\mu_1(g) &0 &\ldots&0\\
0 &\mu_2(g) &\ldots&0\\
\vdots & \vdots & \ddots & \vdots\\
0 & 0 & \hdots & \mu_k(g)
\end{array}\right]\label{eq:reducible}
\end{equation}
for all $g \in G$.  Equivalently, we may write
\begin{equation}
\mu(g)
\coloneqq  \sum_{i=1}^k |i\rangle\langle i| \otimes \mu_i(g).\label{eq:irr-expansion}
\end{equation}
We will also write  $\mu = \mu_1 \oplus \ldots \oplus \mu_k$ to denote that $\mu$ is the direct sum of $\mu_1, \ldots, \mu_k$.

Let $\mu_1$ be a representation of the group $G_1$ and $\mu_2$ be a representation of the group $G_2$.  Then the \emph{tensor product} of $\mu_1$ and $\mu_2$, denoted $\mu_1 \otimes \mu_2$, is the representation defined by
\begin{equation*}
\left(\mu_1 \otimes \mu_2\right)(g, h)
\coloneqq  (\mu_1(g)) \otimes (\mu_2(h)),
\end{equation*}
where the right-hand side uses the ordinary matrix tensor product. We have $\dim(\mu_1 \otimes \mu_2) = \dim(\mu_1) \cdot \dim(\mu_2)$. 

In our setting, a representation~$\mu$ of~$G$ is said to be \emph{reducible} if there are representations $\mu_1$~and~$\mu_2$ such that $\mu \cong \mu_1 \oplus \mu_2$.  Otherwise it is \emph{irreducible}, and is often called an \emph{irrep} for brevity.  Every representation can be uniquely decomposed into a direct sum of irreps (up to isomorphism and rearrangement of summands).  Further, the set of all irreps of~$G$ (up to isomorphism), denoted $\widehat{G}$, is finite.  Indeed, if we define the \emph{regular representation} of~$G$ to be the $|G|$-dimensional representation~$R$ given by $R(g) = \sum_{h \in G} \ket{gh} \bra{h}$), then $R$'s decomposition into irreps contains every $\mu \in \widehat{G}$, with $\mu$ occurring $\dim(\mu)$ times.  As a consequence, we have the formula
\[
    |G| = \sum_{\mu \in \widehat{G}} (\dim \mu)^2.
\]
This fact leads to a natural \emph{probability distribution} on irreps of~$G$:
\begin{definition}\label{def:group-plancherel}
    For a finite group~$G$, the \emph{Plancherel distribution} is the probability distribution on irreps in which $\mu \in \widehat{G}$ has probability $(\dim \mu)^2/|G|$.
\end{definition}

For a group $G$ and a representation $\mu$, the character $\chi_\mu$ is the function $\chi_\mu:G\rightarrow \C$ defined by
\begin{equation*}
\chi_\mu(g) = \trace(\mu(g)),
\end{equation*}
for each $g\in G$.  We have the following simple fact:
\begin{fact}\label{fact:class-function}
Let $\mu$ be a representation of $G$.  Then $\chi_\mu$ is a \emph{class function}; i.e., it is constant on the conjugacy classes of $G$.
\end{fact}

We now recall some basics of Fourier analysis over an arbitrary finite group~$G$ (though we will ultimately only need the case $G = \symm{n}$).  For $f, g \co G \to \C$ we define $\la f, g \ra = \E_{\bu \sim G}[f(\bu)\ol{g(\bu)}]$.  Under this inner product, the characters $(\chi_{\mu})_{\mu \in \wh{G}}$ form an orthonormal basis for the space of class functions $f \co G \to \C$.  For general $f, g \co G \to \C$ we define $(\conv{f}{g})(u) = \E_{\bv \sim G}[f(\bv)g(\bv^{-1} u)]$; this includes a nonstandard normalization by $\frac{1}{|G|}$.  For a class function~$f$ and $\mu \in \wh{G}$ we employ the following ``Fourier notation'': $\four{f}(\mu) = \la f, \chi_\mu \ra$.  (According to standard notation we would have $\four{f}(\mu) = \frac{1}{|G|}\tr\mleft(\wh{\ol{f}}\mright)$). Then Fourier inversion is simply $f = \sum_\mu \four{f}(\mu)\chi_\mu$.  Further, if $g$ is another class function we have the formula $\four{f \ast g}(\mu) = \frac{1}{\dim \mu}\four{f}(\mu)\four{g}(\mu)$.

We close this section by specifically discussing the representation theory of the symmetric group~$\symm{n}$.  Two permutations $\pi, \sigma \in \symm{n}$ are conjugate within the group~$\symm{n}$ if and only if they have the same cycle type. As a result, the conjugacy classes of~$\symm{n}$ can be identified with the partitions of~$n$. As it happens, the set $\irrsymm{n}$ of irreps of the symmetric group can \emph{also} be naturally identified with the partitions of~$n$.  For $\lambda \vdash n$, we will use the notation $\sirrep_\lambda$ for the corresponding irrep of~$\symm{n}$.  (To avoid getting too far afield, we will not actually describe the representation~$\sirrep_\lambda$.)  Recalling Fact~\ref{fact:class-function}, we introduce the following notation:
\begin{definition}
Let $\lambda \vdash n$.  We denote the character $\chi_{\sirrep_\lambda}$ more simply as $\chi_\lambda$.  We remark that $\chi_\lambda$ is known to take on only rational values; in particular, $\overline{\chi_{\lambda}} = \chi_\lambda$. If $\mu \vdash n$ then we let  $\chi_\lambda(\mu)$ denote $\chi_\lambda(\pi)$, where $\pi \in \symm{n}$ is any permutation with cycle type~$\mu$.  This is well defined since $\chi_\lambda$ is constant on the conjugacy classes of~$\symm{n}$.  Finally, we also write $\dim(\lambda)$ for $\dim\mleft(\sirrep_\lambda\mright)$.  It is well known~\cite[Theorem 2.6.5]{Sag01} that $\dim(\lambda)$ is equal to the number of standard Young tableaus of shape~$\lambda$ over alphabet~$[n]$, explaining the notation from Definition~\ref{def:hook-formula}.
\end{definition}

Following Stanley~\cite[Corollary~7.17.5]{Sta99}, we can actually give a definition of the symmetric group characters $\chi_\mu$ in terms of the power sum and Schur polynomials:
\begin{theorem}                                     \label{thm:power-schur-relation}
    In the context of Fourier analysis over the group $G = \symm{n}$, suppose $\mu \vdash n$ and $x \in \C^d$.  Then $p_{(\cdot)}(x) \coloneqq \pi \mapsto p_{\pi}(x)$ is a class function, and its Fourier coefficients are given by
    \[
        \four{p_{(\cdot)}(x)}(\mu) = s_\mu(x).
    \]
\end{theorem}
\noindent Although this can be taken as an implicit definition of the characters $\chi_\mu$, we will more often think of the characters $\chi_\mu$ as ``known'' and of Theorem~\ref{thm:power-schur-relation} as letting us express the Schur polynomials in terms of the power sum polynomials.

\subsection{Weak Schur sampling}                    \label{sec:weak-schur}

In this section we will introduce the weak Schur sampling algorithm. Our treatment of this topic will heavily follow the treatments given in Aram Harrow's thesis~\cite{Har05} and the paper~\cite{CHW07}.

To motivate the algorithm let us briefly consider the classical problem of testing symmetric properties of probability distributions on~$[d]$.  In this model, the tester obtains a random word $\ba = (\ba_1, \dots, \ba_n)$, where each letter $\ba_i$ is drawn independently from an unknown distribution~$\calD$ on~$[d]$.  The tester wants to decide whether $\calD$ satisfies a certain symmetric property~$\calP$.  Since the samples $\ba_1, \dots, \ba_n$ are independent, the tester could---without loss of generality---randomly permute them according to any~$\pi \in \symm{n}$.  Similarly, since the property~$\calP$ is symmetric, the tester could---again, without loss of generality---simultaneously apply any permutation $\sigma \in \symm{d}$ to the letters it sees.  Roughly speaking, the tester can ``factor out'' the action of the group $\symm{n} \times \symm{d}$.  The information that remains is precisely the sorted type~$\blambda \vdash n$ of~$\ba$ (recall Definition~\ref{def:sorted-type}).\footnote{This partition carries the same information as the so-called ``fingerprint'' used in classical property literature~\cite{Bat01, Val08}.}  Thus we see that the task of analyzing property testing of symmetric probability distributions boils down to the task of understanding the random partition $\blambda \vdash n$ (of length at most~$d$) induced as the sorted type of a random word drawn from $\calD^{\otimes n}$.

A similar but more complicated state of affairs holds for quantum spectrum testing.  In this case, there is an unknown $d$-dimensional mixed state~$\rho$, and the tester may measure $n$~copies, $\rho^{\otimes n}$, in an attempt to determine whether~$\rho$ satisfies a certain unitarily-invariant property~$\calP$.  As before, the tester could (without loss of generality) randomly permute the copies according to any~$\pi \in \symm{n}$.  And in this quantum scenario, by the unitary-invariance of~$\calP$, the tester could also (without loss of generality) simultaneously apply any unitary $U \in \unitary{d}$ to each copy.  \emph{Weak Schur sampling} refers to the process of ``factoring out'' this action of $\symm{n} \times \unitary{d}$.  What remains is again a random partition $\blambda \vdash n$ of length at most~$d$, whose distribution depends only on the spectrum of~$\rho$.  (In fact, as we will see later in Remark~\ref{rem:quantum-free}, the distribution of~$\blambda$ is precisely that of $\rsk(\ba)$ where $\ba$ is a random word chosen according to the probability distribution on~$[d]$ defined by $\rho$'s spectrum.)  To understand this situation more thoroughly, we will need to discuss representation theory in more detail.\footnote{In particular, we will go slightly beyond the framework from Section~\ref{sec:rep-theory} by mentioning representations of the unitary group, which is of course not a finite group.}

As mentioned above, the groups $\symm{n}$ and $\unitary{d}$ each have a natural, unitary action on the space~$(\C^d)^{\otimes n}$; the associated representations $\srep$ and~$\urep$ (respectively) are defined on the standard basis vectors $\ket{a_1} \otimes \ket{a_2} \otimes \cdots \otimes \ket{a_n}$ (for $a_i \in [d]$) via
\begin{align*}
\srep(\pi)\; |a_1\rangle \otimes |a_2\rangle \otimes \ldots \otimes |a_n \rangle
&= |a_{\pi^{-1}(1)}\rangle \otimes |a_{\pi^{-1}(2)}\rangle \otimes \ldots \otimes |a_{\pi^{-1}(n)} \rangle,\\
\urep(U)\; |a_1\rangle \otimes |a_2\rangle \otimes \ldots \otimes |a_n \rangle
&= \ \,\!(U |a_1\rangle) \;\!\otimes\, (U |a_2\rangle) \,\otimes\,\! \ldots \otimes (U |a_n \rangle).
\end{align*}
We know the irreps of~$\symm{n}$ are indexed by partitions of~$n$; thus, the representation~$\srep$ must decompose as
\begin{equation}                                \label{eqn:srep-decomposition}
    \srep(\pi) \mathop{\cong}^{\symm{n}} \bigoplus_{\lambda \vdash n} \sirrep_\lambda(\pi) \otimes I_{m_\lambda},
\end{equation}
with $m_\lambda$ denoting the number of copies of $\sirrep_\lambda$ in the decomposition.  The representation~$\urep$ also decomposes into irreps of the group~$\unitary{d}$. As it happens, these (infinitely many) irreps can \emph{also} be naturally identified with partitions; specifically, for each partition $\lambda \in \allpartitions$ with length at most~$d$, there is an associated irrep $\uirrep_\lambda^d \in \irrunitary{d}$.  Furthermore, the theory of \emph{Schur--Weyl duality} states that there is significant joint structure to these two decompositions.  This structure ultimately arises because the two representations~$\srep$ and~$\urep$ commute (i.e., $\srep(\pi)\urep(U) = \urep(U)\srep(\pi)$ for all $\pi \in \symm{n}$, $U\in \unitary{d}$), and hence the simultaneous action $\srep \urep$ defined by  $\srep\urep(\pi, U)\coloneqq \srep(\pi)\urep(U)$ is a representation of the direct product group $\symm{k}\times \unitary{d}$.
\begin{named}{Schur--Weyl duality}
    $\displaystyle \srep \urep \mathop{\cong}^{\symm{n} \times \unitary{d}} \bigoplus_{\lambda \vdash n} \sirrep_\lambda \otimes \uirrep_\lambda^d$.
\end{named}
\noindent In particular, by taking $U = id$ we see that $m_\lambda$, the multiplicity of $\sirrep_\lambda$ in the decomposition of~$\srep$, is equal to $\dim(\uirrep_\lambda^d)$.  Similarly, by taking $\pi = id$, we see that the multiplicity of $\uirrep_\lambda^d$ in the decomposition of~$\urep$ is $\dim(\lambda) = \dim(\sirrep_\lambda)$.

To restate Schur--Weyl duality, there exists a certain $d^n\times d^n$ unitary matrix $\schurbasis$ such that
\begin{equation}
\schurbasis\srep(\pi)\urep(U)\schurbasis^\dagger
= \sum_{\lambda \vdash n} |\lambda \rangle\langle \lambda|\otimes \sirrep_\lambda(\pi)\otimes\uirrep_\lambda^d(U),\label{eq:schur-weyl}
\end{equation}
for all $\pi \in \symm{n}$, $U \in \unitary{d}$.
We view $\schurbasis$ as a unitary linear transformation that performs a change of basis, from the standard basis into the \emph{Schur basis}.  We may now state the weak Schur sampling algorithm:
\begin{nameddef}{Weak Schur sampling}
Given $\rho^{\otimes n}$, where $\rho$ is a $d$-dimensional mixed state, the weak Schur sampling algorithm works as follows:
\begin{enumerate}
\item Measure $\rho^{\otimes n}$ in the Schur basis, receiving basis state $|\blambda\rangle \otimes |\bp\rangle \otimes |\bq\rangle$.
\item Output $\blambda$, a partition of size~$n$ and length at most~$d$.
\end{enumerate}
\end{nameddef}
We will write $\SWdens{n}{\rho}$ for the distribution on partitions induced from $\rho^{\otimes n}$ by the weak Schur sampling algorithm. We will also use the shorthand
\[
    \SWdensProb{n}{\rho}{\lambda} \coloneqq \Pr_{\blambda \sim \SWdens{n}{\rho}}[\blambda = \lambda].
\]
\ignore{
By virtue of Fact~\ref{fact:simple-decomp}, we have
\[
    \Pr_{\blambda \sim \SWdens{n}{\rho}}[\blambda = \lambda] = \trace(I\otimes R^\rho_\lambda) = \dim(\lambda) \cdot \trace(R^\rho_\lambda).
\]}

As we will state shortly, performing weak Schur sampling is without loss of generality in the context of testing unitarily invariant properties.  To see why, suppose $\rho$ is a $d$-dimensional mixed state, and consider the product mixed state $\rho^{\otimes n}$.  Then it's not too hard to show (using invariance under $\srep$ and Schur's Lemma, see e.g.~\cite[equation~(6.1)]{Har05}) that when represented in the Schur basis, it has a ``trivial $\symm{n}$ register'':
\begin{fact}\label{fact:simple-decomp}
We may write
$\schurbasis \rho^{\otimes n} \schurbasis^\dagger =\sum_{\lambda \vdash n}|\lambda \rangle\langle \lambda|\otimes I \otimes R^\rho_\lambda$, for some matrices $R^\rho_\lambda$.  Here, for each $\lambda$ we interpret $I$ as the $\dim(\lambda) \times \dim(\lambda)$ identity matrix.
\end{fact}
\ignore{
\begin{proof}
We may begin by writing
\begin{equation*}
\schurbasis \rho^{\otimes n} \schurbasis^\dagger =\sum_{\lambda, \mu \vdash n}|\lambda \rangle\langle \mu|\otimes R_{\lambda, \mu},
\end{equation*}
for some matrices $R_{\lambda,\mu}$.
Next, $\urep(I) = I$ and $\uirrep_\lambda^n(I) = I$ for all $\lambda$.
Thus, Equation~\eqref{eq:schur-weyl} implies that
\begin{equation*}
\schurbasis\srep(\pi)\schurbasis^\dagger
= \sum_{\lambda \vdash n} |\lambda \rangle\langle \lambda|\otimes \sirrep_\lambda(\pi)\otimes I,
\end{equation*}
where, for each $\lambda$, $I$ is the $\dim(\uirrep_\lambda) \times \dim(\uirrep_\lambda)$ identity matrix.
Because $\rho^{\otimes n}$, we know that $\srep(\pi) \rho^{\otimes n} \srep(\pi)^\dagger = \rho^{\otimes n}$ for all $\pi \in \symm{n}$.
As a result,
\begin{align*}
\sum_{\lambda, \mu \vdash n}|\lambda \rangle\langle \mu|\otimes R_{\lambda, \mu}
& = \schurbasis \rho^{\otimes n} \schurbasis^\dagger\\
& = \schurbasis \srep(\pi) \rho^{\otimes n} \srep(\pi)^\dagger \schurbasis^\dagger\\
& = \sum_{\lambda, \mu \vdash n}|\lambda \rangle\langle \mu|\otimes \left((\sirrep_\lambda(\pi)\otimes I)R_{\lambda, \mu}(\sirrep_\mu(\pi)^{\dagger}\otimes I)\right),
\end{align*}
for all $\pi \in \symm{n}$.
Thus, $R_{\lambda, \mu} = \left((\sirrep_\lambda(\pi)\otimes I)R_{\lambda, \mu}(\sirrep_\mu(\pi)^{\dagger}\otimes I)\right)$.
Writing $R_{\lambda,\mu} = \sum_{i, j} R_{\lambda,\mu}^{(i,j)}\otimes |i\rangle\langle j|$, this implies that
$R_{\lambda,\mu}^{(i,j)} = \sirrep_\lambda(\pi) R_{\lambda,\mu}^{(i,j)} \sirrep_\mu(\pi)^\dagger$, for each $i, j$.
As a result, it is a simple consequence of Schur's Lemma that $R_{\lambda,\mu}^{(i,j)} = 0$ unless $\lambda = \mu$, in which case $R_{\lambda, \mu}^{(i,j)}$ must be a scalar matrix.
This proves the fact.\end{proof}
}
As a consequence, it makes sense that a testing algorithm may discard the~$\symm{n}$ register.  Now in general, the ``$\unitary{d}$ register'' $R^\rho_\lambda$ of $\rho^{\otimes n}$ is not trivial, and thus it may seem like the tester is losing information by discarding it.  (Indeed, this potential loss is the source of the word ``weak'' in the phrase ``weak Schur sampling''.) However when testing unitarily invariant properties of~$\rho$, the state~$\rho^{\otimes n}$ should be treated no differently than the state $\urep(U)\rho^{\otimes n}\urep(U^\dagger) = (U\rho U^\dagger)^{\otimes n}$, for any $U \in \unitary{d}$. In particular, a tester could average over all unitaries~$U$, and this \emph{would} cause the resulting state to have trivial a $\unitary{d}$ register in the Schur basis.  This idea is formalized in the next lemma, which shows that weak Schur sampling is an optimal quantum measurement for the testing of unitarily invariant properties. The lemma, implicit in~\cite{CHW07}, can be found with proof in~\cite[Lemma~19]{MdW13}.
\begin{lemma}                                       \label{lem:wolog}
Let $\calP$ be a unitarily invariant property of $d$-dimensional mixed states.  Assume there exists a tester which uses $n$ copies of the input state $\rho$, accepts all states $\rho\in\calP$ with probability at least $1-\delta$, but accepts all states which are $\eps$-far from $\calP$ with probability at most $1-f(\eps)$ for $\eps>0$.  Then there exists a tester with the same parameters which consists of performing weak Schur sampling on $\rho^{\otimes n}$ and then classically postprocessing the results.
\end{lemma}
\noindent
As a result of this lemma, we are able to focus exclusively on the weak Schur sampling algorithm in this paper.  One final remark: Although our quantum spectrum testing upper bounds are formally only concerned with copy complexity, they can in fact also be implemented \emph{efficiently}, by (quantum) algorithms running in time $\poly(n, \log d, \log(1/\eps))$.  This holds because the only expensive operation is the computation of the Schur change-of-basis, and this can be done in $\poly(n, \log d, \log(1/\eps))$ time; see~\cite[Appendix~A]{BCH05}, \cite[Section~8.1.1]{Har05}.

\subsection{Understanding the weak Schur sampling distribution}\label{sec:understanding}

There are a several ways to understand the probability distribution induced by weak Schur sampling algorithm, each of which proves advantageous in different settings.  Let us begin with a direct calculation that expresses the probabilities in terms of the Schur polynomials.  The following known fact may be attributed to Alicki et al.~\cite{ARS88}; see~\cite[equation~(36)]{Aud06} for further discussion.  We will include a proof for the reader's convenience.
\begin{proposition}     \label{prop:schur-probability}
    Let $\rho$ be a $d$-dimensional density matrix with eigenvalues $\eta_1, \eta_2, \ldots, \eta_d$.  Then
    \begin{equation*}
        \SWdensProb{n}{\rho}{\lambda} 
            = \dim(\lambda)\cdot s_\lambda(\eta_1, \eta_2, \ldots, \eta_d).
    \end{equation*}
    In particular, $\SWdens{n}{\rho}$ depends only on the spectrum of~$\rho$.
\end{proposition}
\begin{remark}                                  \label{rem:quantum-free}
    As this is the exact same formula as in Proposition~\ref{prop:rsk-correspondence}, we conclude that if $\calD$ is the probability distribution on~$[d]$ given by the spectrum of~$\rho$ (in any order), then
    \[
        \SWdensProb{n}{\rho}{\lambda} = \Pr_{\ba \sim \calD^{\otimes n}}[\rsk(\ba) = \lambda].
    \]
    This gives a completely ``quantum-free'' way of analyzing quantum spectrum testing, as mentioned in Fact~\ref{fact:quantum-classically}.  Nevertheless, we will actually use this fact only occasionally (mainly via Theorem~\ref{thm:greene}).  As we will see later, interpreting $\SWdens{n}{\rho}$ via representation theory proves to be more powerful.
\end{remark}
\begin{proof}[Proof of Proposition~\ref{prop:schur-probability}]
    By definition, $\SWdensProb{n}{\rho}{\lambda} = \tr(\Pi_{\lambda} \rho^{\otimes n})$, where $\Pi_\lambda$ denotes the operator that projects onto the subspace corresponding to $\lambda$ in the Schur basis.  It is a basic fact of representation theory (following from orthogonality relations, see e.g.~\cite[Equation~(7)]{CHW07}) that from the decomposition~\eqref{eqn:srep-decomposition} of~$\srep$ we may deduce
    \[
        \Pi_\lambda = \dim(\lambda) \cdot \E_{\bpi\sim\symm{n}}\left[\overline{\chi_{\sirrep_\lambda}(\bpi)} \cdot \srep(\bpi)\right] = \dim(\lambda) \cdot \E_{\bpi\sim\symm{n}}\left[\chi_{\lambda}(\bpi) \cdot \srep(\bpi)\right].
    \]
    Thus
    \[
        \SWdensProb{n}{\rho}{\lambda} = \dim(\lambda) \cdot \E_{\bpi\sim\symm{n}}\left[\chi_\lambda(\bpi) \cdot \tr(\srep(\bpi)\rho^{\otimes n})\right].
    \]
    To compute the trace, we may assume by unitary invariance that $\rho = \diag(\eta_1, \dots, \eta_d)$.  Thus
    \[
        \rho^{\otimes n} = \sum_{\substack{\text{words} \\ (a_1, \dots, a_n) \in [d]^n}} \left(\prod_{i=1}^n \eta_{a_i}\right) \ket{a_1,\dots, a_n}\bra{a_1,\dots, a_n}.
    \]
    Notice that if we let $\calD_\eta$ denote the probability distribution on~$[d]$ in which $\calD_\eta(a) = \eta_a$, then the coefficient $\prod_{i=1}^n \eta_{a_i}$ above is $\calD_\eta^{\otimes n}(a_1, \dots, a_n)$; i.e., the probability that a random length-$n$ word drawn i.i.d.\ from~$\calD_\eta$ is equal to $(a_1, \dots, a_n)$.  From the definition of~$\srep(\pi)$ we further deduce
    \[
        \srep(\pi) \rho^{\otimes n} = \sum_{(a_1, \dots, a_n)} \calD_\eta^{\otimes n}(a_1, \dots, a_n) \ket{a_{\pi^{-1}(1)}, \dots, a_{\pi^{-1}(n)}}\bra{a_1,\dots, a_n}.
    \]
    We immediately conclude that $\tr(\srep(\bpi)\rho^{\otimes n})$ is equal to the sum over all $\pi$-invariant words $(a_1, \dots, a_n)$ of~$\calD_\eta^{\otimes n}(a_1, \dots, a_n)$.  Recalling Definition~\ref{def:power-sum}, this is precisely given by the power sum polynomial $p_{\pi}(\eta_1, \dots, \eta_d)$.  Therefore
    \[
        \SWdensProb{n}{\rho}{\lambda} = \dim(\lambda) \cdot \E_{\bpi\sim\symm{n}}\left[\chi_\lambda(\bpi) \cdot p_{\pi}(\eta_1, \dots, \eta_d)\right],
    \]
    and the proposition now follows from Theorem~\ref{thm:power-schur-relation}.
\end{proof}

\ignore{

\subsubsection{Schur functions}
One way of computing $\Pr[\lambda\mid\rho^{\otimes n}]$
is by using a family of symmetric polynomials called the \emph{Schur polynomials}.
For each partition $\lambda =(\lambda_1,\ldots, \lambda_k)$, $s_\lambda(x_1, \ldots, x_k)$ is a symmetric polynomial in the $x_i$'s;
together these polynomials form a basis for the space of all $k$-variable symmetric polynomials.
There are many equivalent ways of defining these polynomials, one of which is by the following somewhat opaque determinantal formula:
\begin{definition}
Let $\lambda = (\lambda_1, \lambda_2, \ldots, \lambda_k)$ be a partition.  Then
\begin{equation*}
s_\lambda(x_1, \ldots, x_k) \coloneqq  \frac{\det[x_i^{\lambda_j + k -j} ]}{\det[x_i^{k-j}]}.
\end{equation*}
\end{definition}
\noindent
A second expression for the Schur polynomials involves the \emph{power sum symmetric polynomials} $p_\mu$.
Given input $x_1, \ldots, x_n$, the $i$-th power sum symmetric polynomial outputs $p_i(x) = \sum_j x_j^i$.
For a partition $\mu=(\mu_1, \mu_2, \ldots, \mu_m)$, we define the polynomial $p_\mu\coloneqq p_{\mu_1}\cdot p_{\mu_2}\cdots p_{\mu_m}$,
and for a permutation $\pi$, we define $p_\pi\coloneqq p_\mu$, where $\mu$ is the cycle type of $\pi$.  Then we can write the Schur polynomial as follows:
\begin{fact}\label{fact:looks-like-projector-formula}
Let $\lambda = (\lambda_1, \lambda_2, \ldots, \lambda_k)$ be a partition of~$n$. Then\jnote{should be moved to some other section, preferably whenever we introduce the related formula for $s^*$ polynomials (that's the only place that we use this)}
\begin{equation*}
s_\lambda(x_1, \ldots, x_k) = \E_{\pi \in \symm{n}}\left[\chi_\lambda(\pi)\cdot p_\pi(x)\right].
\end{equation*}
\end{fact}
\noindent
A proof of this fact can be found in~\cite[Equation~(38)]{Aud06}.

Our main use of Schur functions comes from the following fact.
\begin{fact}\label{fact:schur-probability}
Let $\rho$ have eigenvalues $\eta_1, \eta_2, \ldots, \eta_d$.  Then\rnote{Note: it now follows from Proposition~\ref{prop:s111} that if all $\eta_i$'s are $1/d$ then the probability is $\frac{\dim(\lambda)^2}{k!} \frac{\rising{d}{\lambda}}{d^{|\lambda|}}$. This is precisely Definition~\ref{def:SW-dist} below.}\rnote{In the first para of Section II of Audenaart he says it was \emph{rediscovered}(!) by Alicki et al.  Perhaps he counts Frobenius as the original discoverer???!}
\begin{equation*}
\Pr[\lambda \mid \rho^{\otimes n}] = \dim(\lambda)\cdot s_\lambda(\eta_1, \eta_2, \ldots, \eta_d).
\end{equation*}
\end{fact}
\noindent
A proof of this fact can be found in~\cite[Equation~(36)]{Aud06}.
}

For the purposes of the testing lower bounds in this paper, the case of greatest interest to us is when $\rho = \frac{1}{d} I_{d \times d}$ is the \emph{maximally mixed $d$-dimensional state}; i.e., the spectrum of~$\rho$ is the uniform distribution $\unif{d} = (\frac1d, \dots, \frac1d)$.  This is also by far the most well-studied case in the literature:
\begin{definition}\label{def:SW-dist}
    The \emph{Schur--Weyl} distribution with parameters $n$ and $d$, which we denote $\SWunif{n}{d}$, is the distribution on partitions $\lambda \vdash n$ of length at most~$d$ given by $\SWdens{n}{\rho}$ in the case that~$\rho$ is the maximally mixed state of dimension~$d$.  Equivalently, it is the distribution of $\rsk(\ba)$, where $\ba \sim [d]^n$ is uniformly random.
\end{definition}
Combining Proposition~\ref{prop:schur-probability} and Proposition~\ref{prop:s111}, together with the homogeneity of the Schur polynomials, we obtain the following known formula (cf.~\cite[equation~(26)]{CHW07}):
\begin{proposition}         \label{prop:SW-dist}
$\displaystyle \SWunifProb{n}{d}{\lambda} = \frac{(\dim \lambda)^2}{n!}\cdot \frac{\rising{d}{\lambda}}{d^n}.$
\end{proposition}
Notice that if $n$ is held fixed and $d \to \infty$, the fraction $\frac{\rising{d}{\lambda}}{d^n}$ tends to~$1$ and we obtain the Plancherel distribution (for~$\symm{n}$) on partitions described in Definition~\ref{def:group-plancherel}.  This recovers the well-known fact that the Plancherel distribution is obtained by running the RSK algorithm on a uniformly random permutation (equivalently, a uniformly random word from $[0,1]^n$). We will write $\Planch{n}$ for this distribution.

\begin{remark}\label{rem:no-rsk-correspondence}
It is easy to see that $\SWunifProb{n}{d}{\lambda} = \frac{1}{d^n}\cdot\dim(\sirrep_\lambda) \cdot \dim(\uirrep_\lambda^d)$.
From Remark~\ref{rem:quantum-free}, we see that there are $\dim(\sirrep_\lambda) \cdot \dim(\uirrep_\lambda^d)$ words $a \in [d]^n$
such that $\rsk(a) = \lambda$.
\end{remark}

\ignore{

\subsubsection{Explicit formulas}
\rnote{get the plancherel discussion into here}
Finally, for some particular choices of the density matrix $\rho$,
it turns out that we have relatively nice explicit formulas for $\Pr\left[\lambda\mid\rho^{\otimes n}\right]$.
First, suppose $\rho = I_d/d$, where $I_d$ is the $d \times d$ identity matrix.  Then $\rho^{\otimes n}$ is $I_{d^n}/d^n$, as is $U \rho^{\otimes n} U^\dagger$.
This means that a measurement of $\rho^{\otimes n}$ in the Schur basis, as done in the first step of weak Schur sampling, will produce a uniformly random one of the $d^n$ Schur basis vectors.
As there are $\dim(\sirrep_\lambda)\cdot\dim(\uirrep_\lambda^d)$ Schur basis vectors labeled by $\lambda$, we have
\begin{equation}
\Pr\left[\lambda\mid \rho^{\otimes k}\right] = \frac{\dim(\sirrep_\lambda)\cdot\dim(\uirrep_\lambda^d)}{d^k}.\label{eq:uniform-distribution}
\end{equation}
More broadly, Fact~\ref{fact:only-nonzero-spectrum} showed that the output distribution of weak Schur sampling is dependent only on the nonzero spectrum of a density matrix. Thus, Equation~\eqref{eq:uniform-distribution} holds for \emph{any} mixed state~$\rho$ satisfying $\rho^2 = \rho/d$.

XXX gives an extremely convenient expression for the dimension of the irreducible representations of $\unitary{d}$:
\begin{theorem}
$\dim(\uirrep_\lambda^d) = \frac{d^n}{n!}\prod_{(i, j)\in \lambda}\left(1 + \frac{j-i}{d}\right)$.
\end{theorem}
\noindent
Plugging this in to Equation~\eqref{eq:uniform-distribution} (and using $\dim(\sirrep_\lambda) = \dim(\lambda)$), we arrive at the main probability distribution this work will study:
\begin{definition}\label{def:SW-dist} \rnote{See an rnote from Prop~\ref{prop:schur-probability}}
The \emph{Schur--Weyl} distribution, denoted $\SWunif{n}{d}$, outputs the partition $\lambda$ with probability
\begin{equation*}
\frac{\dim(\lambda)^2}{n!}\prod_{(i, j) \in \lambda}\left(1 + \frac{j-i}{d}\right).
\end{equation*}
Given a state $\rho$ satisfying $\rho^2 = \rho/d$, this distribution is the same as the distribution of the weak Schur sampling algorithm on the state $\rho^{\otimes k}$.
\end{definition}
\noindent
We note that the first term in this product ($\dim(\lambda)^2/n!$) is exactly the Plancherel distribution on partitions as defined in Definition~\ref{def:group-plancherel}.
Thus, as $d\rightarrow \infty$, $\SWunif{n}{d}$ approaches the Plancherel distribution.
}

\subsection{Asymptotic theory of the symmetric group}

For small $n$, the exact distribution on partitions of $n$ given by the Plancherel or Schur--Weyl distributions is not particularly easy to understand.  As a result, a significant body of work has been devoted to showing asymptotic properties of these distributions as $n$ grows large.

Let us focus first on the Plancherel measure.
Perhaps the most basic thing one could ask for is the ``typical'' width and height of a diagram drawn from this distribution.
Though either of these values could be as large as $n$, Hammersly~\cite{Ham72} showed that both values tend to concentrate around the same number $c\cdot \sqrt{n}$,
for some constant $c$ (later determined to be $c = 2$~\cite{LS77,VK77}).
Therefore, in order to put partitions of different values of~$n$ on equal footing, we can define scaled partitions as follows:
\begin{definition}
Let $\lambda \vdash n$ and recall Definition~\ref{def:russian-function}.  Then $\overline{\lambda}:\R\rightarrow\R_+$ is defined as $\overline{\lambda}(x)\coloneqq \lambda(\sqrt{n}\cdot x)/\sqrt{n}$, for all $x$.
\end{definition}

Logan and Shepp~\cite{LS77} and Vershik and Kerov~\cite{VK77} independently proved the so-called ``law of large numbers'' for the Plancherel distribution, showing that when $\blambda \sim \Planch{n}$ and $n \rightarrow \infty$, the function $\overline{\blambda}$ converges to $\Omega(x)$, the curve defined as
\begin{equation*}
\Omega(x) \coloneqq  \left\{
\begin{array}{ll}
\frac{2}{\pi}(x \arcsin\frac{x}{2} + \sqrt{4-x^2}), & |x|\leq 2,\\
|x| & |x|\geq 2.
\end{array} \right.
\end{equation*}
This ``ice cream cone''-shaped function is pictured in Figure~\ref{fig:biane} ($c = 0$ case).
Though this curve is a limiting shape rather than the Russian notation of any Young diagram,
it is useful to think of it as a continual analogue of a Young diagram, as per the following definition.
\begin{definition}
A \emph{continual diagram} is a function $f:\R\rightarrow \R$ satisfying (i) $f$ is $1$-Lipschitz and (ii) $f(x) = |x|$ when $|x|$ is sufficiently large.
\end{definition}
\noindent
This definition originates in the paper of~\cite{Ker93b}.

More recently, Kerov~\cite{Ker93a} showed a ``central limit theorem'' for the Plancherel measure,
characterizing the deviation of a random Young diagram from the curve $\Omega(x)$ by a certain Gaussian process.
A second proof of this result, also by Kerov, was given in the paper of Ivanov and Olshanski~\cite{IO02}.
Much of our work is based on the techniques of this paper.

Subsequent studies revealed that a similar state of affairs exists for the Schur--Weyl $\SWunif{n}{d}$ distribution, though in this case the features of a ``typical'' $\blambda \sim \SWunif{n}{d}$ depend on the ratio $c\coloneqq  \frac{\sqrt{n}}{d}$. Biane~\cite{Bia01} extended the Plancherel law of large numbers to the Schur--Weyl distribution in the case when~$c$ is a fixed constant and $n, d \to \infty$. In this case, for a random $\blambda \sim \SWunif{n}{d}$, the function $\overline{\blambda}$ will approach a certain limiting curve $\Omega_c$, specified as follows:
\begin{theorem}[\cite{Bia01}]\label{thm:biane}
Fix an absolute constant $c > 0$ and assume $n, d \to \infty$ with $\frac{\sqrt{n}}{d}\rightarrow c$. Then
\begin{equation*}
\Pr_{\blambda \sim \SWunif{n}{d}}\left[\Vert \overline{\blambda} - \Omega_c\Vert_{\infty} \geq \eps\right] \rightarrow 0,
\end{equation*}
where $\Omega_c$ is the continual diagram defined as follows:
\begin{align*}
&\Omega_0(x) = \Omega(x);\\
&\Omega_{c \in (0, 1)}(x)
	=\left\{
		\begin{array}{ll}
			\frac{2}{\pi}\left(x \arcsin(\frac{x+c}{2\sqrt{1+cx}})
				+\frac{1}{c}\arccos(\frac{2+cx - c^2}{2\sqrt{1+cx}})
				+\frac{\sqrt{4-(x-c)^2}}{2}\right) & \text{if $|x-c| \leq 2$},\\
			|x| & \text{otherwise};
		\end{array}\right.\\
&\Omega_{c =1}(x)
	=\left\{
		\begin{array}{ll}
			\frac{x+1}{2} +\frac{1}{\pi}\left((x-1)\arcsin(\frac{x-1}{2})
				+\sqrt{4-(x-1)^2}\right) & \text{if $|x-1| \leq 2$},\\
			|x| & \text{otherwise};
		\end{array}\right.\\
&\Omega_{c >1}(x)
	=\left\{
		\begin{array}{ll}
			x+\frac{2}{c}&\text{if $x \in (\frac{-1}{c},c-2)$}\\
			\frac{2}{\pi}\left(x \arcsin(\frac{x+c}{2\sqrt{1+cx}})
				+\frac{1}{c}\arccos(\frac{2+cx - c^2}{2\sqrt{1+cx}})
				+\frac{\sqrt{4-(x-c)^2}}{2}\right) & \text{if $|x-c| \leq 2$},\\
			|x| & \text{otherwise}.
		\end{array}\right.
\end{align*}
\end{theorem}
\noindent
These curves are pictured for various values of~$c$ in Figure~\ref{fig:biane} (which we have reproduced from~\cite{Mel10a}).  Meliot~\cite{Mel10a,Mel10b} has extended Kerov's central limit theorem to the Schur--Weyl distribution, characterizing the fluctuations of $\overline{\blambda}$ around the limiting curves given by Biane.

\begin{figure}
\centering
\begin{subfigure}{.45\textwidth}
	\centering
	\includegraphics[height=100pt]{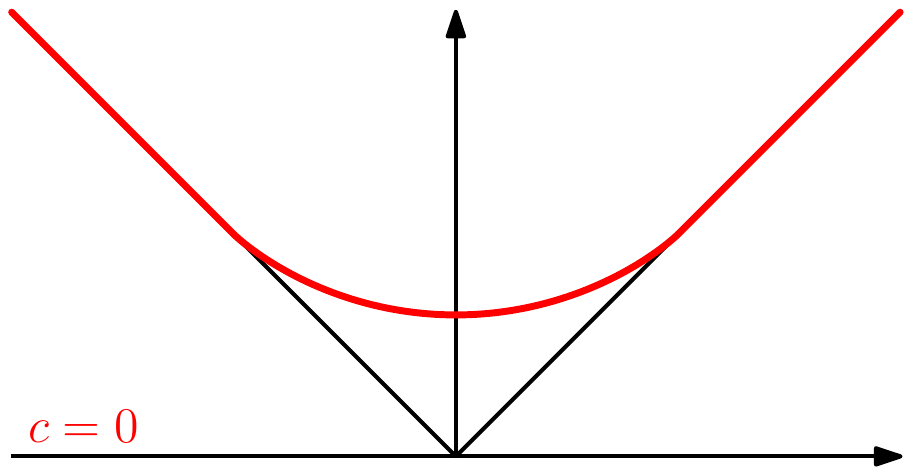}
\end{subfigure}
\begin{subfigure}{.45\textwidth}
	\centering
	\includegraphics[height=100pt]{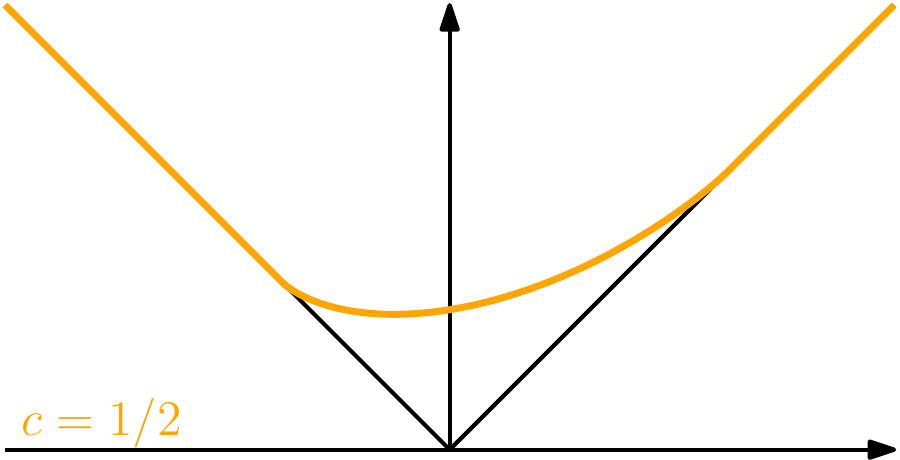}
\end{subfigure}\\
\vspace{5mm}
\begin{subfigure}{.45\textwidth}
	\centering
	\includegraphics[height=100pt]{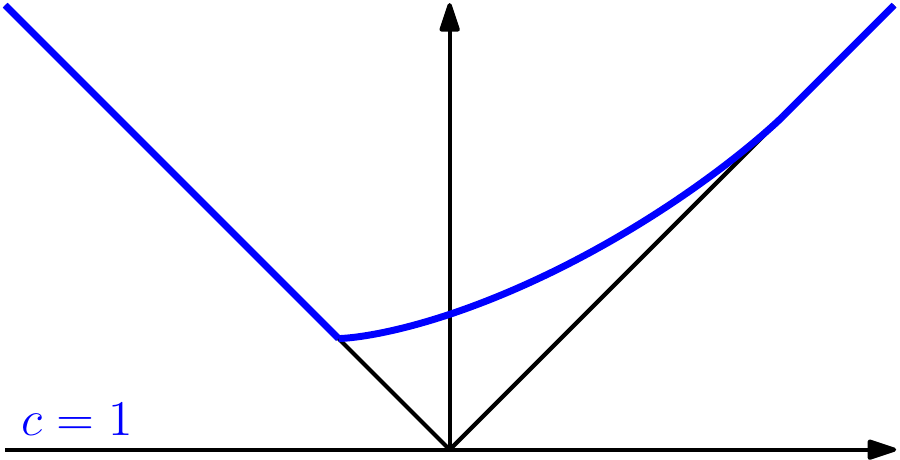}
\end{subfigure}
\begin{subfigure}{.45\textwidth}
	\centering
	\includegraphics[height=100pt]{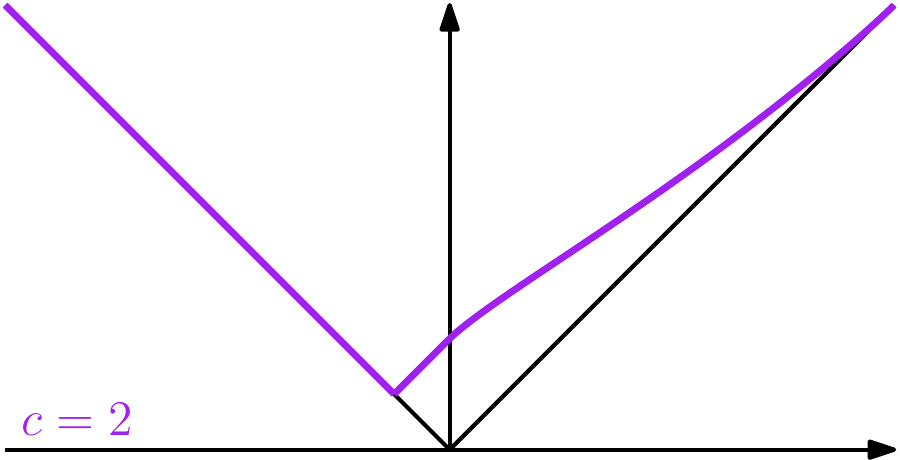}
\end{subfigure}
\caption{The Biane limiting curves $\Omega_c$.  The $c = 0$ case corresponds to the function $\Omega(x)$.}
\label{fig:biane}
\end{figure}

One consequence is these results is that when $n = o(d^2)$, the function $\overline{\blambda}$ converges to the ice cream cone curve $\Omega(x)$ from above. This fact is a manifestation of the discussion at the end of Section~\ref{sec:understanding} concerning $\SWunif{n}{d}$ tending to $\Planch{n}$ as $d\rightarrow\infty$.  Indeed, Childs et al.~\cite{CHW07} showed that when $n = o(d)$, the two distributions are statistically indistinguishable (from which the lower bound in Theorem~\ref{thm:q-bday} follows via the triangle inequality
$\dtv{\SWunif{n}{r}}{\SWunif{n}{2r}}\leq\dtv{\SWunif{n}{r}}{\Planch{n}}+\dtv{\Planch{n}}{\SWunif{n}{2r}}$).

We close this section by recording some simple concentration bounds on the width and length of $\blambda \sim \SWunif{n}{d}$.  They are not as precise as what is suggested by the above limit theorems, but they have the advantage of giving concrete error bounds.  We follow a simple line of argument similar to that in~\cite[Lemma~1.5]{Rom14}.
\begin{proposition}\label{prop:increasing-seq}
Let $\blambda \sim \SWunif{n}{d}$.  For every $B \in \Z^+$ we have $\Pr[\blambda_1 \geq B] \leq \left(\frac{(1+B/d)e^2 n}{B^2}\right)^B$.
The same bound holds for $\Pr[\blambda_1' \geq B]$.
\end{proposition}
\begin{proof}
By Theorem~\ref{thm:greene}, $\Pr[\blambda_1 \geq B]$ (respectively, $\Pr[\blambda'_1 \geq B]$) is equal to the probability that a uniformly random word from $[d]^n$ contains a weakly increasing (respectively, strongly increasing) subsequence of length exactly~$B$.  As weakly increasing subsequences are more probable than strongly increasing ones, it suffices to bound
    \[
        \Pr[\blambda_1 \geq B] \leq \left(\frac{(1+B/d)e^2 n}{B^2}\right)^B.
    \]
    Letting $\bS$ denote the number of weakly increasing subsequences of length $B$ in a random word we have
    \[
        \Pr[\blambda_1 \geq B] \leq \E[\bS] = \binom{n}{B} \cdot \frac{c}{d^B},
    \]
    where $c$ is the number of words in $[d]^{B}$ which are weakly increasing. Evidently~$c$ also equals the number of ``weak $d$-compositions of~$B$'', which~\cite[Chapter 1.2]{Sta11} is~$\binom{d- 1 + B}{B} \leq \binom{d + B}{B}$.  We conclude
    \[
        \Pr[\blambda_1 \geq B] \leq \binom{n}{B} \cdot \frac{\binom{d + B}{B}}{d^{B}} \leq \frac{\left(\frac{en}{B}\right)^B \left(\frac{(1+B/d)ed}{B}\right)^B}{d^B} = \left(\frac{(1+B/d)e^2 n}{B^2}\right)^B,
    \]
    as needed.
\end{proof}

\subsection{Polynomial algebras}               \label{sec:polynomial-algebras}

We have already discussed the power sum and Schur polynomials, which are elements of the $\C$-algebra~$\Lambda$ of symmetric polynomials in indeterminates $x_1, x_2, \dots$.\footnote{Strictly speaking, these are \emph{families} of bounded-degree polynomials, one for each number of indeterminates, which are \emph{stable} in the sense that $p_\lambda(x_1, \dots, x_d, 0) = p_\lambda(x_1, \dots, x_d)$, and similarly for $s_\lambda$.  See, e.g.,~\cite{Mac95} for a formal definition via projective limits.}  Important to our work will be a closely related polynomial algebra~$\Lambda^*$, the algebra of \emph{shifted symmetric} polynomials, formally introduced introduced in~\cite{OO98b}.  This algebra consists of those polynomials which are symmetric in the ``shifted'' indeterminates $\wt{x}_i \coloneqq  x_i - i + c$, where $c$ is any fixed constant.  (The definition does not depend on the constant~$c$.)  When we view the inputs to the shifted symmetric functions $x_1, x_2, \ldots$ as the values $\lambda_1, \lambda_2, \ldots$ of a partition~$\lambda$, the result is (isomorphic to) \emph{Kerov's algebra} of polynomial functions on the set of Young diagrams, also known as the \emph{algebra of observables of diagrams}.  In a nutshell, the importance of this algebra is that, on one hand, it still contains polynomials that are similar to ``power sums'' or ``moments'' of the $\lambda_i$'s; and, on the other hand, it is easier to compute their expected value under $\SWdens{n}{\rho}$ distributions.


We will need to study several families of observables/shifted symmetric polynomials, and their relationships:
\begin{definition}
The following polynomials are known to be elements of $\Lambda^*$.  (We describe the first four as observables of Young diagrams.)
\begin{itemize}
    \item For $k \geq 1$,
            \begin{equation*}
                \pstar{k}(\lambda) \coloneqq  \sum_{i = 1}^{d(\lambda)} \Bigl((a^*_i)^k - (-b^*_i)^k \Bigr) = \sum_{i=1}^\infty\Bigl((\lambda_i - i + \tfrac{1}{2})^k - (-i + \tfrac{1}{2})^k\Bigr).
            \end{equation*}
            These are the most basic polynomials on Young diagrams, giving the ``moments'' of the coordinates.  For more information on them see~\cite{IO02}, where they are introduced (in equation~(1.4)) under the notation $p_k(\lambda)$.  We use the notation $\pstar{k}(\lambda)$ to distinguish them from the ordinary power sum symmetric polynomials.  It is obvious from the second definition above that the $\pstar{k}$ polynomials are in~$\Lambda^*$.  In fact they are algebraically independent, and they generate~$\Lambda^*$.
    \item For $k \geq 0$, the $k$th \emph{content sum} polynomial is  $c_k(\lambda) \coloneqq  \sum_{\square \in [\lambda]} c(\square)^k$.  Although these polynomials are quite natural, we will have little occasion to use them.  The fact that they are in~$\Lambda^*$ was proven in~\cite{KO94}.
    \item For $k \geq 2$,
        \begin{equation*}
            \widetilde{p}_k(\lambda) \coloneqq  k(k-1) \int_{-\infty}^\infty x^{k-2}\sigma(x)\,dx,
        \end{equation*}
        where $\sigma(x) \coloneqq  \frac{1}{2}(\lambda(x) - |x|)$.  These polynomials were introduced and shown to be algebraically independent generators of~$\Lambda^*$ in~\cite[Section~2]{IO02}.  They can shown to be the ``moments of the local extrema of $\lambda(x)$'', and are also useful for studying continual diagrams.  We use them only briefly, to pass between the $\pstar{k}$ polynomials and $\psharp{k}$ polynomials defined below.
    \item For $\lambda \vdash n$ and $\mu \vdash k$, the \emph{central characters} are defined by
        \begin{equation*}
        \psharp{\mu}(\lambda) = \begin{cases}
                            \falling{n}{k}\cdot \frac{\chi_\lambda(\mu \cup 1^{n-k})}{\dim(\lambda)} & \text{if }n \geq k,\\
                            0 & \text{if }n < k.
                                 \end{cases}
        \end{equation*}
        where $\mu \cup 1^{n-k}$ denotes the partition $(\mu, 1, 1, \dots, 1) \vdash n$.  In case $\mu = (k)$ we simply write $\psharp{k}(\lambda)$. Note that we are somewhat unexpectedly applying the character $\chi_\lambda$ to (an extension of)~$\mu$, and not the other way around.  The advantage of the $\psharp{\mu}$ polynomials is that, by virtue of them being characters of the symmetric group (up to some normalizations), their expectations under $\SWdens{n}{\rho}$ can be easily calculated exactly, as we will see below.  A disadvantage is that, by virtue of them being characters of the symmetric group, explicit formulas for them are famously quite complex~\cite{Las08,Fer10} (though in Section~\ref{sec:work-it}\ignore{the proof of our Lemma~\ref{lem:errL2}\rnote{update this pointer}} we will mention a formula that allows one to compute $\psharp{k}$ for small~$k$ fairly easily). Wassermann~\cite[III.6]{Was81} showed that the $\psharp{k}$ polynomials are in~$\Lambda^*$, and in fact~\cite{VK81,KO94,OO98b} more generally the polynomials $\psharp{\mu}$ form a \emph{linear} basis of~$\Lambda^*$.
    \item For $\mu \vdash k$, the \emph{shifted Schur} polynomial in indeterminates $x_1, \dots, x_d$ is
        \[
            s^*_\mu(x_1, \dots, x_d) = \frac{\det\Bigl(\falling{(x_i -i + d)}{(d+\lambda_j -j)}\Bigr)_{ij}}{\det\Bigl(\falling{(x_i-i+d)}{(d-j)}\Bigr)_{ij}} \quad \text{if $\ell(\mu) \leq d$, else~$0$.}
        \]
        These polynomials are the shifted analogues of the Schur polynomials (cf.~Theorem~\ref{thm:schur-determinant}).  They were introduced by Okounkov and Olshanski~\cite{OO98b}, and are similar to the earlier-defined ``factorial Schur functions'' (see, e.g.,~\cite[I.3.20--21]{Mac95}), but with the advantage that they are \emph{stable}---i.e., $s^*_\mu(x_1, \dots, x_d, 0) = s^*_\mu(x_1, \dots, x_d)$.  They arise for us because they can sometimes be used to express the ratio of two Schur functions (see the ``Binomial Formula'' Theorem~\ref{thm:binomial-formula}).  To analyze them, we will use the following ``shifted analogue'' of Theorem~\ref{thm:power-schur-relation}, proved in~\cite[Theorem~8.1]{OO98b},~\cite[Theorem~9.1]{IK01} (see also~\cite[p.25]{Mel10b}):
        \begin{theorem}                                     \label{thm:psharp-ssharp-relation}
            For $\mu \vdash k$, let us think of the central character polynomial~$\psharp{\mu}$ not as an observable of Young diagrams (applied to $\lambda_1, \dots, \lambda_d$) but as a shifted symmetric polynomial in indeterminates $x_1, \dots, x_d$.  In the context of Fourier analysis over the group $G = \symm{k}$, for each fixed $x \in \C^d$ we may think of $\psharp{(\cdot)}(x) \coloneqq \pi \mapsto \psharp{\pi}(x)$ as a class function.  Then its Fourier coefficients are given by
            \[
                \four{\psharp{(\cdot)}(x)}(\mu) = s^*_\mu(x).
            \]
        \end{theorem}
        (Note that give the determinantal definition of the shifted Schur polynomials, one may alternatively take this Theorem as a definition of the shifted symmetric polynomials $\psharp{\mu}(x)$.)
\end{itemize}
\end{definition}

As mentioned, the $\psharp{\mu}$ polynomials are especially important for us as because there is a simple expression for their expectation under any Schur--Weyl distribution.  This is the subject of our next proposition.
\begin{proposition}\label{prop:sw-alg-expect}
Let $\rho$ be a $d \times d$ density matrix with eigenvalues $\eta_1, \ldots, \eta_d$, and let $\mu \vdash k$.  Then
\[
    \E_{\blambda \sim \SWdens{n}{\rho}} [\psharp{\mu}(\blambda)] =   \falling{n}{k} \cdot p_\mu(\eta_1, \ldots, \eta_d).
\]
\end{proposition}
\begin{proof}
It's immediate from the definitions that both sides are~$0$ if $n < k$, so we assume $n \geq k$.  Applying Proposition~\ref{prop:schur-probability} and the definition of $\psharp{\mu}$ we obtain
\begin{align*}
    \E_{\blambda \sim \SWdens{n}{\rho}} [\psharp{\mu}(\blambda)] &= \falling{n}{k} \cdot \sum_{\lambda \vdash n} s_\lambda(\eta_1, \dots, \eta_d) \cdot \chi_\lambda(\mu \cup 1^{n-k}) \\
    &= \falling{n}{k} \cdot p_{\mu \cup 1^{n-l}}(\eta_1, \dots, \eta_d),
\end{align*}
where the second equation is from Theorem~\ref{thm:power-schur-relation}.  But $p_{\mu \cup 1^{n-l}}(\eta_1, \dots, \eta_d) = p_{\mu}(\eta_1, \dots, \eta_d)$, since the two quantities differ only by factors of $p_1(\eta_1, \dots, \eta_d) = \eta_1 + \cdots + \eta_d = 1$.
\end{proof}

Note that in the case of $\eta_1 = \ldots = \eta_d = 1/d$, we have that $p_\mu(\eta_1, \ldots, \eta_d) = d^{\ell(\mu) - k}$.  This gives us the following important corollary:
\begin{corollary}\label{cor:sw-expect}
Let $\mu \vdash k$.  Then $\displaystyle \E_{\blambda \sim \SWunif{n}{d}} [\psharp{\mu}(\blambda)] = n^{\downarrow k} \cdot d^{\ell(\mu) - k}$.
\end{corollary}

\subsubsection{Working with the $\psharp{\mu}$ polynomials}         \label{sec:work-it}

As we will be working heavily with the $\psharp{\mu}$ polynomials, let us describe them further.  We begin with the simpler case of the $\psharp{k}$ polynomials.  Let us see how these polynomials can be written in terms of the $\pstar{k}$ polynomials.  From~\cite[III.6]{Was81} (cf.~\cite[Proposition~3.3]{IO02}) we have the following identity using generating functions:
\[
    \psharp{k} = [t^{k+1}]\left\{-\frac1k \prod_{j=1}^k (1-(j-\tfrac12)t)\cdot \exp\left(\sum_{j=1}^\infty \frac{\pstar{j} t^j}{j}(1-(1-kt)^{-j})\right) \right\}.
\]
One may rewrite this (cf.~\cite[(3.3)]{IO02}) as
\begin{equation}                            \label{eqn:crazy-gen}
    \psharp{k} = [t^{k+1}]\left\{-\frac1k \prod_{j=1}^k (1-(j-\tfrac12)t)\cdot \sum_{i=0}^\infty \frac{(-1)^i}{i!} Q_k(t)^i\right\},
\end{equation}
where
\begin{equation}                            \label{eqn:crazy-gen2}
    Q_k(t) = \sum_{m = 1}^{\infty} Q_{k,m} t^{m+1}, \quad Q_{k,m} = \tfrac{1}{1}\tbinom{m}{0} k^m \pstar{1} + \tfrac{1}{2}\tbinom{m}{1}k^{m-1}\pstar{2} + \tfrac{1}{3}\tbinom{m}{2}k^{m-2}\pstar{3} + \cdots + \tfrac{1}{m}\tbinom{m}{m-1}k \pstar{m}.
\end{equation}
It follows that in~\eqref{eqn:crazy-gen} we may restrict the sum on~$i$ to the range between $0$~and~$\frac{k+1}{2}$, and in~\eqref{eqn:crazy-gen2} we can restrict the sum on~$m$ to the range between $1$~and~$k$.  We thereby obtain a relatively simple finitary method for expressing $\psharp{k}$'s polynomials in terms of $\pstar{j}$'s.  In particular, we can deduce
\begin{equation}                    \label{eqn:first-few-psharps}
    \psharp{1} = \pstar{1}, \qquad \psharp{2} = \pstar{2}, \qquad \psharp{3} = \pstar{3} -\tfrac32 (\pstar{1})^2 + \tfrac54\pstar{1}, \qquad \psharp{4} = \pstar{4} - 4\pstar{2}\pstar{1} +\tfrac{11}{2}\pstar{2}.
\end{equation}

\noindent As observed in~\cite[Proposition~3.4]{IO02}, we can also deduce that in general,
\begin{equation}                            \label{eqn:invertible}
    \psharp{k} = \pstar{k} + \Bigl\{\text{polynomial in }\pstar{1}, \dots, \pstar{k-1} \text{ of gradation at most $k-1$}\Bigr\},
\end{equation}
where \emph{gradation} refers to the canonical grading in which $\prod_{i} \pstar{\lambda_i}$ has gradation~$|\lambda|$.  We can of course inductively invert this relationship, deducing that
\begin{equation}                            \label{eqn:inverted}
    \pstar{k} = \psharp{k} + \Bigl\{\text{polynomial in }\psharp{1}, \dots, \psharp{k-1} \text{ of gradation
     at most $k-1$}\Bigr\}.
\end{equation}
For example,
\begin{equation}                    \label{eqn:first-few-pstars}
    \pstar{1} = \psharp{1}, \qquad \pstar{2} = \psharp{2}, \qquad \pstar{3} = \psharp{3} +\tfrac32 (\psharp{1})^2 -\tfrac54 \psharp{1}, \qquad \pstar{4} = \psharp{4} +4 \psharp{2}\psharp{1} - \tfrac{11}{2}\psharp{2}.
\end{equation}

Recall that the more general $\psharp{\tau}$ polynomials (for $\tau \in \allpartitions$) are known to linearly generate the algebra of observables.  This means that any product $\psharp{\mu_1} \psharp{\mu_2}$ can be converted to a linear combination of $\psharp{\tau}$'s.  In particular, if we applied this conversion in~\eqref{eqn:first-few-pstars} we would get linear expressions for the ``low-degree moments of Young diagrams'' (i.e., the $\pstar{j}$'s) in terms of $\psharp{\tau}$'s; we could then compute the expectation of these, under any Schur--Weyl distribution, using Proposition~\ref{prop:sw-alg-expect}.

We are therefore interested in the \emph{structure constants} $f^{\tau}_{\mu_1\mu_2}$ of $\Lambda^*$ in the basis $\{\psharp{\tau}\}$; i.e., the numbers such that
\[
    \psharp{\mu_1} \psharp{\mu_2} = \sum_{\tau \in \allpartitions} f^{\tau}_{\mu_1 \mu_2} \psharp{\tau}.
\]
These were first determined by Ivanov and Kerov~\cite{IK01} in terms of the algebra of \emph{partial permutations}.  We quote the following formulation from~\cite[Proposition~4.5]{IO02}:
\begin{proposition}\label{prop:structure-coefficients}
    Let $\tau,\mu_1,\mu_2 \in \allpartitions$. Fix a set $R$ of cardinality $|\tau|$ and a permutation $w:R\rightarrow R$ of cycle type~$\tau$. Then
    \begin{equation*}
        f_{\mu_1 \mu_2}^\tau = \frac{z_{\mu_1} z_{\mu_2}}{z_\tau} g^\tau_{\mu_1 \mu_2},
    \end{equation*}
    where $g^\tau_{\mu_1 \mu_2}$ equals the number of quadruples $(R_1, w_1, R_2, w_2)$ such that:
    \begin{enumerate}
    \item $R_1\subseteq R, \quad R_2 \subseteq R,  \quad R_1 \cup R_2 = R$;
    \item $|R_i| = |\mu_i|$ and $w_i : R_i\rightarrow R_i$ is a permutation of cycle type $\mu_i$, for $i = 1,2$;
    \item $\overline{w}_1 \overline{w}_2 = w$, where $\overline{w}_i : R\rightarrow R$ denotes the natural extension of $w_i$ from $R_i$ to the whole of~$R$.
    \end{enumerate}
\end{proposition}
We present an equivalent formulation we have found to be more convenient.  We omit its straightforward combinatorial deduction from Proposition~\ref{prop:structure-coefficients}.
\begin{corollary}\label{cor:structure-coefficients-rewrite}
    Let
    \[
        C^{t}_{r_1r_2} \coloneqq \frac{r_1!r_2!}{(t-r_1)!(t-r_2)!(r_1+r_2-t)!}
    \]
    if the positive integers $r_1,r_2,t$ satisfy $r_1, r_2 \leq t \leq r_1 + r_2$, and let $C^t_{r_1r_2} \coloneqq 0$ otherwise. Then for $\mu \vdash r_1$, $\nu \vdash r_2$, $\tau \vdash t$,
    \[
        f_{\mu \nu}^\tau = C^{t}_{r_1r_2} \cdot
            \Pr_{\bw_1, \bw_2}\left[\overline{\bw}_1 \overline{\bw}_2 \textnormal{ has cycle type } \tau\right],
    \]
    where $\bw_1$ is a uniformly random permutation on $\{1, \ldots, r_1\}$ of cycle type $\mu$, and $\bw_2$ is a uniformly random permutation on $\{t-r_2+1, \ldots, t\}$ of cycle type~$\nu$.
\end{corollary}
As very simple examples, we can compute
\begin{equation}                            \label{eqn:simple-psharp-products}
    (\psharp{1})^2 = \psharp{(1,1)} + \psharp{1}, \qquad \psharp{2}\psharp{1} = \psharp{(2,1)} + 2\psharp{2}, \qquad (\psharp{2})^2 = \psharp{(2,2)} + 4 \psharp{3} + 2\psharp{(1,1)}.
\end{equation}
Substituting these into~\eqref{eqn:first-few-psharps}, we obtain the formulas
\begin{equation}                \label{eqn:low-moments}
    \pstar{1} = \psharp{1}, \qquad \pstar{2} = \psharp{2}, \qquad \pstar{3} = \psharp{3} +\tfrac32\psharp{(1,1)} + \tfrac14\psharp{1}, \qquad \pstar{4} = \psharp{4} +4\psharp{(2,1)}+\tfrac52 \psharp{2},
\end{equation}
which will be useful to us later.

Given the formula for the structure constants, it's not hard to show that
\[
    \psharp{\mu} \psharp{\nu} = \psharp{\mu \cup \nu} + \Bigl\{\text{linear combination of $\psharp{\tau}$'s with $|\tau| < |\mu \cup \nu|$}\Bigr\},
\]
where $\mu \cup \nu$ denotes the partition formed by joining the parts of $\mu$ and $\nu$ and sorting them in nonincreasing order (i.e., $m_w(\mu \cup \nu) = m_w(\mu) + m_w(\nu)$).  In fact, we will require a stronger statement, based on the following notion introduced in~\cite{IK01}:
\begin{definition}
    For a partition $\lambda \in \allpartitions$, its \emph{weight} is defined to be $\weight(\lambda) = |\lambda| + \ell(\lambda)$.
\end{definition}
Now \'{S}niady~\cite[Corollary~3.8]{Sni06} proved:
\begin{proposition}                                     \label{prop:should-be-in-a-structure-constants-discussion-or-theorem}
    $\displaystyle
        \psharp{\mu} \psharp{\nu} = \psharp{\mu \cup \nu} + \Bigl\{ \textnormal{linear combination of $\psharp{\tau}$'s with $\weight(\tau) \leq \weight(\mu) + \weight(\nu) - 2$}\Bigr\}.
    $
\end{proposition}

\section{The empirical Young diagram algorithm}

The empirical Young diagram (EYD) algorithm works as follows:
\begin{nameddef}{The EYD algorithm}
Given $\rho^{\otimes n}$:
\begin{enumerate}
\item Sample $\blambda \sim \SWdens{n}{\rho}$.
\item Output $\underline{\blambda}\coloneqq (\blambda_1/n, \ldots, \blambda_d/n)$.
\end{enumerate}
\end{nameddef}
\noindent
This algorithm has, either implicitly or explicitly, arisen in several independent research threads.
The first was the work of Alicki, Rudnicki, and Sadowski~\cite{ARS88}, who showed that if $\rho$ has eigenvalues
$\eta_1 \geq \ldots \geq \eta_d$, then $\underline{\blambda}\rightarrow \eta$ as $n\rightarrow \infty$, and furthermore sketched a central limit theorem for the fluctuations.
Ten years later, Keyl and Werner~\cite{KW01} independently reproved the first part of this result (and showed an ``error rate'' for the EYD algorithm which, for any fixed~$d$, decreases exponentially in~$n$); they also explicitly suggested the EYD algorithm for spectrum estimation.  Further independent work, developing the research on the ``Gaussian Unitary Ensemble'' nature of the fluctuations, was performed by Its--Tracy--Widom, Houdr\'{e} and coauthors, and others~\cite{ITW01,Lit08,HX13}

\subsection{The upper bound}

Following Keyl and Werner's paper~\cite{KW01}, a short, simplified proof of correctness containing explicit error bounds was discovered in~\cite{HM02}.
A small bug in their derivation was corrected by~\cite{CM06}, whose Corollary~$1$ states:
\begin{theorem}\label{thm:kw-prelim}
Let $\rho$ be a mixed state with eigenvalues $\eta_1 \geq \ldots \geq \eta_d$.
Let $S$ be any set of partitions of $n$, and set $d_{\mathrm{KL}} := \min_{\lambda \in S} \dkl{\underline{\lambda}}{\eta}$.  Then
\begin{equation*}
\Pr_{\blambda \sim \SWdens{n}{\rho}}[\blambda \in S]
\leq (n+1)^{d(d+1)/2}\cdot e^{-n \cdot d_{\mathrm{KL}}}.
\end{equation*}
\end{theorem}
\noindent
If we apply Theorem~\ref{thm:kw-prelim} with the set of partitions $S = \{\lambda \vdash n\mid \dtv{\underline{\lambda}}{\eta} > \eps\}$
and use Pinsker's inequality, we get the following corollary:
\begin{corollary}\label{cor:kw-bound}
Let $\rho$ be a mixed state with eigenvalues $\eta_1 \geq \ldots \geq \eta_d$. Then
\begin{equation*}
\Pr_{\blambda \sim \SWdens{n}{\rho}}[\dtv{\underline{\blambda}}{\eta}>\eps]
\leq (n+1)^{d(d+1)/2}\cdot e^{-2n\eps^2}.
\end{equation*}
In particular, $O(d^2/\eps^2) \cdot \log(d/\eps)\cdot\log(1/\delta)$ samples are sufficient
to output an estimate $\underline{\lambda}$ satisfying $\dtv{\underline{\lambda}}{\eta} \leq \eps$ with probability at least $1-\delta$.
\end{corollary}
\noindent
This means that any unitarily invariant property of mixed states is testable with $O(d^2/\eps^2)\cdot \log(d/\eps)$ copies.

We now give a simplified proof of Theorem~\ref{thm:kw-prelim}.
This will largely follow the outline of the proof found in~\cite{HM02,CM06}, except we will reinterpret their majorizing step in light of the RSK algorithm.
\begin{proof}[Proof of Theorem~\ref{thm:kw-prelim}]
Define the probability distribution $\calD = (\eta_1, \ldots, \eta_d)$.
For a fixed partition $\lambda \in S$,
Remark~\ref{rem:quantum-free} shows that upper-bounding $\SWdensProb{n}{\rho}{\lambda}$
is equivalent to upper-bounding $\Pr_{\ba \sim \calD^{\otimes n}}[\rsk(\ba) = \lambda]$.
By Proposition~\ref{prop:majorizer}, $\rsk(a) = \lambda$ only if $\lambda$ \emph{majorizes} $c(a)$.

By Remark~\ref{rem:no-rsk-correspondence}, there are exactly $\dim(\sirrep_\lambda)\cdot \dim(\uirrep_\lambda^d)$
words $a \in [d]^n$ for which $\rsk(a) = \lambda$.
By the majorizing step, the probability that such an $a$ is drawn from $\calD^{\otimes n}$ is
\begin{equation*}
\prod_i \eta_i^{c_i(a)} \leq \prod_i \eta_i^{\lambda_i}.
\end{equation*}
From this point on, the rest of the argument is as in~\cite{HM02,CM06}.
Recall the well-known upper bounds (cf.~\cite[Equations (1.21) and (1.22)]{Chr06})
\begin{equation*}
\dim(\sirrep_\lambda) \leq\frac{n!}{\prod_i \lambda_i !},
\qquad
\dim(\uirrep_\lambda^d)\leq (n+1)^{d(d-1)/2}.
\end{equation*}
Thus, we can upper-bound $\Pr_{\ba \sim \calD^{\otimes n}}[\rsk(\ba) = \lambda]$ by
\begin{equation*}
(n+1)^{d(d-1)/2}\cdot \frac{n!}{\prod_i \lambda_i !}\cdot \prod_i \eta_i^{\lambda_i}
\leq (n+1)^{d(d-1)/2}\cdot \exp(-n \cdot \dkl{\underline{\lambda}}{\eta}).
\end{equation*}
To recover Theorem~\ref{thm:kw-prelim}, we now union bound over all $\lambda \in S$, of which there are at most $(n+1)^d$.
\end{proof}

\subsection{The lower bound}

Our main result of this section is that Corollary~\ref{cor:kw-bound} is nearly tight, even when~$\rho$ is the maximally mixed state.
In particular, we show the following lower bound:
\begin{theorem}\label{thm:gw-lower}
There is a $\delta > 0$ such that for sufficiently small values of $\eps$,
\begin{equation*}
\Pr_{\blambda \sim \SWunif{n}{d}}[\dtv{\underline{\blambda}}{\unif{d}} > \eps] \geq \delta
\end{equation*}
unless $n = \Omega(d^2/\eps^2)$.
\end{theorem}
We will split the lower bound into two cases.

\begin{theorem}\label{thm:small-n}
For every constant $C > 0$, there are constants $\delta, \eps >0$ such that
\begin{equation*}
\Pr_{\blambda \sim \SWunif{n}{d}}[\dtv{\underline{\blambda}}{\unif{d}} > \eps] \geq \delta
\end{equation*}
when $n < C d^2$ and $d$ is sufficiently large.
\end{theorem}

\begin{theorem}\label{thm:large-n}
There are absolute constants $C >0$ and $0 < \delta < 1$ such that
\begin{equation*}
\Pr_{\blambda \sim \SWunif{n}{d}}[\dtv{\underline{\blambda}}{\unif{d}} > \eps] \geq \delta
\end{equation*}
when $n \geq C d^2$, unless $n = \Omega(d^2/\eps^2)$.
\end{theorem}

To prove Theorem~\ref{thm:gw-lower}, let $C$ and $\delta_1$ be the constants in Theorem~\ref{thm:large-n}.
Apply Theorem~\ref{thm:small-n} with the value of $C$, and let $\delta_2$ and $\eps_0$ be the resulting constants.
Set $\delta := \min\{\delta_1, \delta_2\}$.
Then we see that for all $\eps \leq \eps_0$,
\begin{equation*}
\Pr_{\blambda \sim \SWunif{n}{d}}[\dtv{\underline{\blambda}}{\unif{d}}) > \eps] \geq \delta
\end{equation*}
unless $n = \Omega(d^2/\eps^2)$, giving Theorem~\ref{thm:gw-lower}.

Theorem~\ref{thm:small-n} might appear somewhat superfluous, as Theorem~\ref{thm:large-n} already proves the lower bound for sufficiently large values of $n$ (i.e., $n \geq Cd^2$), and intuitively having fewer copies of~$\rho$ shouldn't improve the performance of the EYD algorithm.
However, this intuition, though it may be true in some approximate sense, is false in general: there are regimes of state estimation where the performance of the EYD algorithm does \emph{not} increase monotonically with the value of $n$.
For example, if $n$ is a multiple of $d$, then when $\blambda \sim \SWunif{n}{d}$,
$\underline{\blambda}$ will equal $\unif{d}$ with some nonzero probability.
On the other hand, a random $\blambda \sim \SWunif{n+1}{d}$ will never be uniform, because $n+1$ is not a multiple of~$d$.
Thus, decreasing the value of $n$ can sometimes help (according to some performance metrics),
and this shows why we need Theorem~\ref{thm:small-n} to supplement Theorem~\ref{thm:large-n}.

The proof of Theorem~\ref{thm:small-n} is quite technical, and we defer it to Section~\ref{sec:eyd-cont}.
Our proof of Theorem~\ref{thm:large-n} is simpler and appears below.  It is a good illustration of the basic technique of using polynomial functions on Young diagrams.  The intuition behind the proof is as follows:  By the (traceless) Gaussian Unitary Ensemble fluctuations predicted in~\cite{ITW01}, we expect that for $\blambda \sim \SWunif{n}{d}$, the empirical distribution $\underline{\blambda}$ will deviate from $\unif{d}$ by  roughly~$\Theta(1/\sqrt{n})$ in each coordinate. This will yield total variation distance $\Theta(d/\sqrt{n})$, necessitating $n \geq \Omega(d^2/\eps^2)$ to achieve $\dtv{\underline{\blambda}}{\unif{d}} \leq \eps$. Actually analyzing the precise rate of convergence to Gaussian fluctuations in terms of~$n$ is difficult, and is overkill anyway; instead, we use the Fourth Moment Method to lower bound the fluctuations.

\begin{proof}[Proof of Theorem~\ref{thm:large-n}]
Our goal is to show that for $n\geq 10^{10} d^2$,
with 1\% probability over a random $\blambda \sim \SWunif{n}{d}$, at least $\frac{d}{200}$ coordinates $i \in [d]$ satisfy
\begin{equation*}
\left|\blambda_i - \frac{n}{d}\right| \geq \frac{\sqrt{n}}{1000}.
\end{equation*}
When this event occurs,
\begin{align*}
\dtv{\underline{\blambda}}{\unif{d}}
= \frac{1}{2}\cdot \sum_{i = 1}^{d}\left|\frac{\blambda_i}{n} - \frac{1}{d}\right|
&= \frac{1}{2}\cdot \sum_{i = 1}^{d}\frac{1}{n}\cdot\left|\blambda_i - \frac{n}{d}\right|\\
&\geq
\frac{1}{2}\cdot \frac{d}{200}\cdot \frac{1}{n}\cdot \frac{\sqrt{n}}{1000}
= \frac{1}{400000} \cdot \frac{d}{\sqrt{n}},
\end{align*}
which is bigger than~$\eps$ unless $n = \Omega(d^2/\eps^2)$.
Showing this will prove Theorem~\ref{thm:large-n} with the parameters $C = 10^{10}$ and $\delta=.01$.

To begin, let us define a family of polynomials.
\begin{definition}
Given $k \geq 1$ and $c \in \R$, we define $\pstarc{k}{c}(\lambda)\coloneqq \sum_{i=1}^\infty (\lambda_i - i - c)^k - (-i - c)^k$.
\end{definition}
\noindent
This generalizes the definition of the $\pstar{k}$ polynomials,
as $\pstarc{k}{-\frac12} = \pstar{k}$.
\begin{fact}\label{fact:tedious-fact}
Let $c \in \R$.  Then
\begin{itemize}
\item $p^*_{2, c} = (-2c-1)\psharp{1} + \psharp{2}$, and
\item $p^*_{4, c} = (-4c^3-6c^2-4c-1)\psharp{1}+(6c^2+6c+4)\psharp{2}+(-6c-3)\psharp{(1, 1)}+(-4c-2)\psharp{3}+4\psharp{(2, 1)}+\psharp{4}$.
\end{itemize}
\end{fact}
\begin{proof}
    By explicit computation, one can check that
    \begin{align*}
        \pstarc{2}{c} = 2 (-c-\tfrac{1}{2}) \pstar{1} + \pstar{2}, \qquad
        \pstarc{4}{c} = 4 (-c-\tfrac{1}{2})^3 \pstar{1} + 6 (-c-\tfrac{1}{2})^2 \pstar{2}
    				+ 4 (-c-\tfrac{1}{2}) \pstar{3} + \pstar{4}.
    \end{align*}
    (Indeed, it's not hard to show that in general, $\pstarc{k}{c} = \sum_{j=1}^k \binom{k}{j}(-c - \tfrac{1}{2})^{k-j} \pstar{j}$.)  The claim now follows from~\eqref{eqn:low-moments}.
\end{proof}

For any $c$, these formulas allow us to compute the expected value of $\pstarc{2}{c}$ and $\pstarc{4}{c}$ over a random $\blambda \sim \SWunif{n}{d}$, by using Corollary~\ref{cor:sw-expect}.  Furthermore, for any $k$ and $d$, $\sum_{i=1}^d (-i - c)^k$ is a constant which doesn't depend on~$\blambda$.
Combining these two facts allows us to compute average value over a random $\blambda \sim \SWunif{n}{d}$
of $\sum_{i=1}^d (\blambda_i - i - c)^k$, for $k = 2,4$.
In particular, we are interested in computing this expectation when $c = \frac{n}{d}$.
Write $\bL_i \coloneqq  \blambda_i - i - \frac{n}{d}$.
Then
\begin{equation}\label{eq:squared-sum}
\E_{\blambda \sim \SWunif{n}{d}}\left[\sum_{i=1}^d \bL_i ^2\right]
= -\tfrac{n}{d}+nd+\tfrac{d^3}{3}+\tfrac{d^2}{2}+\tfrac{d}{6}
\geq -\tfrac{n}{d}+nd
\geq \tfrac{3nd}{4},
\end{equation}
where in the last step we used the fact that $n/d \leq nd /4 $ because $d \geq 2$.

Similarly, as $n \geq 10^{10} d^2 \geq d^2$, we can use the bound
\begin{align*}
\E_{\blambda \sim \SWunif{n}{d}}\left[\sum_{i=1}^d \bL_i ^4\right]
&= 2n-\tfrac{d}{30}-\tfrac{4n}{d^2}-\tfrac{6n}{d^3}+2nd^2+\tfrac{d^5}{5}+\tfrac{d^3}{3}+\tfrac{3n^2}{d^3}+\tfrac{d^4}{2}+nd^3+2 n^2 d+n d-\tfrac{5n^2}{d}+\tfrac{4n}{d}\\
&\leq 2n +2nd^2+\tfrac{d^5}{5}+\tfrac{d^3}{3}+\tfrac{3n^2}{d^3}+\tfrac{d^4}{2}+nd^3+2 n^2 d+n d +\tfrac{4n}{d}\\
& \leq 6 n^2 d,
\end{align*}
where in the last step we have used only trivial bounds involving the facts that $n \geq d^2$ and $d \geq 2$.

For a fixed $\lambda$, let $\calL(\lambda) \coloneqq  \{ i \in [d] \mid |L_i| \geq 5\sqrt{n}\}$.
Then
\begin{align*}
\E_{\blambda \sim \SWunif{n}{d}}\left[\sum_{i \in \calL(\blambda)} \bL_i ^2\right]
 \leq \frac{1}{25 n}\E_{\blambda \sim \SWunif{n}{d}}\left[\sum_{i \in \calL(\blambda)} \bL_i ^4\right]
 \leq \frac{1}{25 n}\E_{\blambda \sim \SWunif{n}{d}}\left[\sum_{i=1}^d \bL_i ^4\right]
 \leq \frac{nd}{4}.
\end{align*}
Thus, by~\eqref{eq:squared-sum},
\begin{equation*}
\E_{\blambda \sim \SWunif{n}{d}}\left[\sum_{i \in [d]\setminus \calL(\blambda)} \bL_i ^2\right] \geq \frac{nd}{2}.
\end{equation*}
Now define
\begin{equation*}
\calM(\lambda) \coloneqq  \left\{ i \in [d]~\middle|~\frac{\sqrt{n}}{200} \leq |L_i| < 5\sqrt{n}\right\},
\end{equation*}
and let $\calE$ be the event that $|\calM(\lambda)| \geq d/200$.
We claim that $p = \Pr[\calE] \geq 1/100$.  This is because if $p < 1/100$, then
\begin{equation*}
\E_{\blambda \sim \SWunif{n}{d}}\left[\sum_{i \in [d]\setminus \calL(\blambda)} \bL_i ^2\right]
\leq
p \cdot 25 n d + (1-p)\cdot \left(\frac{25n d}{200} + \left(1-\frac{1}{200}\right)\cdot\frac{n d}{200^2}  \right)
< \frac{nd}{2},
\end{equation*}
which is a contradiction.

Now let us use the assumption that $n \geq 10^{10}d^2$.
Consider any coordinate $i \in [d]$ satisfying
\begin{equation*}
|\bL_i| = \left|\blambda_i - i - \frac{n}{d}\right| \geq \frac{\sqrt{n}}{200}.
\end{equation*}
By our assumption that $n \geq 10^{10} d^2$, this implies that
\begin{equation*}
\left|\blambda_i - \frac{n}{d}\right| \geq \frac{\sqrt{n}}{1000}.
\end{equation*}
As a result, when $\calE$ holds, which happens with at least $1\%$ probability, there are at least $\frac{d}{200}$ coordinates $i \in [d]$ such that
\begin{equation*}
\left|\blambda_i - \frac{n}{d}\right| \geq \frac{\sqrt{n}}{1000}.
\end{equation*}
This completes the proof.
\end{proof}

\section{A quantum Paninski theorem}
In this section, we prove Theorem~\ref{thm:paninski-intro}, that $\Theta(d/\eps^2)$ copies are necessary and sufficient to test whether or not a given state $\rho \in \C^{d \times d}$ is the maximally mixed state, i.e., has spectrum $(\frac1d, \dots, \frac1d)$.

\subsection{The upper bound}                    \label{sec:mixedness-upper}

The upper bound for Theorem~\ref{thm:paninski-intro} will follow from our analysis of the following simple algorithm.
\begin{nameddef}{Mixedness Tester}
Given $\rho^{\otimes n}$, where $\rho$ is $d$-dimensional:
\begin{enumerate}
\item Sample $\blambda \sim \SWdens{n}{\rho}$.
\item Accept if $\psharp{2}(\blambda) \leq \left(1+\frac{\eps^2}{2}\right)\cdot\frac{n (n-1)}{d}$.  Reject otherwise.
\end{enumerate}
\end{nameddef}
\noindent
We remark that the tester Childs et al.~\cite{CHW07} used to distinguish the maximally mixed states of dimension~$\frac{d}{2}$ and~$d$ also depended only on the magnitude of $\psharp{2}(\blambda) = 2c_1(\blambda)$; see~\cite[equations~(49),~(50)]{CHW07}.

\begin{theorem}                                  \label{thm:mixedness-upper}
The Mixedness Tester can test whether a state $\rho \in \C^{d \times d}$ is the maximally mixed state using $n = O(d/\eps^2)$ copies of $\rho$.
\end{theorem}
\begin{proof}
We will run the Mixedness Tester with $n = 100 d/\eps^2$.  Both the ``completeness'' and the ``soundness'' analysis will require the last identity from~\eqref{eqn:simple-psharp-products}, namely
\begin{equation}\label{eq:thanks-maple}
    (\psharp{2})^2 = \psharp{(2,2)} + 4 \psharp{3} + 2\psharp{(1,1)}.
\end{equation}

\paragraph{Completeness.}
Suppose first that $\rho$ is the maximally mixed state, so that in fact $\blambda \sim \SWunif{n}{d}$.  We compute the mean and variance of $\psharp{2}(\blambda)$ using~\eqref{eq:thanks-maple} and Corollary~\ref{cor:sw-expect}:
\begin{gather}
\E_{\blambda \sim \SWunif{n}{d}}[\psharp{2}(\blambda)] = \frac{n(n-1)}{d},\label{eq:thanks-maple-avg}\\
\Var_{\blambda \sim \SWunif{n}{d}}\left[\psharp{2}(\blambda)\right]
= \E_{\blambda \sim \SWunif{n}{d}} \left[\psharp{2}(\blambda)^2\right] - \left(\E_{\blambda \sim \SWunif{n}{d}} [\psharp{2}(\blambda)] \right)^2
= \frac{2n(n-1)(d^2-1)}{d^2} \leq 2 n(n-1).\label{eq:thanks-maple-var}
\end{gather}
Thus by Chebyshev's inequality,
\begin{equation*}
\Pr_{\blambda \sim \SWunif{n}{d}}\left[\psharp{2}(\blambda) > \left(1+\frac{\eps^2}{2}\right)\cdot\frac{n(n-1)}{d}\right]
\leq \frac{8d^2}{n(n-1)\eps^4}\leq \frac{1}{3},
\end{equation*}
by our choice of $n$.  Thus indeed when $\rho$ is the maximally mixed state, the Mixedness Tester accepts with probability at least $2/3$.

\paragraph{Soundness.}
Suppose now that $\rho$ is a density matrix whose spectrum $\eta = (\eta_1, \ldots, \eta_d)$ satisfies $\dtvsymm{\eta}{\unif{d}} \geq \eps$.  Writing $\eta_i = \frac1d + \Delta_i$, this means that
\[
\eps \leq \frac12 \cdot \sum_{i=1}^d |\Delta_i| \leq \frac12 \sqrt{d \cdot \sum_{i=1}^d \Delta_i^2},
\]
using Cauchy--Schwarz; hence
\begin{equation}\label{eq:sum-of-squares-lower-bound}
\sum_{i=1}^d \Delta_i^2 \geq \frac{4\eps^2}{d}.
\end{equation}
Using~\eqref{eq:thanks-maple} and Proposition~\ref{prop:sw-alg-expect},
we can calculate the difference between the mean of $\psharp{2}(\blambda)$ and the cutoff used by the Mixedness Tester as
\begin{align*}
\E_{\blambda\sim \SWdens{n}{\rho}}\left[ \psharp{2}(\blambda)\right]
-
\frac{n(n-1)}{d} \cdot \left(1+\frac{\eps^2}{2}\right)
&= n(n-1)\left(\sum_{i=1}^d \eta_i^2 - \frac{1}{d}\left(1+\frac{\eps^2}{2}\right) \right).\\
&= n(n-1)\left(\sum_{i=1}^d \Delta_i^2 - \frac{\eps^2}{2d} \right)\\
&\geq \frac{n(n-1)}{2}\sum_{i=1}^d \Delta_i^2,
\end{align*}
where the last line follows from~\eqref{eq:sum-of-squares-lower-bound}.
Similarly, we can calculate the variance of $\psharp{2}(\blambda)$ as
\begin{align}
\Var_{\blambda \sim \SWdens{n}{\rho}}\left[\psharp{2}(\blambda)\right]
& = n(n-1)\left(2+4n\left(\sum \eta_i^3 - \left(\sum \eta_i^2 \right)^2\right)+6\left(\sum \eta_i^2\right)^2-8\sum \eta_i^3\right)\nonumber\\
& \leq n(n-1)\left(8 + 4n\left(\sum \eta_i^3 - \left(\sum \eta_i^2 \right)^2\right)\right)\nonumber\\
& = n(n-1)\left(8 + 4n\left(\frac{1}{d}\sum \Delta_i^2 + \sum \Delta_i^3 - \left(\sum \Delta_i^2\right)^2\right)\right)\nonumber\\
& \leq n(n-1)\left(8 + 8n\left(\sum \Delta_i^2\right)\right).\nonumber
\end{align}
Applying Chebyshev's inequality gives us
\begin{align*}
\Pr_{\blambda \sim \SWdens{n}{\rho}}\left[\psharp{(2)}(\blambda) < \left(1+\frac{\eps^2}{2}\right)\cdot\frac{n(n-1)}{d}\right]
& \leq \frac{1}{n(n-1)\left(\sum_{i=1}^d \Delta_i^2 \right)^2}\cdot \left(32 + 32 n \left(\sum_{i=1}^d \Delta_i^2\right)\right)\\
& \leq \frac{4}{n^2\left(\eps^2/d\right)^2}+ \frac{16}{n\left(\eps^2/d\right)},
\end{align*}
where the second step follows from~\eqref{eq:sum-of-squares-lower-bound}. By our choice of $n$, this is at most $1/3$.
Thus, when $\rho$ is $\eps$-far from the maximally mixed state, the Mixedness Tester rejects with probability at least $2/3$, as required.
\end{proof}

\subsection{The lower bound: overview}
For almost all of the lower bound proof we will assume~$d$ is even.  In the end we will indicate how to obtain the lower bound when~$d$ is odd.  For $0 \leq \eps \leq \frac12$, let $\pan{\eps}{d}$ denote the probability distribution on~$[d]$ in which
\[
    \pan{\eps}{d}(j) = \frac{1+(-1)^{j-1} 2\eps}{d}.
\]
This is essentially the same probability distribution that Paninski~\cite{Pan08} studies in his lower bound.
As usual, we also identify $\pan{\eps}{d}$ with the diagonal density matrix having these entries; i.e.,
\[
    \pan{\eps}{d} = \diag\left(\frac{1+2\eps}{d}, \frac{1-2\eps}{d}, \frac{1+2\eps}{d}, \frac{1-2\eps}{d}, \dots, \frac{1+2\eps}{d}, \frac{1-2\eps}{d}\right).
\]
Note that $\dtvsymm{\pan{\eps}{d}}{\unif{d}} = \eps$.  We also remark that when $\eps = \frac12$, the distribution $\pan{\eps}{d}$ is the uniform distribution on $\frac{d}{2}$ elements (the odd-numbered ones).  As in~\cite{Pan08}, it proves to be most convenient to study the chi-squared distance between $\SWdens{n}{\pan{\eps}{d}}$ and $\SWunif{n}{d}$; our main theorem is the following:
\begin{theorem}                                     \label{thm:main-pan-lower}
    $\displaystyle \dchi{\SWdens{n}{\pan{\eps}{d}}}{\SWunif{n}{d}} \leq \exp((4n\eps^2/d)^2) - 1$.
\end{theorem}
Since this distance is small unless $n = \Omega(d/\eps^2)$, our lower bound is complete.  More precisely:
\begin{corollary}                                       \label{cor:pan-lower}
    For even $d$, testing whether a $d$-dimensional mixed state $\rho$ has the the property of being the maximally mixed requires $n \geq .15 d/\eps^2$ copies.
\end{corollary}
\begin{proof}
    In light of Lemma~\ref{lem:wolog} we know that any $\eps$-tester may as well make its testing decision based on a draw $\blambda \sim \SWdens{n}{\rho}$.  Since $\dtvsymm{\pan{\eps}{d}}{\unif{d}} = \eps$, the tester must be able to distinguish a draw from $\SWdens{n}{\pan{\eps}{d}}$ and a draw from $\SWdens{n}{d}$ with probability advantage~$1/3$; this is possible if and only if $\dtv{\SWdens{n}{\pan{\eps}{d}}}{\SWunif{n}{d}} \geq 1/3$.  But
    \[
        \dtv{\SWdens{n}{\pan{\eps}{d}}}{\SWunif{n}{d}} \leq \frac12 \sqrt{\dchi{\SWdens{n}{\pan{\eps}{d}}}{\SWunif{n}{d}}} \leq \frac12 \sqrt{\exp((4n\eps^2/d)^2) - 1} < 1/3.
    \]
    if $n < .15 d/\eps^2$.
\end{proof}
\noindent We remark that by taking $\eps = \frac12$ we exactly recover the lower bound from Theorem~\ref{thm:q-bday} due to Childs et al.~\cite{CHW07}.

There are two major steps in the proof of Theorem~\ref{thm:main-pan-lower}.  The first major step will be proving the following formula:
\begin{theorem}                                     \label{thm:paninski-var}
    Let $x \in \R^d$ satisfy $x_1 + \cdots + x_d = 0$ .  Then
    \[
        \E_{\blambda \sim \SWunif{n}{d}}\left[\PAREN{
                            \frac{s_{\blambda}(1+x_1, \dots, 1+x_d)}{s_{\blambda}(1, \dots, 1)} - 1}^2\right]
        = \sum_{\substack{\mu \in \allpartitions \\ 0 < \ell(\mu) \leq d}}
                \frac{s_\mu(x)^2}{\rising{d}{\mu} \cdot d^{|\mu|}} \cdot \falling{n}{|\mu|}.
    \]
    (The sum has only finitely many terms since $\falling{n}{|\mu|} = 0$ when $|\mu| > n$.)
\end{theorem}
Once the above theorem is established, the following consequence is essentially immediate:
\begin{corollary}                                     \label{cor:paninski-chi-squared}
    Let  $x \in \R^d$ satisfy $x_1 + \cdots + x_d = 0$ and $x_i \geq -1$ for all~$i$.  We write~$\calQ_x$ for the probability distribution on~$[d]$ in which~$i$ has probability $\frac{1+x_i}{d}$. Then
    \[
          \dchi{\SWdens{n}{\calQ_x}}{\SWunif{n}{d}}
        = \sum_{\substack{\mu \in \allpartitions \\ 0 < \ell(\mu) \leq d}}
                \frac{s_\mu(x)^2}{\rising{d}{\mu} \cdot d^{|\mu|}} \cdot \falling{n}{|\mu|}.
    \]
\end{corollary}
\begin{proof}
    By definition, $\dchi{\SWdens{n}{\calQ_x}}{\SWunif{n}{d}}$ is equal to
    \begin{align*}
         \E_{\blambda \sim \SWunif{n}{d}}\left[\PAREN{\frac{\SWdensProb{n}{\calQ_x}{\blambda}}{\SWunifProb{n}{d}{\blambda}}-1}^2\right]
          &= \E_{\blambda \sim \SWunif{n}{d}}\left[\PAREN{\frac{s_{\blambda}\paren{\tfrac{1+x_1}{d}, \dots, \tfrac{1+x_d}{d}} \dim(\blambda)}{ s_{\blambda}\paren{\tfrac{1}{d}, \dots, \tfrac{1}{d}}\dim(\blambda)}-1}^2\right].
    \end{align*}
    where we used Proposition~\ref{prop:schur-probability}.  In turn, this equals the quantity on the left in Theorem~\ref{thm:paninski-var} after canceling the common factor of~$d^{-|\blambda|} \dim \blambda$ in the fraction (recall the homogeneity of the Schur polynomials).
\end{proof}

Let us sketch the intuition of the proof once Theorem~\ref{thm:paninski-var} is established.  We are ultimately interested in the case $x = 2\eps \cdot c$, where $\eps > 0$ is thought of as ``small'' and $c \in \R^d$ satisfies $c_1 + \cdots + c_d = 0$; specifically, $c = c_\pm \coloneqq (+1,-1,+1,-1, \dots, +1,-1)$.  For simplicity, let us write $\eps$ instead of $2\eps$.  Since $s_\mu$ is homogeneous of degree~$|\mu|$, this means $s_\mu(x)^2 = s_\mu(c)^2\eps^{2|\mu|} $.  For the sake of intuition, let us consider the summands in Theorem~\ref{thm:paninski-var} when $|\mu| = k$ is ``small''; i.e., the coefficients on~$\eps^{2k}$. 
For $k = 1$ we have only $\mu = (1)$, and the associated summand actually drops out:  this is because $s_{(1)}(x) = x_1 + \cdots + x_d = 0$.  For $k \geq 2$, the term $\falling{n}{|\mu|}$ is asymptotically~$n^k$ and the denominator $d^{|\mu|}\cdot\rising{d}{\mu}$ is asymptotically~$d^{2k}$.  It remains to analyze $s_\mu(c_\pm)$.  This is the second major step in the proof of Theorem~\ref{thm:main-pan-lower}: in Section~\ref{sec:maya} we establish an exact formula for it.  Naively one might expect  $|s_\mu(c_\pm)|$ to scale like~$d^k$ when $|\mu| = k$; however, as we will see it scales only like $d^{k/2}$ (and will in fact be~$0$ whenever~$k$ is odd).  Thus the summands with $|\mu| = k$ small scale asymptotically as~$n^k \cdot \frac{\eps^{2k}}{d^k}$, whence we get that $\dchi{\SWdens{n}{\calQ_{\eps \cdot c_\pm}}}{\SWunif{n}{d}}$ is small if $n \ll \frac{d}{\eps^2}$.

\subsection{Proof of Theorem~\ref{thm:paninski-var}}
To analyze the quantity in Theorem~\ref{thm:paninski-var} we will require the so-called Binomial Formula.  (It generalizes the ``usual'' Binomial Formula, viz.\ $(1+x)^\ell = \sum_{m \geq 0} x^m \falling{\ell}{m}/m!$, in the case $d = 1$.)
\begin{theorem}                                     \label{thm:binomial-formula}
    The following polynomial identity holds:
    \[
          \frac{s_{\lambda}(1+x_1, \dots, 1+x_d)}{s_\lambda(1, \dots, 1)}
        = \sum_{\substack{\mu \in \allpartitions \\ \ell(\mu) \leq d}} \frac{s_\mu(x)}{\rising{d}{\mu}}\cdot s^*_\mu(\lambda).
    \]
    (The sum is actually finite since we may include the restriction $\mu \subseteq \lambda$ due to the factor $s^*_\mu(\lambda)$.)
\end{theorem}
\noindent In this form with the shifted Schur polynomials, the result appears in Okounkov and Olshanski's work~\cite[Theorem 5.1]{OO98b} (see also~\cite{OO98a}). In a form involving factorial Schur polynomials it dates back to Lascoux~\cite{Las78}; see~\cite[Example I.3.10]{Mac95}.

The $\mu = \emptyset$ summand in Theorem~\ref{thm:binomial-formula} is always equal to~$1$; it follows that the quantity on the left of Theorem~\ref{thm:paninski-var} is
\[
\E_{\blambda \sim \SWunif{n}{d}}\left[\PAREN{\sum_{0 < \ell(\mu) \leq d} \frac{s_\mu(x)}{\rising{d}{\mu}}\cdot s^*_\mu(\blambda)}^2\right]
    = \sum_{0 < \ell(\mu), \ell(\nu) \leq d} \frac{s_\mu(x)s_{\nu}(x)}{\rising{d}{\mu}\rising{d}{\nu}} \E_{\blambda \sim \SWunif{n}{d}}\left[s^*_\mu(\blambda)s^*_{\nu}(\blambda)\right].
\]
Therefore proving Theorem~\ref{thm:paninski-var} reduces to proving
\begin{equation}                                    \label{eqn:paninski-var2}
    x_1 + \cdots + x_d = 0 \implies
    \sum_{0 < \ell(\mu), \ell(\nu) \leq d} \frac{s_\mu(x)s_{\nu}(x)}{\rising{d}{\mu}\rising{d}{\nu}} \E_{\blambda \sim \SWunif{n}{d}}\left[s^*_\mu(\blambda)s^*_{\nu}(\blambda)\right]
    = \sum_{0 < \ell(\mu) \leq d}  \frac{s_\mu(x)^2}{\rising{d}{\mu} \cdot d^{|\mu|}} \cdot \falling{n}{|\mu|}.
\end{equation}
This is the main difficult step of the proof; the surprising aspect here is that we only get a contribution on the order of~$n^{k}$ from the terms with $|\mu| = k$, whereas naively one would expect~$n^{2k}$.
Showing that the $n^{k+1}, n^{k+2}, \dots, n^{2k}$ contributions ``drop out'' is the essence of the proof.

In aid of proving~\eqref{eqn:paninski-var2}, it's tempting to guess that $\E[s^*_\mu(\blambda)s^*_{\nu}(\blambda)] = 1_{\{\mu = \nu\}} \cdot \frac{\rising{d}{\mu}}{d^{|\mu|}}\cdot\falling{n}{|\mu|}$; however such a statement is false.  Instead, what \emph{is} true is the following:
\begin{theorem}                                     \label{thm:paninski-internal} \label{thm:paninski-internal2}
    Let $x \in \R^d$ satisfy $x_1 + \cdots + x_d = 0$ and let $\mu \in \allpartitions$ satisfy $|\mu| = r_1$ and $0 < \ell(\mu) \leq d$.  Assume $r_2 \geq r_1$. Then
    \[
        \sum_{\substack{|\nu| = r_2 \\ \ell(\nu) \leq d}} \frac{s_\nu(x)}{\rising{d}{\nu}}
         \E_{\blambda \sim \SWunif{n}{d}}\left[s^*_\mu(\blambda)s^*_{\nu}(\blambda)\right]
        = 1_{\{r_2 = r_1\}} \cdot 
              \frac{s_\mu(x)}{d^{|\mu|}}\cdot \falling{n}{|\mu|}. 
    \]
\end{theorem}
\noindent To deduce~\eqref{eqn:paninski-var2} from Theorem~\ref{thm:paninski-internal}, simply write
\[
\sum_{0 < \ell(\mu), \ell(\nu) \leq d} \frac{s_\mu(x)s_{\nu}(x)}{\rising{d}{\mu}\rising{d}{\nu}} \E_{\blambda \sim \SWunif{n}{d}}\left[s^*_\mu(\blambda)s^*_{\nu}(\blambda)\right] = \sum_{r_1,r_2 > 0} \sum_{\substack{|\mu| = r_1 \\ \ell(\mu) \leq d}} \sum_{\substack{|\nu| = r_2 \\ \ell(\nu) \leq d}} \frac{s_\mu(x)s_{\nu}(x)}{\rising{d}{\mu}\rising{d}{\nu}} \E_{\blambda \sim \SWunif{n}{d}}\left[s^*_\mu(\blambda)s^*_{\nu}(\blambda)\right].
\]
Then use Theorem~\ref{thm:paninski-internal} when $r_2 \geq r_1$ and use it with the roles of $\mu$ and $\nu$ reversed when $r_2 < r_1$.


As for the proof of Theorem~\ref{thm:paninski-internal2} itself, the first step is to compute the expected product of the shifted Schur polynomials. One possible approach for this might be to use the Littlewood--Richardson rule for factorial Schur functions (see \cite[Proposition~4.2]{MS99} or~\cite[Corollary~3.3]{Mol09}) to write $s^*_\mu s^*_\nu$ as a linear combination of $s^*_\tau$ polynomials. Unfortunately, these Littlewood--Richardson coefficients seem somewhat difficult to work with. Instead, we will expand the shifted Schur polynomials in terms of the central characters and then multiply them via the known structure constants.  We do this in the below lemma, carried out for a generic Schur--Weyl distribution.  In this lemma, $\symmset{R}$ denotes the symmetric group acting on the finite set~$R$.
\begin{lemma}                                       \label{lem:expected-shifted-schur-prod}
    Let $q = (q_1, \dots, q_d)$ be a probability distribution on~$[d]$ and let $\mu \vdash r_1$, $\nu \vdash r_2$.  Then
    \[
          \E_{\blambda \sim \SWdens{n}{q}}\left[s^*_\mu(\blambda)s^*_{\nu}(\blambda)\right]
        = \sum_{t = r_1 \vee r_2}^{r_1+r_2} C^{t}_{r_1r_2} \cdot \falling{n}{t} \cdot
              \E_{\substack{\bw_1 \sim \symmset{R_1} \\
                            \bw_2 \sim \symmset{R_2}}}
              \left[
                  \chi_{\mu}(\bw_1) \chi_{\nu}(\bw_2) p_{\overline{\bw}_1 \overline{\bw}_2}(q)
              \right].
    \]
Here, for each choice of~$t$, we let $R_1, R_2$ denote (arbitrary but fixed) subsets of~$[t]$ having cardinality $r_1, r_2$, respectively, with $R_1 \cup R_2 = [t]$.  (E.g., $R_1 = \{1, \dots, r_1\}$, $R_2 = \{t - r_2 + 1, \dots, t\}$.)  Also, $\overline{\bw}_1$ denotes the extension of $\bw_1$ to $\symm{t}$ formed by letting $\overline{\bw}_1$ fix each element of $[t] \setminus R_1$; similarly for~$\overline{\bw}_2$.
\end{lemma}
\begin{proof}
    Recall the notation $\cyctype{w}$ from Section~\ref{sec:partitions} used denote the cycle type of a permutation~$w$.  In this proof, we also use the following notation:  We write $\brho \sim \symm{r}$ to denote that $\brho$ is a random partition of~$r$ formed by first choosing $\bw \sim \symm{r}$ uniformly and then taking $\brho = \cyctype{\bw}$.

    Using Theorem~\ref{thm:psharp-ssharp-relation} for the first equality below, and Corollary~\ref{cor:structure-coefficients-rewrite} for the third equality, we have
    \begin{align*}
        &\phantom{{}=}  \E_{\blambda \sim \SWdens{n}{q}}\left[s^*_\mu(\blambda)s^*_{\nu}(\blambda)\right] \\
        &= \E_{\blambda \sim \SWdens{n}{q}}\left[
                \E_{\brho_1 \sim \symm{r_1}}[\chi_{\mu}(\brho_1) \cdot \psharp{\brho_1}(\blambda)]
           \cdot\E_{\brho_2 \sim \symm{r_2}}[\chi_{\nu}(\brho_2) \cdot \psharp{\brho_2}(\blambda)]
                                          \right] \\
        &= \E_{\substack{\brho_1 \sim \symm{r_1} \\ \brho_2 \sim \symm{r_2}}}
           \left[
                    \chi_{\mu}(\brho_1) \chi_{\nu}(\brho_2) \cdot
                            \E_{\blambda \sim \SWdens{n}{q}}
                            \left[
                                    \psharp{\brho_1}(\blambda) \cdot \psharp{\brho_2}(\blambda)
                            \right]
           \right]\\
        &= \E_{\substack{\brho_1 \sim \symm{r_1} \\ \brho_2 \sim \symm{r_2}}}
           \left[
                    \chi_{\mu}(\brho_1) \chi_{\nu}(\brho_2) \cdot
                            \E_{\blambda \sim \SWdens{n}{q}}
                            \left[
                                    \sum_{t = r_1 \vee r_2}^{r_1+r_2} \sum_{\tau \vdash t} C^{t}_{r_1r_2} \cdot
                                        \Pr_{\substack{\bw_1 \sim \symmset{R_1} \mid \brho_1
                                             \\        \bw_2 \sim \symmset{R_2} \mid \brho_2}}
                                        [\cyctype{\overline{\bw}_1 \overline{\bw}_2} = \tau]
                                    \cdot \psharp{\tau}(\blambda)
                            \right]
           \right],
    \end{align*}
    where here $\bw_i$ is chosen to be a uniformly random permutation on $R_i$ (as in the lemma's statement), \emph{conditioned on} having cycle type~$\brho_i$.  By Proposition~\ref{prop:sw-alg-expect} the above equals
    \begin{align*}
        &\phantom{{}=} \E_{\substack{\brho_1 \sim \symm{r_1} \\ \brho_2 \sim \symm{r_2}}}
           \left[
                    \chi_{\mu}(\brho_1) \chi_{\nu}(\brho_2) \cdot
                                    \sum_{t = r_1 \vee r_2}^{r_1+r_2} \sum_{\tau \vdash t} C^{t}_{r_1r_2} \cdot
                                        \Pr_{\substack{\bw_1 \sim \symmset{R_1} \mid \brho_1
                                             \\        \bw_2 \sim \symmset{R_2} \mid \brho_2}}
                                        [\cyctype{\overline{\bw}_1 \overline{\bw}_2} = \tau]
                                    \cdot \falling{n}{t} \cdot p_\tau(q)
           \right] \\
        &= \sum_{t = r_1 \vee r_2}^{r_1+r_2} C^{t}_{r_1r_2} \cdot \falling{n}{t} \cdot
              \E_{\substack{\brho_1 \sim \symm{r_1},\ \brho_2 \sim \symm{r_2} \\
                            \bw_1 \sim \symmset{R_1} \mid \brho_1 \\
                            \bw_2 \sim \symmset{R_2} \mid \brho_2}}
              \left[
                  \chi_{\mu}(\brho_1) \chi_{\nu}(\brho_2) \cdot
                  \sum_{\tau \vdash t} 1_{\{\cyctype{\overline{\bw}_1 \overline{\bw}_2} = \tau\}} \cdot
                  p_\tau(q)
              \right]
    \end{align*}
    The summation on the inside here simply equals $p_{\cyctype{\overline{\bw}_1 \overline{\bw}_2}}(q)$; we may also replace $\chi_\mu(\brho_1)$ with $\chi_\mu(\bw_1)$, and similarly for $\chi_{\nu}(\brho_2)$.  Thus to complete the proof it remains to show that~$\bw_1$ and~$\bw_2$ have the same distribution as in the statement of the lemma.  But this is clear:  if we first pick a random permutation of~$r_i$ symbols, then take its cycle type, then set~$\bw_i$ to be a random permutation of~$r_i$ symbols of this cycle type, this is the same as simply taking~$\bw_i$ to be a uniformly random permutation of~$r_i$ symbols.
\end{proof}

We will also require the following Fourier-theoretic lemma:

\begin{lemma}                                     \label{lem:beginning-lemma}
    For $u \in \symm{r}$, $\nu \vdash r$, and $d \in \Z^+$,
    \[
        \E_{\bw \sim \symm{r}}[\chi_\nu(\bw) \cdot d^{\ell(u \bw)}] = \frac{{\chi_\nu}(u)\rising{d}{\nu}}{r!}.
    \]
\end{lemma}
\begin{proof}
    Define the class function $e$ on $\symm{r}$ by
    \[
        e(v) = p_{v}(\underbrace{1, \dots, 1}_{d\ \text{entries}}) = d^{\ell(v)}.
    \]
    Since $\chi_\nu(\bw)  = \chi_\nu(\bw^{-1})$ because $\chi_\nu$ is a class function, the quantity on the left in the proposition's statement is
    \begin{multline*}
          \E_{\bw \sim \symm{r}}[\chi_\nu(\bw^{-1}) \cdot d^{\ell(u \bw)}]
        = \E_{\bv \sim \symm{r}}[\chi_\nu(\bv^{-1} u) \cdot d^{\ell(\bv)}]
        = (\conv{e}{\chi_\nu})(u) = \sum_{\mu \vdash r} \four{\conv{e}{\chi_\nu}}(\mu) \chi_\mu(u) \\
        = \sum_{\mu \vdash r} \tfrac{1}{\dim \mu} \four{e}(\mu) \four{\chi_\nu}(\mu) \chi_\mu(u) = \tfrac{1}{\dim \nu} \four{e}(\nu)  \chi_\nu(u) = \tfrac{1}{\dim \nu}  s_\nu(1, \dots, 1) \chi_\nu(u) = \frac{{\chi_\nu}(u)\rising{d}{\nu}}{r!},
    \end{multline*}
    the last equality being Proposition~\ref{prop:s111}.
\end{proof}

We can now complete the proof of Theorem~\ref{thm:paninski-internal2} (and therefore also Theorem~\ref{thm:paninski-var}):
\begin{proof}[Proof of Theorem~\ref{thm:paninski-internal2}]
    We will use Lemma~\ref{lem:expected-shifted-schur-prod} in the case of $\SWunif{n}{d}$, i.e., $q = (\tfrac1d, \dots, \tfrac1d)$; in this case, for $\tau \vdash t$ we have $p_\tau(q) = d^{\ell(\tau)-t}$.  We thereby obtain
    \begin{multline}
        \sum_{\substack{|\nu| = r_2 \\ \ell(\nu) \leq d}} \frac{s_\nu(x)}{\rising{d}{\nu}}
             \E_{\blambda \sim \SWunif{n}{d}}\left[s^*_\mu(\blambda)s^*_{\nu}(\blambda)\right] \\
        = \sum_{t = r_2}^{r_1+r_2} C^{t}_{r_1r_2} \cdot \frac{\falling{n}{t}}{d^t} \cdot
          \E_{\bw_1 \sim \symmset{R_1}}
          \Brak{
                  \chi_{\mu}(\bw_1) \cdot
                  \sum_{\substack{|\nu| = r_2 \\ \ell(\nu) \leq d}} \frac{s_\nu(x)}{\rising{d}{\nu}}
                    \cdot \E_{\bw_2 \sim \symmset{R_2}}
              \left[
                  \chi_{\nu}(\bw_2) d^{\ell(\overline{\bw}_1 \overline{\bw}_2)}
              \right]
          }.            \label{eqn:paninski-come-home}
    \end{multline}
    (Here we are using the convention $\ell(\ol{\bw}_1 \ol{\bw}_2) = \ell(\cyctype{\ol{\bw}_1 \ol{\bw}_2})$.) We now would like to analyze the number of cycles of $\ol{\bw}_1 \ol{\bw}_2$ within $\symm{t}$.  In~$\ol{\bw}_1$'s cycle decomposition, there are some cycles that act \emph{only} on elements of~$R_1 \setminus R_2$.  Let's write $\ell^\setminus(\bw_1)$ for the number of such cycles, and let's define $\ol{\bw}_1^\cap \in \symm{t}$ to be $\ol{\bw}_1$ with those cycles deleted.  Thus
    \[
        \ell(\overline{\bw}_1 \overline{\bw}_2) = \ell^\setminus(\bw_1) + \ell(\ol{\bw}_1^\cap \cdot \overline{\bw}_2).
    \]
    Next, let $\bw_1^\bot$ denote the permutation obtained by deleting every element of $R_1 \setminus R_2$ from the cycle decomposition of~$\ol{\bw}_1^\cap$.  Though $\bw_1^\bot$ acts only on $R_1 \cap R_2$, we will view it as an element of~$\symmset{R_2}$.  Although we don't have $\bw_1^\bot \cdot \bw_2 = \ol{\bw}_1^\cap \cdot \overline{\bw}_2$, it's not too hard to see that
    \[
        \ell(\ol{\bw}_1^\cap \cdot \overline{\bw}_2) = \ell(\bw_1^\bot \cdot \bw_2).
    \]
    Thus we obtain
    \[
        \eqref{eqn:paninski-come-home} =
        \sum_{t = r_2}^{r_1+r_2} C^{t}_{r_1r_2} \cdot \frac{\falling{n}{t}}{d^t} \cdot
          \E_{\bw_1 \sim \symmset{R_1}}
          \Brak{
                  \chi_{\mu}(\bw_1) d^{\ell^\setminus(\bw_1)} \cdot
                  \sum_{\substack{|\nu| = r_2 \\ \ell(\nu) \leq d}} \frac{s_\nu(x)}{\rising{d}{\nu}}
                    \cdot \E_{\bw_2 \sim \symmset{R_2}}
              \left[
                  \chi_{\nu}(\bw_2) d^{\ell(\bw_1^{\bot} \cdot \bw_2)}
              \right]
          }.
    \]
    Applying Lemma~\ref{lem:beginning-lemma}, we deduce
    \[
        \eqref{eqn:paninski-come-home} =
        \sum_{t = r_2}^{r_1+r_2} C^{t}_{r_1r_2} \cdot \frac{\falling{n}{t}}{d^t} \cdot
          \E_{\bw_1 \sim \symmset{R_1}}
          \Brak{
                  \chi_{\mu}(\bw_1)  d^{\ell^\setminus(\bw_1)} \cdot
                  \frac{1}{r_2!}\sum_{\substack{|\nu| = r_2 \\ \ell(\nu) \leq d}} s_\nu(x)\chi_\nu(\bw_1^\bot)
          }.
    \]
    Notice that we may extend the summation over~$\nu$ to include $\ell(\nu) > d$ as well: since~$x$ has~$d$ coordinates, $s_\nu(x) = 0$ anyway when $\ell(\nu) > d$ by Proposition~\ref{prop:identically-zero}.  Having done this, we replace $s_\nu(x)$ with $\E_{\bv \sim \symm{r_2}}[\chi_\nu(\bv)p_{\bv}(x)]$, obtaining
    \begin{align*}
        \eqref{eqn:paninski-come-home} &=
        \sum_{t = r_2}^{r_1+r_2} C^{t}_{r_1r_2} \cdot \frac{\falling{n}{t}}{d^t} \cdot
          \E_{\bw_1 \sim \symmset{R_1}}
          \Brak{
                  \chi_{\mu}(\bw_1) d^{\ell^\setminus(\bw_1)} \cdot
                  \frac{1}{r_2!}\sum_{|\nu| = r_2}
                  \E_{\bv \sim \symm{r_2}}[\chi_\nu(\bv)\cdot p_{\bv}(x)]\chi_\nu(\bw_1^\bot)
          } \\
        &= \sum_{t = r_2}^{r_1+r_2} \frac{C^{t}_{r_1r_2}}{r_2!} \cdot \frac{\falling{n}{t}}{d^t} \cdot
          \E_{\bw_1 \sim \symmset{R_1}}
          \Brak{
                  \chi_{\mu}(\bw_1) d^{\ell^\setminus(\bw_1)} \cdot
                  \E_{\bv \sim \symm{r_2}} \Brak{p_{\bv}(x) \cdot
                        \sum_{|\nu| = r_2}\chi_{\nu}(\bv)\chi_{\nu}(\bw_1^\bot)}
          }.
    \end{align*}
    We claim that the inner expectation is~$0$ in most cases.  First, $p_{\bv}(x)$ vanishes whenever $\bv$ has a fixed point, since $p_1(x) = x_1 + \cdots + x_d = 0$ by assumption.  Next, suppose that~$\bv$ has no fixed points. By the orthogonality relations of representation theory, the innermost sum vanishes unless~$\bv$ and $\bw_1^\bot$ are conjugate. Since $\bw_1^\bot \in \symmset{R_2}$ acts only on~$\R_1 \cap R_2$, it \emph{must} have a fixed point (and therefore not be conjugate to~$\bv$) unless $R_2 \setminus R_1 = \emptyset$.  Since $r_2 \geq r_1$, this can only happen if $|\mu| = r_1 = r_2 = t$.  We conclude that the inner expectation can only be nonzero in case $|\mu| = r_1 = r_2 = t$. In this case we have $C^t_{r_1r_2} = r_2!$ and $\ell^\setminus(\bw_1) = 0$, whence
    \[
        \eqref{eqn:paninski-come-home} = 1_{\{r_2 = r_1\}} \cdot \frac{\falling{n}{r_1}}{d^{r_1}} \cdot
          \E_{\bw_1 \sim \symm{r_1}}
          \Brak{
                  \chi_{\mu}(\bw_1) \cdot
                  \E_{\bv \sim \symm{r_1}} \Brak{ p_{\bv}(x) \cdot
                        \sum_{|\nu| = r_1}\chi_{\nu}(\bv)\chi_{\nu}(\bw_1^\bot)}
          }.
    \]
    Once again, the summation is~$0$ if $\bv$ and $\bw_1$ are not conjugate; otherwise it equals~$z_{\cyctype{\bw_1}}$.  Further, having chosen~$\bw_1$, the probability that $\bv$ is conjugate to~$\bw_1$ is precisely~$z_{\cyctype{\bw_1}}^{-1}$.  Thus these factors cancel and we obtain
    \[
        \eqref{eqn:paninski-come-home} = 1_{\{r_2 = r_1\}} \cdot \frac{\falling{n}{r_1}}{d^{r_1}} \cdot
          \E_{\bw_1 \sim \symm{r_1}}
                      [\chi_{\mu}(\bw_1) \cdot p_{\bw_1}(x)]
          = 1_{\{r_2 = r_1\}} \cdot \frac{\falling{n}{r_1}}{d^{r_1}} \cdot s_{\mu}(x),
    \]
    completing the proof.
\end{proof}

\ignore{
\newcommand{\normch}[1]{\widehat{\chi}_{#1}}
We use the following fact about convolution of characters is sometimes called the ``generalized\rnote{Is is the same as ``great'' or ``grand''?} orthogonality relation''.\cite{theorem 2.13 in Character Theory of Finite Groups by I. Martin Isaacs, see alsomath.stackexchange.com/questions/152300/convolution-of-irreducible-characters-of-a-finite-group}\rnote{Actually, if you look in Isaacs you get $\chi_\mu(\bw u)\chi_nu(\bw^{-1})$ inside the expectation.  It's the same for the following reason.  First, we may replace $\bw$ with $\bw^{-1}$, getting $\chi_\mu(\bw^{-1} u)\chi_\nu(\bw)$ (that's much more natural anyway). Finally, we can replace $\bw$ with $u\bw u^{-1}$; that's still uniformly random.  Then we get $\chi_\mu(u\bw^{-1})\chi_{\nu}(u \bw u^{-1}) = \chi_\mu(u\bw^{-1})\chi_nu(\bw)$, where the equality uses that characters are class functions (invariant to conjugation of inputs).}
\begin{proposition}                                     \label{prop:convolution-formula}
    Let $G$ be a finite group with irreducible characters $\widehat{G}$.\rnote{is that good notation/phraseology?} Then for $\mu, \nu \in \widehat{G}$ and $u \in G$ we have\rnote{Note that LHS is literally $\conv{\chi_\mu}{\chi_\nu}(u)$.  That makes the following ``obvious'' to me.  To understand $\chi_\mu \ast \chi_\nu$, look at its Fourier transform. Fourier coeffs are pointwise product of those of $\chi_\mu$, $\chi_\nu$.  Now it's obvious you get $0$ if $\mu \neq \nu$.  Otherwise you just get $\chi_\nu$ (up to various normalizations).}
    \[
        \E_{\bw \sim G}[\chi_\mu(u\bw^{-1}) \chi_\nu(\bw)] = \begin{cases}
                                                                      \normch{\nu}(u) & \text{if $\mu = \nu$,} \\
                                                                      0 & \text{else,}
                                                                  \end{cases}
    \]
    where $\normch{\mu}$ denotes the normalized character\rnote{did we already define it?} defined by $\normch{\mu}(u) = \chi_\mu(u) / \chi_\mu(1)$.
\end{proposition}
We derive the following consequence:\rnote{same as John's ``beginning lemma''}:
}

\subsection{A formula for $s_\mu(+1, -1, +1, -1, \dots)$}               \label{sec:maya}
For this formula we will need to recall the notion of the $2$-quotient of a partition.
\ignore{
Several equivalent definitions are known; we give a slightly nonstandard one, see~\cite[Theorem 2.7.37]{JK81}:
\begin{definition}
    Let $\mu \vdash k$.  Let $\evenhooks{\mu}$ denote all boxes $\square \in [\mu]$ with even hook length.  Observe that the hook length of~$\square$ is equal to the number of boxes in the rim-path from $\hand(\square)$ to $\foot(\square)$.\rnote{does the reader know what I'm talking about at this point?} Since the content alternates in parity along this path, we see that $\square \in \evenhooks{\mu}$ if and only if $\hand(\square)$ and $\foot(\square)$ have contents of different parity.  We therefore have the partition $\evenhooks{\mu} = \evenhandhooks{0}{\mu} \sqcup \evenhandhooks{1}{\mu}$, where
    \[
        \evenhandhooks{b}{\mu} = \{\square \in [\mu] : c(\text{hand}(\square)) \equiv b,\ c(\text{foot}(\square)) \equiv b+1\ \text{(mod~$2$)}\}.
    \]
    It is not hard to see that the boxes of $\evenhandhooks{0}{\mu}$ form a valid Young diagram after being left- and top-justified, as do the boxes of $\evenhandhooks{1}{\mu}$. The associated pair of partitions $(\mu^{(0)}, \mu^{(1)})$ is called the \emph{$2$-quotient} of~$\mu$.
\end{definition}
\begin{definition}
    Let $\mu \vdash k$.  We say that $\mu$ is \emph{balanced} if exactly half of its hook lengths are even; i.e., $|\evenhooks{\mu}| = |\mu^{(0)}| + |\mu^{(1)}| = k/2$.  In particular, $k$ must be even.  In the literature, the condition is also referred to as the ``$2$-core of $\mu$ being empty''.  An equivalent condition~\cite[Section~2.1]{Rio12} is that $[\mu]$ can tiled by dominoes.  Another equivalent condition (following from~\cite[Theorem~2.7.41]{JK81}) is that exactly half of $[\mu]$'s boxes have even content.
\end{definition}
}%
This definition essentially encodes the ways in which a partition can be tiled by dominoes.
\begin{definition}                                      \label{def:2-hook}
    Given a partition~$\mu$, a \emph{$2$-hook} in $[\mu]$ is a hook of length~$2$; i.e., a domino whose removal from~$[\mu]$ results in a valid Young diagram.
\end{definition}
\begin{definition}  \label{def:balanced}
    A partition $\mu$ is said to be \emph{balanced} (or to have an \emph{empty $2$-core}) if $[\mu]$ can be reduced to the empty diagram by successive removal of $2$-hooks.
\end{definition}
\begin{definition}
    Given a partition~$\mu$ we write $\evencontents{\mu}$ (respectively, $\oddcontents{\mu}$) for the set of boxes $\square \in [\mu]$ with even (respectively, odd) content~$c(\square)$.
\end{definition}
\begin{remark}
    It's obvious from Definition~\ref{def:balanced} that if $\mu \vdash k$ is balanced then $|\evencontents{\mu}| = k/2$.  In fact, the converse also holds (this follows from, e.g.,~\cite[Theorem~2.7.41]{JK81}).
\end{remark}
\myfig{.375}{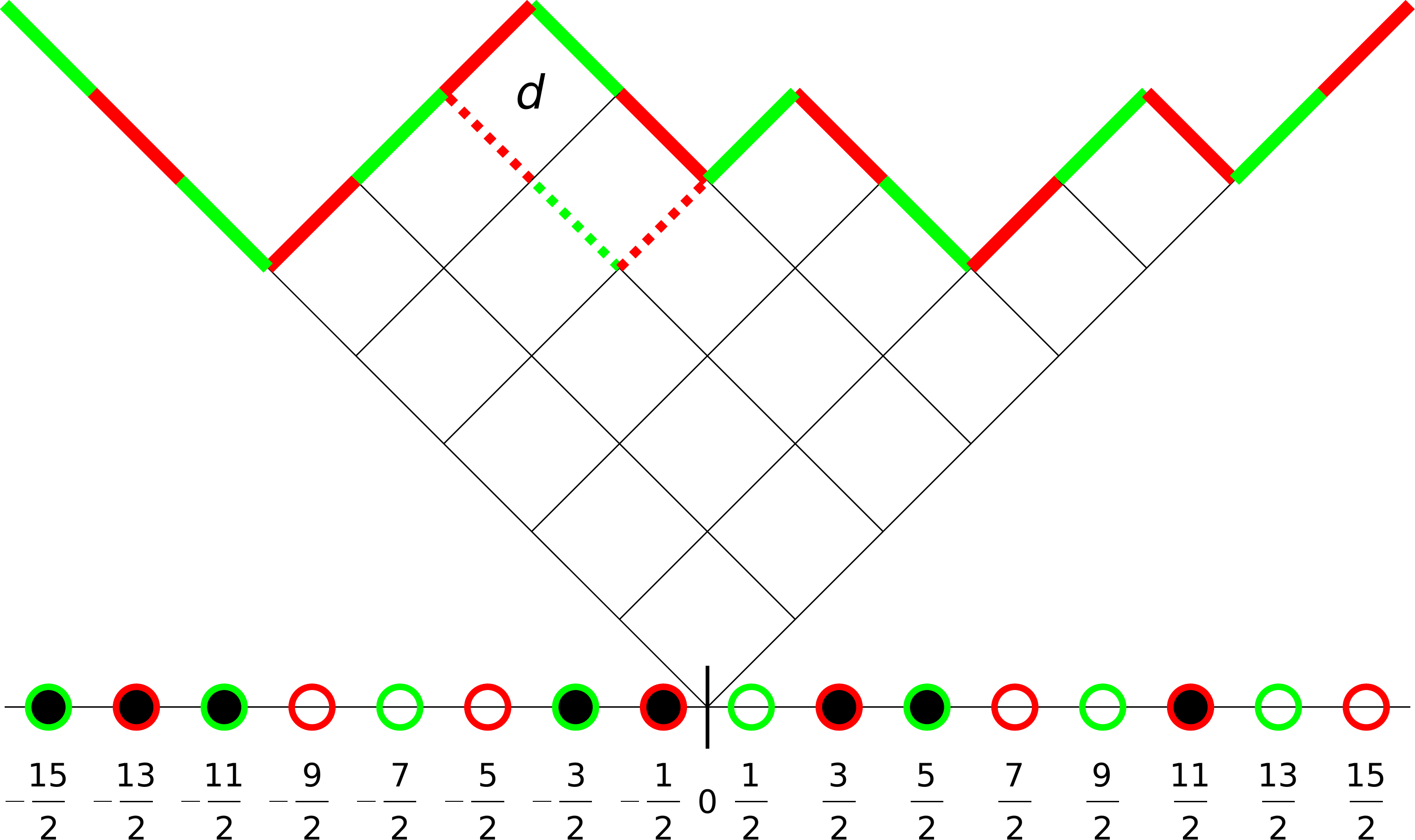}{The Russian and Maya diagrams for $\mu = (6, 4, 4, 3, 3) \vdash 20$.  The segments and pebbles corresponding to the $2$-quotient pair are colored green and red.  The dashed lines outline a $2$-hook that could be removed; $d$~is the square in this $2$-hook with even content (namely, $-2$).}{fig:partition-diagram-example}
\myfig{.375}{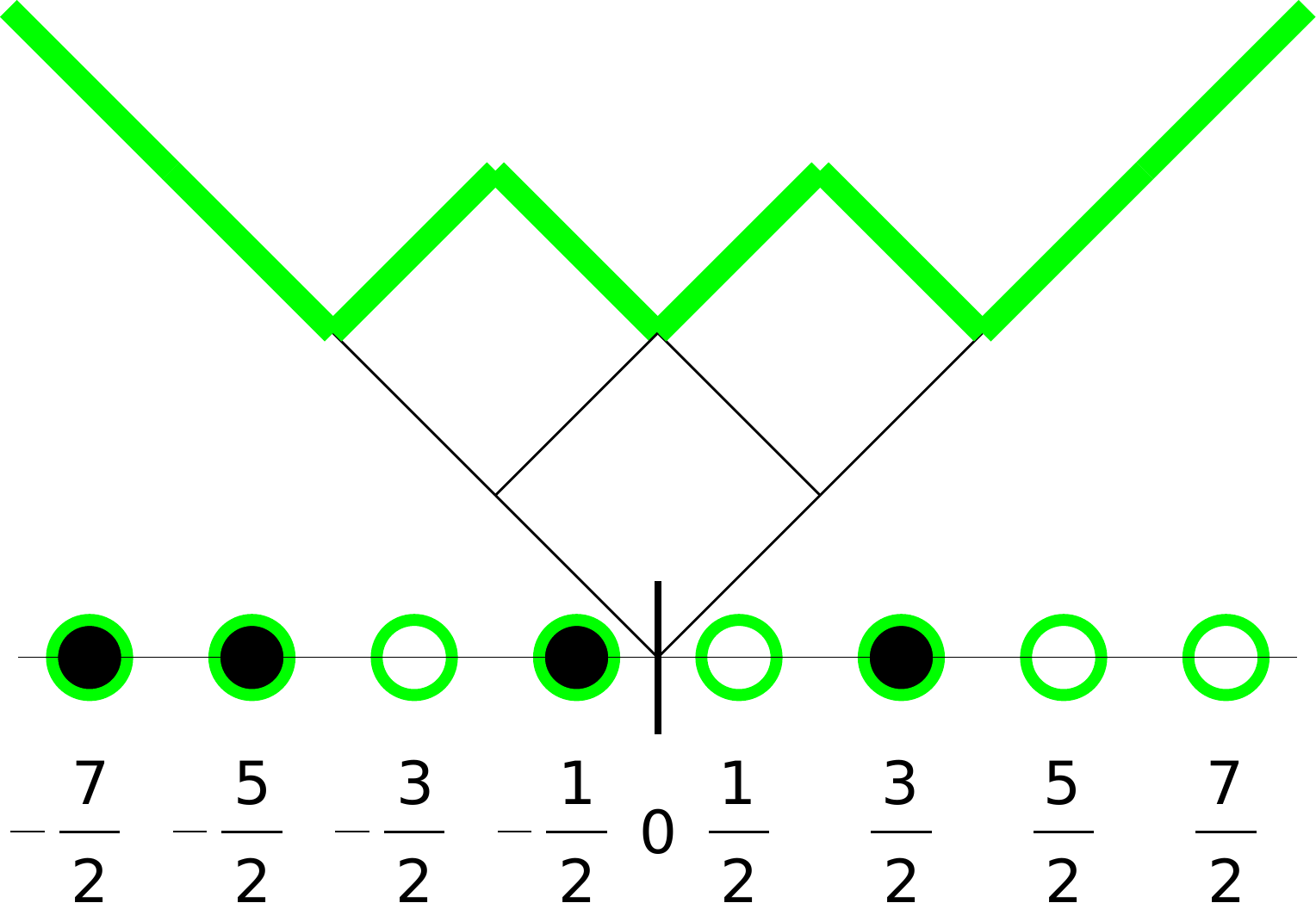}{The diagram for $2$-quotient partition $\mu^{(0)} = (2, 1) \vdash 3$.}{fig:partition-diagram-example-quotient0}
\myfig{.375}{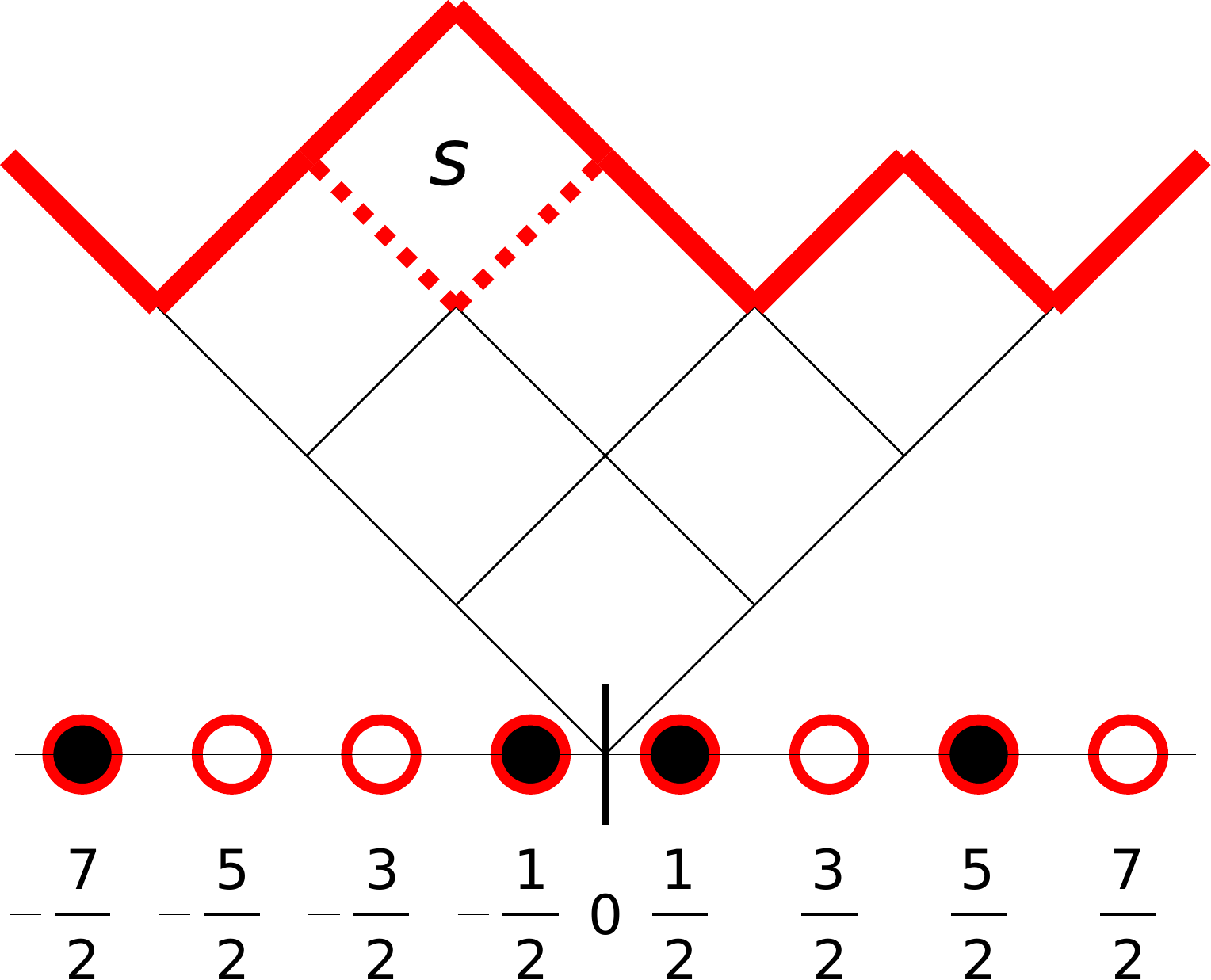}{The diagram for $2$-quotient partition $\mu^{(1)} = (3,2,2) \vdash 7$. The $1$-hook square~$s$ (with content~$-1$) is associated to the $2$-hook in Figure~\ref{fig:partition-diagram-example} that contains square~$d$.}{fig:partition-diagram-example-quotient1}

\begin{definition}
    Let $\mu$ be a partition.  From the Maya diagram for~$[\mu]$, form two new Maya diagrams by taking the two alternating sequences of pebbles.  More precisely, for $b \in \{0,1\}$, let $\mu^{(b)}$ denote the partition whose Maya diagram is formed by the pebbles at positions $2z + (-1)^b \frac12$, $z \in \Z$.  (See Figure~\ref{fig:partition-diagram-example}, in which $b = 0$ is associated to green and $b = 1$ is associated to red.) The pair $(\mu^{(0)}, \mu^{(1)})$ is called the \emph{$2$-quotient} of~$\mu$.  (See Figures~\ref{fig:partition-diagram-example-quotient0},~\ref{fig:partition-diagram-example-quotient1} respectively.)
\end{definition}
\begin{remark}                                      \label{rem:centered-maya}
    Note that when the Maya diagrams for $\mu^{(0)}$, $\mu^{(1)}$ are formed, each of the two origin mark positions may need to be adjusted from the former origin mark position coming from $\mu$'s origin mark.  It is a fact (see, e.g.,~\cite[Section~2.1]{Rio12}) that $\mu$ is balanced if and only if \emph{neither} origin mark position must be adjusted.
\end{remark}
\begin{fact}                                        \label{fact:hook-pebbles}
    A $2$-hook in $[\mu]$ naturally corresponds to a sequence of three pebbles in $[\mu]$'s Maya diagram of the form {(white, \textasteriskcentered, black)}. (See the dashed domino containing the label~$d$ in Figure~\ref{fig:partition-diagram-example}.) In turn, this corresponds to a ``$1$-hook'' in one of $\mu^{(0)}, \mu^{(1)}$; i.e., a square on the rim whose removal leaves a valid Young diagram (see the square labeled~$s$ in Figure~\ref{fig:partition-diagram-example-quotient1}).  Removal of the $2$-hook from~$[\mu]$ corresponds to replacing the sequence {(white, \textasteriskcentered, black)} by {(black, \textasteriskcentered, white)}.  (One thinks of the ``filled'' black pebble as jumping two positions to the left, onto the ``empty'' white pebble.)  In turn, this corresponds to removing the associated $1$-hook from either $\mu^{(0)}$ or $\mu^{(1)}$.
\end{fact}

We will require the following lemma.  It is likely to be known; however we were unable to find its statement in the literature.  The analogous lemma for hook lengths is well known (see, e.g.,~\cite[Lemma~2.1.ii]{Rio12}).
\begin{lemma}                                       \label{lem:quotient-content}
    Let $\mu \vdash k$ be a balanced partition with $2$-quotient $(\mu^{(0)}, \mu^{(1)})$.  Then the multiset $\{c(\square) : \square \in [\mu^{(0)}], \square \in [\mu^{(1)}]\}$ is equal to the multiset $\{\tfrac12 c(\square) : \square \in \evencontents{\mu}\}$.
\end{lemma}
\begin{proof}
    The statement is proved by induction on the deconstruction of~$\mu$ from $2$-hooks, with the base case being $\mu = \emptyset$.  We rely on the fact that since~$\mu$ is balanced, the Maya diagrams of~$\mu^{(0)}$ and~$\mu^{(1)}$ can be seen alternating within the Maya diagram for~$\mu$, with all three origin markers ``lining up'' (see Remark~\ref{rem:centered-maya}). By way of induction, suppose we consider the removal of some $2$-hook $D$ from~$[\mu]$.  This corresponds (see Fact~\ref{fact:hook-pebbles}) to removing a $1$-hook (square)~$s$ from $\mu^{(b)}$, for some $b \in \{0,1\}$. Exactly one of $D$'s two squares is in~$\evencontents{\mu}$;  call that square~$d$.  (See Figures~\ref{fig:partition-diagram-example},~\ref{fig:partition-diagram-example-quotient1} for illustration.) By induction, it suffices to show that $\frac12 c(d) = c(s)$.  But this is easily seen from the combination of the Russian and Maya diagrams, as the content of a square is simply the horizontal displacement of its center.
\end{proof}

We are now ready to establish a formula for $s_\mu(+1, -1, +1, -1, \dots)$.

\begin{theorem}\label{thm:plus-minus-one}
    Let $\mu \vdash k$ and let $d$ be even.  Then
    \begin{equation*}
        s_\mu(\underbrace{+1, -1, +1, -1,\dots}_{d \textnormal{ entries}}) =
            \begin{cases}
                0\vphantom{{\displaystyle \underbrace{\frac{1}{k}}}_{2}} & \textnormal{if $\mu$ is not balanced,} \\
                \chi_\mu(\underbrace{2,2,\dots,2}_{k/2 \textnormal{ entries}})
                        \cdot {\displaystyle \frac{1}{\df{k}}}
                        \cdot (\rising{d}{\evencontents{\mu}})
                    & \textnormal{if $\mu$ is balanced.}
            \end{cases}
    \end{equation*}
\end{theorem}
\begin{proof}
    The first part of the proof relies on a formula from~\cite[Theorem 4.3]{RSW04}, specialized to the case of $\text{``}t\text{''} = 2$:
    \begin{multline*}
        s_\mu(\underbrace{+1, -1, +1, -1,\dots}_{d \textnormal{ entries}}) \\
         =  \begin{cases}
                0 & \textnormal{if $\mu$ is not balanced,} \\
                \sgn(\chi_\mu(\underbrace{2,2,\dots,2}_{k/2 \textnormal{ entries}})) \cdot s_{\mu^{(0)}}(\underbrace{1, 1, \dots, 1}_{d/2 \textnormal{ entries}}) \cdot s_{\mu^{(1)}}(\underbrace{1, 1, \dots, 1}_{d/2 \textnormal{ entries}})
                    & \textnormal{if $\mu$ is balanced,}
            \end{cases}
    \end{multline*}
    where $(\mu^{(0)},\mu^{(1)})$ is the $2$-quotient of~$\mu$.    Thus it suffices to show
    \begin{equation}                                \label{eqn:222}
        s_{\mu^{(0)}}(1, 1, \dots, 1) \cdot s_{\mu^{(1)}}(1, 1, \dots, 1) = \frac{|\chi_\mu(2,2,\dots,2)| \cdot (\rising{d}{\evencontents{\mu}})}{(k/2)! \cdot 2^{k/2}}
    \end{equation}
    assuming $\mu$ is balanced.  Applying Proposition~\ref{prop:s111}, the left-hand side of~\eqref{eqn:222} is
    \[
        \frac{(\rising{\frac{d}{2}}{\mu^{(0)}})\cdot (\rising{\frac{d}{2}}{\mu^{(1)}}) \cdot \dim \mu^{(0)} \cdot \dim \mu^{(1)}}{|\mu^{(0)}|! \cdot |\mu^{(1)}|!}.
    \]
    Next, we appeal to~\cite[formula~(2.2)]{Rio12}, which states
    \[
        \chi_\mu(2, 2, \dots, 2) = \sigma_\mu \cdot \binom{|\mu|/2}{|\mu^{(0)}|,|\mu^{(1)}|} \cdot \dim \mu^{(0)} \cdot \dim \mu^{(1)},
    \]
    where $\sigma_\mu \in \{\pm 1\}$ is a certain sign.    Thus to verify~\eqref{eqn:222} it remains to show
    \begin{equation}                                \label{eqn:222222}
        (\rising{\tfrac{d}{2}}{\mu^{(0)}})\cdot (\rising{\tfrac{d}{2}}{\mu^{(1)}}) = \frac{\rising{d}{\evencontents{\mu}}}{2^{k/2}}.
    \end{equation}
    But this follows immediately from Lemma~\ref{lem:quotient-content}.
\end{proof}

\subsection{Wrapping up the lower bound}
In this section we complete the proof of Theorem~\ref{thm:main-pan-lower}.  We begin by applying Corollary~\ref{cor:paninski-chi-squared} with $x = (+2\eps, -2\eps, +2\eps, -2\eps, \dots)$.  Using Theorem~\ref{thm:plus-minus-one} and the homogeneity of Schur polynomials, we obtain the following after a few manipulations:
\begin{theorem}                                     \label{thm:paninski-pre-estimation}
    For $d$ even and $0 \leq \eps \leq \frac12$,
    \begin{equation}                        \label{eqn:pan0}
          \dchi{\SWdens{n}{\pan{\eps}{d}}}{\SWunif{n}{d}}
        = \sum_{k = 2, 4, 6, \dots} \falling{n}{k} (2\eps)^{2k}d^{-k} \cdot \Paren{\frac{1}{\df{k}^2} \sum_{\substack{\mu \vdash k \textnormal{ balanced}\\ 0 < \ell(\mu) \leq d}} \chi_\mu(2, \dots, 2)^2\cdot\frac{\rising{d}{\evencontents{\mu}}}{\rising{d}{\oddcontents{\mu}}}}.
    \end{equation}
\end{theorem}
To estimate this quantity we will use the following very crude bound:
\begin{proposition}                                     \label{prop:even-over-odd}
    Let $d \in \Z^+$ and let $\mu \vdash k$ be balanced, with $0 < \ell(\mu) \leq d$.  Then
    \begin{equation}            \label{eqn:estim-pan1}
        \frac{\rising{d}{\evencontents{\mu}}}{\rising{d}{\oddcontents{\mu}}} \leq 2^{k/2}.
    \end{equation}
\end{proposition}
\begin{proof}
    Fix any domino-tiling for~$\mu$.  Each of the $k/2$ dominoes contains one cell of even content~$c_e$ and one cell of odd content~$c_o$, with $|c_e - c_o| = 1$.  Thus each contributes a factor of $\frac{d+c_e}{d+c_o} \leq \frac21 = 2$ to $(\rising{d}{\evencontents{\mu}})/(\rising{d}{\oddcontents{\mu}})$.
\end{proof}
By character orthogonality relations we also have
\begin{equation}                \label{eqn:estim-pan2}
    \sum_{\substack{\mu \vdash k \textnormal{ balanced}\\ 0 < \ell(\mu) \leq d}} \chi_\mu(2, \dots, 2)^2 \leq \sum_{\mu \vdash k} \chi_\mu(2, \dots, 2)^2 = z_{(2, \dots, 2)} = \df{k}.
\end{equation}
Combining~\eqref{eqn:estim-pan1},~\eqref{eqn:estim-pan2}, we get that the parenthesized expression in~\eqref{eqn:pan0} is at most $2^{k/2}/\df{k} = 1/(k/2)!$.  Using also $\falling{n}{k} \leq n^k$, the right-hand side of~\eqref{eqn:pan0} is thus bounded by
\[
    \sum_{k = 2, 4, 6, \dots} n^k (2\eps)^{2k}d^{-k}/(k/2)! = \exp((4n\eps^2/d)^2) - 1,
\]
completing the proof of Theorem~\ref{thm:main-pan-lower}.\\

We end by indicating how to obtain the testing lower bound in the case when $d \geq 3$ is odd.  In this case we define $\pan{\eps}{d}$ to be $(\frac{1+2\eps}{d}, \frac{1-2\eps}{d}, \dots, \frac{1+2\eps}{d}, \frac{1-2\eps}{d}, \frac{1}{d})$.  This distribution has $\dtvsymm{\pan{\eps}{d}}{\unif{d}} = \frac{d-1}{d}\eps \geq \frac23 \eps$; since this differs from~$\eps$ only by a constant factor, the lower bound of $\Omega(d/\eps^2)$ is not affected.  Now Corollary~\ref{cor:paninski-chi-squared} is applied with $x = (+2\eps, -2\eps, \dots, +2\eps, -2\eps, 0)$.  By stability of the shifted Schur polynomials we have $s_\mu(+1, -1, \dots, +1, -1, 0) = s_\mu(+1, -1, \dots, +1, -1)$, where there are $d-1$ entries in the latter. Now we get $\chi_\mu(2,2,\dots,2)\cdot {\frac{1}{\df{k}}} \cdot \rising{(d-1)}{\evencontents{\mu}}$ out of Theorem~\ref{thm:plus-minus-one}, and we can simply upper-bound $(d-1)$ by~$d$ and proceed with the remainder of the proof.

\section{Hardness of distinguishing uniform distributions}\label{sec:uniform}

In this section, we prove Theorem~\ref{thm:unif-dist}, namely that $O(r^2/\Delta)$ copies are sufficient to distinguish between the cases when $\rho$'s spectrum is uniform on either $r$ or $r+\Delta$ eigenvalues ($1 \leq \Delta \leq r$), and that $\wt{\Omega}(r^2/\Delta)$ copies are necessary.  To be more precise, our lower bound on the number of copies~$n$ will be
\begin{equation}    \label{eqn:precise-rank-lower-bound}
    n \geq r^{2 - O(1/\log^{.33} \r)}/\Delta.
\end{equation}

\subsection{The upper bound}

The proof of the upper bound is quite similar to that of Theorem~\ref{thm:mixedness-upper} for the Mixedness Tester.  We employ the following tester:
\begin{nameddef}{Uniform Distribution Distinguisher}
Given $\rho^{\otimes n}$:
\begin{enumerate}
\item Sample $\blambda \sim \SWdens{n}{\rho}$.
\item Accept if $\psharp{2}(\blambda) \leq e \coloneqq n(n-1) \cdot \frac12\left(\frac{1}{r} + \frac{1}{r+\Delta}\right)$.  Reject otherwise.
\end{enumerate}
\end{nameddef}

\noindent As for the analysis, from Equations~\eqref{eq:thanks-maple-avg} and~\eqref{eq:thanks-maple-var}:
\begin{equation*}
\E_{\blambda \sim \SWunif{n}{m}}\left[\psharp{2}(\blambda)\right] = \frac{n(n-1)}{m},
\quad\text{and}\quad
\Var_{\blambda \sim \SWunif{n}{m}}\left[\psharp{2}(\blambda)\right]
\leq 2n(n-1).
\end{equation*}
We see that the variance is the same whether $m = r$ or $m = r+\Delta$; only the expectation is different, and the tester's acceptance cutoff~$e$ is precisely the midway point between the two expectations.  If $m = r$, then Chebyshev's inequality implies
\begin{equation*}
    \Pr_{\blambda \sim \SWunif{n}{m}}\left[\psharp{2}(\blambda) \geq e\right] \leq \frac{8 r^2 (r+\Delta)^2}{n (n-1)\Delta^2} \leq \frac{32 r^4}{(n-1)^2 \Delta^2},
\end{equation*}
and we have the same upper bound by Chebyshev for $\Pr_{\blambda \sim \SWunif{n}{m}}\left[\psharp{2}(\blambda) \leq e\right]$ when $m = r+\Delta$.  This upper bound is at most $1/3$ provided $n \geq 4\sqrt{6} \cdot \frac{r^2}{\Delta} + 1$, completing the proof of the upper bound in Theorem~\ref{thm:unif-dist}.

The end of Section~\ref{sec:one-sided} gives a different $O(r^2)$-copy tester (the ``Rank Tester'') for the $r$-versus-$(r+1)$ case. In this case it's superior to the Uniform Distribution Distinguisher in that it has one-sided error (i.e., it never rejects in the rank-$r$ case).

\subsection{The lower bound}

\renewcommand{\r}{r}
\newcommand{\rp}{r_{\!\scriptscriptstyle{+}\!}}

The bulk our work for the lower bound will be devoted to the case of $\Delta = 1$.  The extension to larger~$\Delta$ is very tedious and will be dealt with in Section~\ref{sec:tv-improvement}. So let $\r \in \Z^+$ be a parameter which we think of as tending to infinity, and for brevity let $\rp = \r+1$. Our task is to show that the distributions $\SWunif{n}{\r}$ and $\SWunif{n}{\rp}$ are very close in total variation distance unless $n \geq \wt{\Omega}(r^2)$. For notational convenience we will write
\[
    n = \frac{\r^2}{\omega^2}
\]
and seek to show that $\SWunif{n}{\r}$ and $\SWunif{n}{\rp}$ are close once~$\omega$ is sufficiently large as a function of~$\r$. Ultimately we will select $\omega = \exp(\Theta(\log^{.67} \r))$. For now, though, let's keep~$\omega$ general, subjecting it only to the following assumption:
\begin{equation}                                            \label{eqn:omega-assumption}
    200 \leq \omega \leq \sqrt{\r}.
\end{equation}

\subsubsection{Initial approximations}
It proves more convenient to study the Kullback--Leibler divergence between $\SWunif{n}{\r}$ and $\SWunif{n}{\rp}$:
\begin{align}
    \dkl{\SWunif{n}{\r}}{\SWunif{n}{\rp}}
        &= \E_{\blambda \sim \SWunif{n}{\r}}
                \left[\ln\left(\frac{\SWunif{n}{\r}[\blambda]}{\SWunif{n}{\rp}[\blambda]}\right)\right] \nonumber\\
        &= \E_{\blambda \sim \SWunif{n}{\r}}
                \left[\ln\left(\frac{\rp^{n}}{\r^{n}}\cdot\frac{\rising{\r}{\blambda}}{\rising{\rp}{\blambda}}\right)\right] \nonumber\\
        &= n \ln\left(\frac{\rp}{\r}\right) + \E_{\blambda \sim \SWunif{n}{\r}}
                \left[\ln\left(\frac{\prod_{\square \in [\blambda]} (\r + c(\square))}{\prod_{\square \in [\blambda]} (\rp + c(\square))}\right)\right],     \label{eqn:dkl}
\end{align}
where the second equality used Proposition~\ref{prop:SW-dist}.  (We remark that the logarithms above are always finite since $\supp(\SWunif{n}{\r}) \subseteq \supp(\SWunif{n}{\rp})$.)

Recalling that $\rp = r+1$, it is very easy to verify (cf.~\cite[Exercise I.1.11]{Mac95}, \cite[Section~2.5]{CGS04}) that the large fraction inside the inner logarithm of~\eqref{eqn:dkl} is equal to
\[
    \prod_{i=1}^{\ell(\blambda)} \frac{\r - (i-1)}{\r - (i-1-\blambda_i)} = \Phi(-(\r+\tfrac12); \blambda),
\]
where $\Phi$ denotes a generating function for the modified Frobenius coordinates, defined in~\cite{IO02} and similar to the ``Frobenius function'' from~\cite{Las08,CSST10}.  Proposition~1.2 in~\cite{IO02} observes that
\[
    \Phi(z;\lambda) = \prod_{i}\frac{z+b^*_i}{z-a^*_i},
\]
where the $a^*_i$'s and $b^*_i$'s are the modified Frobenius coordinates of~$\lambda$; as a consequence, Proposition~1.4 in~\cite{IO02} states that
\begin{equation}                                \label{eqn:io-formal}
    \ln \Phi(z; \lambda) = \sum_{k=1}^\infty \frac{\pstar{k}(\lambda)}{k}z^{-k}.
\end{equation}
However we cannot immediately take $z = -(\r+\tfrac12)$ and conclude
\begin{equation}                                            \label{eqn:dkl2}
    \eqref{eqn:dkl} \mathrel{\overset{\scriptstyle ?}{=}} n \ln\left(1+\frac{1}{\r}\right) + \E_{\blambda \sim \SWunif{n}{\r}}
                        \left[\sum_{k=1}^\infty \frac{(-1)^k \pstar{k}(\blambda)}{k(\r+\tfrac12)^k}\right]
\end{equation}
because~\eqref{eqn:io-formal} is merely a formal identity of generating functions and does not hold for all real~$z$. More specifically, it's necessary that the Taylor series for $\ln(1+b_i/z)$ and $\ln(1-a_i/z)$ converge, which happens provided $|b_i/(\r+\tfrac12)|, |a_i/(\r+\tfrac12)| \leq 1$.  These conditions are equivalent to $\ell(\blambda) = \blambda_1' \leq \r+1$ and $\blambda_1 \leq \r+1$.  The first condition is automatic, since $\blambda \sim \SWunif{n}{\r}$.  The second condition does not always hold; however, we will show (see Lemma~\ref{lem:lambda-usual} below) that it holds with overwhelming probability when $n \ll \r^2$.  Indeed the ``central limit theorems'' for the Schur--Weyl distributions suggest that both $\blambda_1$~and~$\blambda'_1$ will almost always be~$O(\sqrt{n}) = O(\frac{\r}{\omega})$.  Let us therefore make a definition:
\begin{definition}
    We say that $\lambda \vdash n$ is \emph{usual} if $\lambda_1, \lambda'_1 \leq \frac{10}{\omega}\r$.  Since we are assuming $\omega \geq 200$, usual $\lambda$'s satisfy $\lambda_1, \lambda_1' \leq \frac{1}{20}\r \leq \r+1$.
\end{definition}
Thus when $\lambda$ is usual we may apply~\eqref{eqn:dkl2}.  Since the quantity inside the expectation in~\eqref{eqn:dkl} is clearly always negative, we may write
\begin{align}
    \dkl{\SWunif{n}{\r}}{\SWunif{n}{\rp}} = \eqref{eqn:dkl} &\leq n \ln\left(1+\frac{1}{\r}\right) +
            \E_{\blambda \sim \SWunif{n}{\r}}
                \left[1_{\{\blambda \text{ usual}\}} \cdot \ln\left(\frac{\prod_{\square \in [\blambda]} (\r + c(\square))}{\prod_{\square \in [\blambda]} (\rp + c(\square))}\right)\right] \nonumber\\
        &= n \ln\left(1+\frac{1}{\r}\right) +
            \E_{\blambda \sim \SWunif{n}{\r}}
                \left[1_{\{\blambda \text{ usual}\}} \cdot \sum_{k=1}^\infty \frac{(-1)^k \pstar{k}(\blambda)}{k(\r+\tfrac12)^k}\right] \nonumber\\
        &= n \ln\left(1+\frac{1}{\r}\right) - \frac{1}{\r+\tfrac12}\cdot\E_{\blambda \sim \SWunif{n}{\r}}\left[1_{\{\blambda \text{ usual}\}} \cdot \pstar{1}(\blambda)\right]             \label{eqn:dkl-p1}\\
        & \hspace{1.05in} +
             \E_{\blambda \sim \SWunif{n}{\r}}
                \left[1_{\{\blambda \text{ usual}\}} \cdot \sum_{k=2}^\infty\frac{(-1)^k \pstar{k}(\blambda)}{k(\r+\tfrac12)^k}\right]. \label{eqn:dkl-prest}
\end{align}
Recall  that $\pstar{1}(\lambda)$ is simply $|\lambda|$; thus the expectation in~\eqref{eqn:dkl-p1} is simply $n \Pr[\blambda \text{ usual}]$.  As Lemma~\ref{lem:lambda-usual} below shows, $\Pr[\blambda \text{ usual}] = 1 - \delta$ for $\delta \lll \frac{1}{60\r^2}$.  Thus:
\begin{equation}                                    \label{eqn:dkl-cancel}
    \eqref{eqn:dkl-p1} = n\left(\ln\left(1+\frac{1}{\r}\right) - \frac{1}{\r+\tfrac12} + \frac{\delta}{r+\tfrac12}\right) \leq n\left(\frac{1}{12\r^3} + \frac{1/(60\r^2)}{\r+\tfrac12}\right) \leq \frac{n}{10\r^3} = \frac{1}{10 \omega^2 \r}.
\end{equation}
\begin{lemma}                                       \label{lem:lambda-usual}
    Let $\blambda \sim \SWunif{n}{\r}$.  Then
    $
        \Pr[\blambda \textnormal{ unusual}] \leq 2^{-20 \r/\omega}.
    $
\end{lemma}
\begin{proof}
    Write $B = \lceil \frac{10}{\omega}\r \rceil$.  By Proposition~\ref{prop:increasing-seq} and the fact that $B \leq r$,
    \[
        \Pr[\blambda_1 \geq B],\Pr[\blambda_1' \geq B] \leq  \left(\frac{2e^2 n}{B^2}\right)^B \leq \left(\frac{2e^2}{100}\right)^{10\r/\omega} \leq 2^{-1 -20 \r/\omega}.
    \]
    The lemma now follows from the union bound.
\end{proof}


\renewcommand{\L}{L}
\renewcommand{\c}{C}
Turning to~\eqref{eqn:dkl-prest}, let's write
\[
    \L^*_\c(\lambda) \coloneqq \sum_{k=2}^C\frac{(-1)^k \pstar{k}(\lambda)}{k(\r+\tfrac12)^k},
\]
recalling that $\L^*_\infty(\lambda)$ is definitely convergent if $\lambda$ is usual. The infinite sum in~\eqref{eqn:dkl-prest} is inconvenient, as is the $+\frac12$ in the denominator.  We clean these issues up with the following lemma:
\begin{lemma}                                       \label{lem:down-with-infinity}
    Assuming $\lambda \vdash n$ is usual, if
    \[
        \c \geq \frac{3\log(10\r)}{\log(\omega/10)},
    \]
    it follows that
    \[
        \ABS{\L^*_\infty(\lambda) - \L_\c(\lambda)} \leq \frac{201}{\omega^3},
    \]
    where $\L_\c(\lambda)$ denotes the same quantity as $\L^*_\c(\lambda)$ except with no $+\frac12$ in the denominator.
\end{lemma}
\begin{proof}
    For any $\lambda \vdash n$ (not necessarily usual), we have the crude bound $|\pstar{k}(\lambda)| \leq 2\sqrt{n} B^k$ whenever $\lambda_1, \lambda_1' \leq B$.  This is because each modified Frobenius coordinate $a^*_i$ or $b^*_i$ (of which there are at most $\sqrt{n}$ each) is at most~$B$.
    For \emph{usual}~$\lambda$ we may take $B = \frac{10}{\omega}\r$.  Thus we have
    \[
        \ABS{\L^*_\infty(\lambda) - \L^*_\c(\lambda)} \leq \sum_{k = \c+1}^\infty \frac{|\pstar{k}(\lambda)|}{k(\r+\tfrac12)^k} \leq \sum_{k = \c+1}^\infty \frac{2\frac{\r}{\omega} (10\frac{\r}{\omega})^k}{k\r^k} \leq 2\r\sum_{k=\c+1}^\infty \left(\frac{10}{\omega}\right)^k \leq 4r\left(\frac{10}{\omega}\right)^\c \leq \frac{1}{250\r^2},
    \]
    where the last inequality used the assumption about~$\c$ (and the second-to-last inequality used $\omega \geq 200$ in a crude way). Further,
    \[
        \ABS{\L^*_\c(\lambda) - \L_\c(\lambda)} \leq \sum_{k=2}^\c \frac{|\pstar{k}(\lambda)|}{k}\left(\frac{1}{\r^k} - \frac{1}{(\r+\tfrac12)^k}\right) \leq  \sum_{k=2}^\c \frac{2\frac{\r}{\omega} (10\frac{\r}{\omega})^k}{k}\left(\frac{k}{2\r^{k+1}}\right) = \frac{1}{\omega}\sum_{k=2}^\c\left(\frac{10}{\omega}\right)^k \leq \frac{200}{\omega^3}.
    \]
    Finally, $\frac{200}{\omega^3} + \frac{1}{250\r^2} \leq \frac{201}{\omega^3}$ by our assumption~\eqref{eqn:omega-assumption} that $\omega \leq \sqrt{\r}$.
\end{proof}
Let us use this lemma in~\eqref{eqn:dkl-prest}, and also apply~\eqref{eqn:dkl-cancel} in~\eqref{eqn:dkl-p1}.  Assuming  the lemma's hypotheses, we obtain
\begin{align*}
    \dkl{\SWunif{n}{\r}}{\SWunif{n}{\rp}} &\leq \E_{\blambda \sim \SWunif{n}{\r}}\left[1_{\{\blambda \textnormal{ usual}\}}\cdot \L_\c(\blambda)\right] + \tfrac{1}{10\omega^2\r} + \tfrac{201}{\omega^3}\\
    &\leq \E_{\blambda \sim \SWunif{n}{\r}}\left[\L_\c(\blambda)\right] - \E_{\blambda \sim \SWunif{n}{\r}}\left[1_{\{\blambda \text{ unusual}\}}\cdot \L_\c(\blambda)\right] +  \tfrac{202}{\omega^3}.
\end{align*}
We can use Cauchy--Schwarz to bound
\begin{equation}                                        \label{eqn:cruel-and-unusual}
    \ABS{\E_{\blambda \sim \SWunif{n}{\r}}\left[1_{\{\blambda \text{ unusual}\}}\cdot \L_\c(\blambda)\right]} \leq \sqrt{\E[1_{\{\blambda \text{ unusual}\}}^2]}\sqrt{\E[\L_\c(\blambda)^2]} \leq 2^{-10\r/\omega} \sqrt{\E[\L_\c(\blambda)^2]},
\end{equation}
where the last inequality used Lemma~\ref{lem:lambda-usual}.  Finally, we can afford to use an extraordinarily crude bound on~$\E[\L_\c(\blambda)^2]$:
\[
    \E[\L_\c(\blambda)^2] \leq \c \sum_{k=2}^\c\E[\pstar{k}(\blambda)^2] \leq \c \sum_{k=2}^\c (2\sqrt{n} n^k)^2 \leq n^{3\c} \leq \r^{6\c},
\]
where the second inequality used the crude bound on~$|\pstar{k}(\lambda)|$ from the proof of Lemma~\ref{lem:down-with-infinity}. (In fact, in Section~\ref{sec:tv-improvement} we will actually show that this quantity is quite tiny.) If we now make the very weak assumption that $\c \leq \frac{3\r}{\omega \log \r}$,  we may conclude $\eqref{eqn:cruel-and-unusual} \leq 2^{-\r/\omega} \ll \frac{1}{\omega^3}$.

Now we can summarize all of the preparatory work we have done so far:
\begin{proposition}                                     \label{prop:rank-warmup}
    Assuming $\frac{3\log(10\r)}{\log(\omega/10)} \leq \c \leq \frac{3\r}{\omega \log \r}$, for $\blambda \sim \SWunif{n}{\r}$ we have
    \[
        \dkl{\SWunif{n}{\r}}{\SWunif{n}{\rp}} \leq \E\left[\L_\c(\blambda)\right] + \tfrac{203}{\omega^3},
    \]
    where
    \begin{equation}                            \label{eqn:LC}
        \L_\c(\lambda) \coloneqq \sum_{k=2}^\c\frac{(-1)^k \pstar{k}(\lambda)}{k\r^k}.  
    \end{equation}
\end{proposition}
\noindent (It is straightforward to check using~\eqref{eqn:omega-assumption} that the range of values for~$\c$ is nonempty.)\\

We now come to the main task: showing that $\E[\L_\c(\blambda)]$ is small.

\subsubsection{Passing to the $\psharp{}$ polynomials}                      \label{sec:L-psharp}
\newcommand{\partfact}[1]{\mathrm{fact}(#1)}
In this section and the following one, we will use the notation
\[
    \partfact{\mu} = \prod_{w \geq 1} m_w(\mu)!
\]
where, recall, $m_w(\mu)$ is the number of parts of~$\mu$ equal to~$w$.

The following proposition is essentially immediate from known formulas:
\begin{proposition}                                     \label{prop:p-to-psharp}
    For any $k \in \Z^+$, we have the following identity on observables:
    \[
        \pstar{k} = \sum_{\mu~:~\weight(\mu) = k+1} \frac{\falling{k}{(\ell(\mu)-1)}}{\partfact{\mu}} \psharp{\mu} + \calO_{k},
    \]
    where $\calO_{k}$ is an observable with $\weight(\calO_{k}) \leq k$.  More precisely,
    \[
        \calO_{k} = \sum_{\mu~:~\weight(\mu) \leq k} c_{k,\mu} \psharp{\mu}
    \]
    for some rational coefficients $c_{k,\mu}$.
\end{proposition}
\begin{proof}
    From~\cite[Corollary~2.8]{IO02} we have
    \[
        \pstar{k} = \frac{1}{k+1}\cdot\wt{p}_{k+1} + \Bigl\{\text{a linear combination of $\wt{p}_k, \dots, \wt{p}_2$}\Bigr\}.
    \]
    From~\cite[Corollary~3.7]{IO02} (cf.~\cite[Lemma 10.10]{Mel10b}) we have
    \[
        \wt{p}_{k+1} = \sum_{\mu~:~\weight(\mu) = k+1} \frac{\falling{(k+1)}{\ell(\mu)}}{\partfact{\mu}} \prod_{i \geq 1} (\psharp{i})^{m_i(\mu)}.
    \]
    The result is now easily deduced from Proposition~\ref{prop:should-be-in-a-structure-constants-discussion-or-theorem}.
\end{proof}

Substituting the above result into~\eqref{eqn:LC} yields:
\begin{equation}                                    \label{eqn:LC-approx}
    \L_\c(\lambda) =  \sum_{k=2}^\c \frac{(-1)^k}{k\r^k} \cdot \sum_{\weight(\mu) = k+1} \frac{\falling{k}{(\ell(\mu)-1)}}{\partfact{\mu}} \psharp{\mu}(\lambda) + \sum_{k=2}^\c \frac{(-1)^k \calO_{k}(\lambda)}{k\r^k}.
\end{equation}
Taking the expectation over $\blambda \sim \SWunif{n}{\r}$, and using Corollary~\ref{cor:sw-expect} to evaluate the expectation of $\psharp{\mu}$, we obtain:
\begin{align}
    \E_{\blambda \sim \SWunif{n}{\r}}[\L_\c(\blambda)]
        &=  \sum_{k=2}^\c \frac{(-1)^k}{k\r^k} \cdot \sum_{\weight(\mu) = k+1} \frac{\falling{k}{(\ell(\mu)-1)}}{\partfact{\mu}} \falling{n}{|\mu|} \r^{\ell(\mu) - |\mu|} \label{eqn:EL} \\
        & {}+ \sum_{k=2}^\c \frac{(-1)^k\E_{\blambda \sim \SWunif{n}{\r}}[\calO_{k}(\blambda)]}{k\r^k}. \label{eqn:errL}
\end{align}
We will show in Lemmas~\ref{lem:errL},~\ref{lem:errL2} below that the ``error term''~\eqref{eqn:errL} is small assuming $n \ll \r^2$.  Thus we focus on~\eqref{eqn:EL}.

\subsubsection{Showing the ``main term'' is small: some intuition}
Before diving into manipulations, let's take a high-level look at the contributions to~\eqref{eqn:EL} from $k = 2, 3, 4, 5, \dots$, focusing on the powers of~$n$ and~$\r$.  First consider the case of $k = 2$.  Here there is only one $\mu$ with $\weight(\mu) = 3$, namely $\mu = (2)$, which has $|\mu| = 2$ and $\ell(\mu) = 1$.  Thus from $k = 2$ we pick up a factor on the order of $\frac{n^2}{\r^3}$; more precisely, $\frac{\falling{n}{2}}{2\r^3}$.  This looks rather bad from the point of view of proving a quadratic lower bound for~$n$: the term $\frac{\falling{n}{2}}{2\r^3}$ is not small unless $n \ll r^{3/2}$.  The main surprise in our proof is that this term will be exactly canceled by ``lower-degree'' contributions from larger~$k$.

To see an example of this, consider the $k = 3$ contribution in~\eqref{eqn:EL}.  Here there are two $\mu$'s with $\weight(\mu) = 4$, namely $\mu = (3)$ and $\mu = (1,1)$.  The first gives a contribution on the order of $\frac{n^3}{\r^5}$; more precisely, $-\frac{\falling{n}{3}}{3\r^5}$.  The second gives a contribution of $-\frac{\falling{n}{2}}{2\r^3}$, thereby precisely canceling the $k = 2$ term.  Thus we are left (so far) with $-\frac{\falling{n}{3}}{3\r^5}$, which is small if $n \ll \r^{5/3}$.  This is still far from a quadratic bound, but it's better than the~$\r^{3/2}$ bound we were faced with previously.

In turn, the $-\frac{\falling{n}{3}}{3\r^5}$ contribution will be canceled by a certain $k = 3$ term, namely $\frac{\falling{n}{3}}{\r^5}$ from $\mu = (2,1)$, together with a certain $k = 4$ term, namely $\frac{2\falling{n}{3}}{3\r^5}$ from $\mu = (1,1,1)$.  Indeed, if we sum up through $k = 6$, the total contribution is $-\frac{5\falling{n}{4}}{\r^7} - \frac{\falling{n}{5}}{5\r^9}$, which is small if $n \ll \r^{7/4}$.  This gets us still closer to a quadratic bound.

In fact, looking carefully at small partitions suggests that perfect cancelation is achieved if we group contributions according to~$|\mu|$.  This proves to be the case, as we will show below.  In the end \eqref{eqn:EL} does not precisely vanish because for $m > \c/2$, not all~$\mu$'s with $|\mu| = m$ appear in~\eqref{eqn:EL}.  However the ``leftover contributions'' are of the shape $\r (\frac{n}{\r^2})^{k}$ for $k > \c/2$, a quantity we can ensure is small by taking $\omega$ and $\c$ large enough.   (There is a tradeoff involved preventing us from taking~$\c$ too large: our ``error bound''~\eqref{eqn:errL} increases with~$\c$.)

\subsubsection{Proof that the ``main term'' is small}
Although~\eqref{eqn:EL} has a double summation, the summed quantity is simply counted exactly once for each $\mu$ with $3 \leq \weight(\mu) \leq \c+1$.  As suggested above, let us rearrange the summation according to~$|\mu|$.  We will use the notation $s = |\mu| - 1$ and $h = \ell(\mu) - 1$, so that $\weight(\mu) = s + h + 2$ (i.e., $k = s+h+1$) and $\weight(\mu) \leq \c+1 \iff h \leq \c-1-s$:
\begin{align*}
    \eqref{eqn:EL} &= \sum_{s = 1}^{\c-1} \sum_{h = 0}^{\min(s, \c-1-s)} \sum_{\substack{\mu \vdash s+1 \\ \ell(\mu) = h+1}}
                        \frac{(-1)^{s+h+1}}{(s+h+1)\r^{s+h+1}} \frac{\falling{(s+h+1)}{h}}{\partfact{\mu}} \falling{n}{(s+1)}\r^{h-s} \\
                   &= \sum_{s = 1}^{\c-1} (-1)^{s+1} \cdot \frac{\falling{n}{(s+1)}}{\r^{2s+1}}
                      \sum_{h = 0}^{\min(s, \c-1-s)} (-1)^{h}\falling{(s+h)}{(h-1)}
                      \sum_{\substack{\mu \vdash s+1 \\ \ell(\mu) = h+1}} \frac{1}{\partfact{\mu}}.
\end{align*}
(We remark that we switched from $\r+\frac12$ to $\r$ in Lemma~\ref{lem:down-with-infinity} so as to obtain nice cancelations on~$\r$ here. We also recall the convention $\falling{m}{(-1)} = \frac{1}{m+1}$.)  It is not hard to show (see, e.g.,~\cite[Lemma~11]{Mel10a}) that
\[
    \sum_{\substack{\mu \vdash s+1 \\ \ell(\mu) = h+1}} \frac{1}{\partfact{\mu}} = \frac{1}{(h+1)!}\binom{s}{h}.
\]
Substituting this into the above, and also using $\falling{(s+h)}{(h-1)} = \frac{(s+h)!}{(s+1)!}$, we get
\begin{align*}
    \eqref{eqn:EL} &=  \sum_{s = 1}^{\c-1} (-1)^{s+1}\cdot \frac{\falling{n}{(s+1)}}{\r^{2s+1}}
                       \sum_{h = 0}^{\min(s, \c-1-s)} (-1)^{h}
                                                            \frac{(s+h)!}{(s+1)!(h+1)!}\binom{s}{h} \\
                   &=  \sum_{s = 1}^{\c-1} \frac{(-1)^{s+1}}{s+1} \cdot \frac{\falling{n}{(s+1)}}{\r^{2s+1}}
                       \sum_{h = 0}^{\min(s, \c-1-s)} \frac{(-1)^{h}}{h+1} \binom{s+h}{h}\binom{s}{h}.
\end{align*}
We now obtain the promised cancelation.  Specifically, it is a known combinatorial identity (see, e.g.,~\cite[page 182]{GKP94}) that for all $s \in \Z^+$, the inner summation equals~$0$ provided~$h$ ranges all the way up to~$s$.  In other words, all contributions from $s \leq \frac{\c-1}{2}$ vanish.  For larger~$s$, it's not hard to bound the inner ``partial sum'' crudely by, say,~$9^s$ in absolute value.  We therefore finally conclude:
\begin{equation}                                    \label{eqn:EL-bound}
    |\eqref{eqn:EL}| \leq \sum_{\frac{\c}{2} \leq s \leq \c-1} \frac{1}{s+1} \cdot \frac{\falling{n}{(s+1)}}{\r^{2s+1}} \cdot 9^s \leq \frac{n}{\r}   \sum_{s \geq \frac{\c}{2}}\left(\frac{9n}{\r^2}\right)^s = \frac{\r}{\omega^2}\sum_{s \geq \frac{\c}{2}}\left(\frac{9}{\omega^2}\right)^s \leq \r \left(\frac{3}{\omega}\right)^\c.
\end{equation}

\subsubsection{Bounding the ``error term''}

In this section we bound the ``error term''~\eqref{eqn:errL}, using the following lemma:
\begin{lemma}                                       \label{lem:p-sharp-expect-bound}
    Suppose $n = \frac{r^2}{\omega^2}$.  Then $\displaystyle 0 \leq \E_{\blambda \sim \SWunif{n}{\r}}[\psharp{\mu}(\blambda)] \leq \r^{\weight(\mu)} \cdot (1/\omega^2)^{|\mu|}$.
\end{lemma}
\begin{proof}
    By Corollary~\ref{cor:sw-expect},
    $\displaystyle
        \E_{\blambda \sim \SWunif{n}{\r}}\left[\psharp{\mu}(\blambda)\right] = \falling{n}{|\mu|} \r^{\ell(\mu)-|\mu|} \leq n^{|\mu|}\r^{\weight(\mu)-2|\mu|} = \r^{\weight(\mu)} \cdot (1/\omega^2)^{|\mu|}.
    $
\end{proof}
We will first use this lemma to bound~\eqref{eqn:errL} in a ``soft'' way, thinking of~$\c$ as an absolute universal constant.  This is enough to get a testing lower bound like $n \geq \Omega_{\delta}(\r^{2-\delta})$ for every $\delta > 0$.  Subsequently we do some technical work (which the uninterested reader may skip) to get a more explicit lower bound.
\begin{lemma}                                       \label{lem:errL}
    For all $\c \geq 2$ there is a constant $A_\c$ such that $|\eqref{eqn:errL}| \leq A_\c \cdot \frac{1}{\omega^2}$.
\end{lemma}
\begin{proof}
    It suffices to show that for all $k \geq 2$ there is a constant $A'_k$ such that
    \[
        \frac{\E_{\blambda \sim \SWunif{n}{\r}}[\calO_k(\blambda)]}{\r^k} \leq A'_k \cdot \frac{1}{\omega^2}.
    \]
    But recalling Proposition~\ref{prop:p-to-psharp}, the left-hand side is
    \[
        \sum_{\mu~:~\weight(\mu) \leq k} c_{k,\mu} \E_{\blambda \sim \SWunif{n}{\r}}\left[\frac{\psharp{\mu}(\blambda)}{\r^k}\right],
    \]
    and each expectation here is at most $(\frac{1}{\omega^2})^{|\mu|} \leq \frac{1}{\omega^2}$ by Lemma~\ref{lem:p-sharp-expect-bound}. This completes the proof.
\end{proof}

\begin{lemma}                                       \label{lem:errL2}
    In fact, the constants~$A_\c$ from Lemma~\ref{lem:errL} satisfy $A_\c \leq 2^{O(\c^2 \log \c)}$.
\end{lemma}
\begin{proof}
    The proof involves some tedious analysis using the results of Section~\ref{sec:work-it}.  It suffices to show that
    \begin{equation}                                \label{eqn:tedious}
        \sum_{\mu : \weight(\mu) \leq k} |c_{k,\mu}| \leq 2^{O(k^2 \log k)},
    \end{equation}
    where, recall, the coefficients $c_{k,\mu}$ are defined by
    \begin{equation}                                    \label{eqn:pull-the-trigger}
        \pstar{k} = \sum_{\mu~:~\weight(\mu) = k+1} \frac{\falling{k}{(\ell(\mu)-1)}}{\partfact{\mu}} \psharp{\mu} + \sum_{\mu~:~\weight(\mu) \leq k} c_{k,\mu} \psharp{\mu}.
    \end{equation}
    Let us return to the relationship between the $\pstar{}$ and $\psharp{}$ polynomials described in Section~\ref{sec:work-it}. Specifically, we'll need identities~\eqref{eqn:crazy-gen},~\eqref{eqn:crazy-gen2}, which express each $\psharp{k}$ as a polynomial in $\pstar{1}, \dots, \pstar{k}$ via the power series $Q_k(t)$.

    Given any polynomial~$R$ in indeterminates $p_1, \dots, p_k$ (either $\pstar{}$'s or $\psharp{}$'s), write $\|R\|$ for the sum of the absolute values of~$R$'s coefficients.  This is a submultiplicative norm.  Observe from~\eqref{eqn:crazy-gen2} that $\|Q_{k,m}\| \leq (k+1)^{m+1}$ (indeed, one may show it's precisely $\frac{(k+1)^{m+1}-k^{m+1}-1}{m+1}$).  Thus the coefficient on~$t^s$ in $Q_k(t)^i$ is a polynomial in $\pstar{1}, \dots, \pstar{k}$ of norm at most~$O(k)^s$.  Hence the same is true for the coefficient on~$t^s$ in the expression $\sum_{i=0}^\infty \frac{(-1)^i}{i!}Q_k(t)^i$ from~\eqref{eqn:crazy-gen2}.  As the coefficient on each power of~$t$ in $\prod_{j=1}^k (1-(j-\frac12)t)$ is a number of magnitude at most $(k-\frac12)^k$, we finally deduce that the relationship~\eqref{eqn:invertible} can be expressed more quantitatively as
    \[
        \psharp{k} = \pstar{k} + R_k(\pstar{1}, \dots, \pstar{k-1}), \quad \text{where } 1+\|R_k\| \leq \exp(b k \log k), \quad b \text{ a universal constant}.
    \]
    We inductively invert this relationship as in~\eqref{eqn:inverted}, writing
    \begin{equation}                    \label{eqn:banana}
        \pstar{k} = S_k(\psharp{1}, \dots, \psharp{k}), \quad \text{where } S_k = \psharp{k} + \Bigl\{\text{polynomial in } \psharp{1}, \dots, \psharp{k-1}\text{ of gradation at most $k-1$}\Bigr\}.
    \end{equation}
    If we let $s(k) = \|S_k\|$, using convexity of $\exp(b k \log k)$ we get the inductive bound
    \[
        s(k) \leq \exp(b k \log k) s(k-1),
    \]
    leading to the bound $s(k) \leq \exp(O(k^2 \log k))$.  This is nearly enough to complete the proof; the only issue is that in~\eqref{eqn:banana} we have a polynomial in the $\psharp{j}$'s, whereas in~\eqref{eqn:pull-the-trigger} we have the products of $\psharp{j}$'s expanded out into linear combinations of $\psharp{\mu}$'s.  However Lemma~\ref{lem:structure-bound} below, which crudely bounds the magnitude of the structure constants for the $\psharp{}$'s, shows that each monomial $\prod_{i} \psharp{\lambda_i}$ with gradation $|\lambda| = w$ can be replaced by a linear polynomial in  $\psharp{\mu}$'s (with $|\mu| \leq w$) wherein each coefficient has magnitude at most~$4^{w^2 \log w}$.  Since $w$ is always bounded by~$k-1$ and since there are at most $2^{O(\sqrt{k})} \ll \exp(O(k^2 \log k))$ partitions $\mu$ with $|\mu| \leq k$, we conclude that each of these linear polynomials has norm at most $\exp(O(k^2 \log k))$.  Thus making these replacements in~$S_k$ only increases its norm by another multiplicative factor of $\exp(O(k^2 \log k))$.  The proof is complete.
\end{proof}

\begin{lemma}       \label{lem:structure-bound}
    Let $\lambda \vdash w$, and suppose
    $\displaystyle
        \prod_{i=1}^{\ell(\lambda)} \psharp{\lambda_i} = \sum_{\mu} c_\mu \psharp{\mu}
    $
    within $\Lambda^*$.  Then $|c_\mu| \leq 4^{w^2 \log w}$ for all $\mu$.
\end{lemma}
\begin{proof}
    The proof is an induction on $\ell = \ell(\lambda)$, the base case of $\ell = 1$ being trivial.  
    Now for general $\lambda$ with $\lambda_\ell = k$ we have
    \begin{equation}                        \label{eqn:apple-pie}
    \prod_{i=1}^{\ell} \psharp{\lambda_i}
    = \left(\prod_{i = 1}^{\ell-1} \psharp{\lambda_i}\right)\cdot \psharp{k}
    = \left(\sum_{\mu} d_\mu \psharp{\mu}\right)\cdot \psharp{k}
    = \sum_\mu d_\mu \sum_\tau f^\tau_{\mu k}\psharp{\tau}
    = \sum_\tau \psharp{\tau} \sum_\mu d_\mu f^\tau_{\mu k},
    \end{equation}
    where each $|d_\mu|$ is at most $4^{(w-k)^2 \log(w - k)} \leq 4^{(w-1)^2 \log(w)}$ by induction.
    By Corollary~\ref{cor:structure-coefficients-rewrite}, the structure constants $f^{\tau}_{\mu k}$ satisfy $|f^\tau_{\mu k}| \leq |C^{|\tau|}_{|\mu| k}| \leq |\mu|! k!\leq w^w$.  Since the number of partitions of $(w-k)$ is trivially at most $w^w$,
    the coefficient on $\psharp{\tau}$ in~\eqref{eqn:apple-pie} has magnitude at most
    \begin{equation*}
    \sum_\mu |d_\mu f^\tau_{\mu k}|
    \leq w^{2w} \cdot \max_{\mu} |d_\mu|
    \leq w^{2w} \cdot 4^{(w-1)^2 \log (w)} \leq 4^{w^2 \log w},
    \end{equation*}
    completing the induction.
\end{proof}

\subsubsection{Combining the bounds}
Combining \eqref{eqn:EL-bound}, and Lemmas~\ref{lem:errL},~\ref{lem:errL2}, we get that under the hypotheses of Proposition~\ref{prop:rank-warmup},
\begin{equation}                                        \label{eqn:final-kl-bound}
    \dkl{\SWunif{n}{\r}}{\SWunif{n}{\rp}} \leq \r\left(\frac{3}{\omega}\right)^\c + \exp(O(\c^2 \log \c))\cdot \frac{1}{\omega^2} + \tfrac{203}{\omega^3} \leq \exp(O(\c^{2.01})) \cdot \frac{1}{\omega^2}.
\end{equation}
In the above we used $\r\left(\frac{3}{\omega}\right)^\c \leq \r\left(\frac{10}{\omega}\right)^\c \leq \r\left(\frac{1}{10\r}\right)^3 \leq \frac{1}{\omega^3}$, the second inequality here following from the assumed lower bound on~$\c$. It's now evident that we should take~$\c$ as small as we can; in particular, to equal $\lceil 3\frac{\log(10\r)}{\log(\omega/10)} \rceil$.  We conclude:
\begin{theorem}                                     \label{thm:r-vs-r+1}
    For any $200 \leq \omega \leq \sqrt{\r}$, if $n = \frac{\r^2}{\omega^2}$ then
    \[
        \dkl{\SWunif{n}{\r}}{\SWunif{n}{\r+1}}  \leq \exp(O((\log \r) / (\log \omega))^{2.01}) \cdot \omega^{-2}.
    \]
    In particular, for $\omega = \exp(O(\log^{.67} \r))$ and hence $n = \r^{2 - O(1/\log^{.33} \r)}$, the above bound is $o_{\r}(1)$.
\end{theorem}
By Pinsker's inequality we may conclude also that $\dtv{\SWunif{n}{\r}}{\SWunif{n}{\rp}} \leq o_r(1)$ unless $n = \r^{2 - O(1/\log^{.33} \r)} = \wt{\Omega}(r^2)$.  This completes the proof of the rank-$r$ versus rank-$(r+1)$ testing lower bound; in particular, the more precise bound~\eqref{eqn:precise-rank-lower-bound} in the case $\Delta = 1$.

\subsection{Extension to $\Delta > 1$}         \label{sec:tv-improvement}
Let us henceforth fix the parameter $\c = \lceil 3\frac{\log(10\r)}{\log(\omega/10)} \rceil$.  To recap the preceding section we saw that
\begin{equation}        \label{eqn:remind-me}
    \abs{\E[\L_\c(\blambda)]} \leq \exp(O(\c^{2.01})) \cdot \frac{1}{\omega^2}, \quad
    \text{and hence}\quad \dkl{\SWunif{n}{\r}}{\SWunif{n}{\r+1}} \leq \exp(O(\c^{2.01})) \cdot \frac{1}{\omega^2}.
\end{equation}
If we apply Pinsker's inequality to the latter bound we obtain
\[
    \dtv{\SWunif{n}{\r}}{\SWunif{n}{\r + 1}} \leq \exp(O(\c^{2.01})) \cdot \frac{1}{\omega}.
\]
The key to getting a good lower bound when $\Delta > 1$ is to show that Pinsker's inequality is not sharp in our setting, and in fact the following is true:
\begin{theorem}                                     \label{thm:pinsker-not-sharp}
    For any $200 \leq \omega \leq \sqrt{\r}$, if $n = \frac{\r^2}{\omega^2}$ then
    \[
        \dtv{\SWunif{n}{\r}}{\SWunif{n}{\r+1}}  \leq \exp(O(\c^{2.01})) \cdot \frac{1}{\omega^2}.
    \]
\end{theorem}
From this we can obtain the testing bound~\eqref{eqn:precise-rank-lower-bound} for rank-$r$ versus rank-$(r + \Delta)$ (where $1 \leq \Delta \leq r$) simply by using the triangle inequality. Specifically, given $r \leq r' \leq 2r$ and~$n$, define $\omega_{r'}$ by $n = \frac{(r')^2}{\omega_{r'}^2}$.  Applying Theorem~\ref{thm:pinsker-not-sharp} for each $r'$, we get
\[
    \dtv{\SWunif{n}{\r'}}{\SWunif{n}{\r'+1}}  \leq \exp(O((\log \r') / (\log \omega_{r'}))^{2.01}) \cdot \frac{1}{\omega_{r'}^2} \quad \text{for all $r \leq r' < 2r$.}
\]
But $\omega_{r'}$ is within a factor of~$2$ of $\omega_r$ for all $r \leq r' \leq 2r$; thus by adjusting the constant in the $O(\cdot)$, the above holds with $\omega_r$ in place of $\omega_{r'}$.  Applying the triangle inequality, we get
\[
    \dtv{\SWunif{n}{\r}}{\SWunif{n}{\r'+\Delta}}  \leq \exp(O((\log \r) / (\log \omega_r))^{2.01}) \cdot \frac{1}{\omega_r^2} \cdot \Delta.
\]
Again, taking $\omega_r = \exp(O(\log^{.67} \r))$, we get
\[
    \dtv{\SWunif{n}{\r}}{\SWunif{n}{\r'+\Delta}}  \leq \frac{n}{r^{2-O(1/\log^{.33} r)}} \cdot \Delta,
\]
and this completes the proof of the rank-testing lower bound~\eqref{eqn:precise-rank-lower-bound}.\\

Thus it remains to prove Theorem~\ref{thm:pinsker-not-sharp}.  The main result we need for this is the following:
\begin{theorem}                     \label{thm:Delta-bound-variance}
    $\displaystyle \Var_{\blambda \sim \SWunif{n}{\r}}[\L_\c(\blambda)] \leq \exp(O(\c^{2.01})) \cdot \frac{1}{\omega^4}.$
\end{theorem}

To prove Theorem~\ref{thm:Delta-bound-variance} we will employ the following lemma:
\begin{lemma}                                       \label{lem:psharp-var}
    Let $\mu$ be a partition with $\weight(\mu) = k \geq 2$.  Then
     \[
        \Var_{\blambda \sim \SWunif{n}{\r}}[\psharp{\mu}(\blambda)] \leq \exp(O(k^2 \log k)) \cdot \r^{2k-2} \cdot (1/\omega^4).
     \]
\end{lemma}
\begin{proof}
    If $|\mu| = 1$ then $\psharp{\mu}(\lambda) = n$ which has variance~$0$.  Thus we may assume $|\mu| \geq 2$ and hence $k \geq 3$.  Using Proposition~\ref{prop:should-be-in-a-structure-constants-discussion-or-theorem},
    \begin{equation}                                        \label{eqn:var-psharp}
        \Var[\psharp{\mu}(\blambda)] = \E[\psharp{\mu}(\blambda)^2] - \E[\psharp{\mu}(\blambda)]^2 = \E[\psharp{\mu \cup \mu}(\blambda)] - \E[\psharp{\mu}(\blambda)]^2 + \E[q_\mu(\blambda)]
    \end{equation}
    where $q_\mu(\lambda)$ is a certain linear combination of $\psharp{\nu}$ polynomials, each of weight at most $2k-2$.
    Regarding the first two quantities here, Corollary~\ref{cor:sw-expect} tells us that
    \[
        \E[\psharp{\mu \cup \mu}(\blambda)] - \E[\psharp{\mu}(\blambda)]^2 =\falling{n}{(2|\mu|)} \r^{2\ell(\mu)-2|\mu|} - (\falling{n}{|\mu|} \r^{\ell(\mu)-|\mu|})^2 = \r^{2\ell(\mu)-2|\mu|}(\falling{n}{(2|\mu|)} - (\falling{n}{|\mu|})^2),
    \]
    which is evidently nonpositive.  Thus it suffices to prove the upper bound
    \begin{equation}            \label{eqn:verif-me}
        \left|\E[q_\mu(\blambda)]\right| \leq \exp(O(k^2 \log k)) \cdot \r^{2k-2} \cdot (1/\omega^4).
    \end{equation}
    By Lemma~\ref{lem:structure-bound}, the coefficients on the $\psharp{\nu}$'s in the linear combination $q_\mu(\lambda)$ each have magnitude at most $\exp(O(k^2 \log k))$, and there are at most $2^{O(\sqrt{k})}$ of them.  Thus~\eqref{eqn:verif-me} follows provided we can show $\E[\psharp{\nu}(\blambda)] \leq \r^{2k-2}/\omega^4$ for all $\nu$ of weight at most $2k-2$.  This is immediate from Lemma~\ref{lem:p-sharp-expect-bound} for all $\nu \neq (1)$, and when $\nu = (1)$ it still holds: Lemma~\ref{lem:p-sharp-expect-bound} gives us the bound $\r^2/\omega^2 \leq \r^3/\omega^4 \leq \r^{2k-2}/\omega^4$, the first inequality using $\omega \leq \sqrt{\r}$ and the second using $k \geq 3$.
\end{proof}

We can now prove Theorem~\ref{thm:Delta-bound-variance}.
\begin{proof}[Proof of Theorem~\ref{thm:Delta-bound-variance}.]
    Recall identity~\eqref{eqn:LC-approx}:
    \[
        \L_\c(\lambda) =  \sum_{k=2}^\c \frac{(-1)^k}{k\r^k} \cdot \sum_{\weight(\mu) = k+1} \frac{\falling{k}{(\ell(\mu)-1)}}{\partfact{\mu}} \psharp{\mu}(\lambda) + \sum_{k=2}^\c \frac{(-1)^k \calO_{k}(\lambda)}{k\r^k}.
    \]
    We claim that for each $2 \leq k \leq \c$,
    \begin{equation}    \label{eqn:oh-me}
        \Var\left[\frac{(-1)^k \calO_{k}(\blambda)}{k\r^k}\right] \leq \exp(O(\c^{2.01})) \cdot \frac{1}{\omega^4},
    \end{equation}
    and that furthermore for each $\mu$ of weight $k+1$ we have
    \begin{equation}    \label{eqn:oh-my}
        \Var\left[ \frac{(-1)^k}{k\r^k}\cdot\frac{\falling{k}{(\ell(\mu)-1)}}{\partfact{\mu}} \psharp{\mu}(\blambda)\right] \leq \exp(O(\c^{2.01})) \cdot \frac{1}{\omega^4}.
    \end{equation}
    This is sufficient to complete the proof, as in general
    \begin{equation}        \label{eqn:var-is-cool}
        \Var[\bX_1 + \cdots + \bX_m] \leq m (\Var[\bX_1] + \cdots + \Var[\bX_m]);
    \end{equation}
    in our particular case we have only $m = \exp(O(\sqrt{\c}))$ summands, and this factor can be absorbed into the target variance bound of $\exp(O(\c^{2.01})) \cdot (1/\omega^4)$.  To verify~\eqref{eqn:oh-me}, first recall that each $\calO_k(\blambda)$ is a linear combination of $\psharp{\nu}(\lambda)$'s for $\nu$ of weight at most~$k \leq \c$; further, the sum of the absolute value of the coefficients is at most $\exp(O(\c^{2.01}))$ (see~\eqref{eqn:tedious}).  Using~\eqref{eqn:var-is-cool} again, it therefore suffices to check that
    \[
        \Var\left[\frac{\psharp{\nu}(\blambda)}{\r^k}\right] \leq \exp(O(\c^{2.01})) \cdot \frac{1}{\omega^4}
    \]
    when $\weight(\nu) \leq k \leq \c$.  By Lemma~\ref{lem:psharp-var} this is true, with a factor of $\r^{-2}$ to spare.

    To verify~\eqref{eqn:oh-my}, we may ignore the factor $\frac{(-1)^k}{k \cdot \partfact{\mu}}$, and also ignore the factor $\falling{k}{(\ell(\mu)-1)}$ as it contributes at most a multiplicative $\c^\c \ll \exp(O(\c^{2.01}))$.  Thus it suffices to show $\Var[\psharp{\mu}(\blambda)/\r^k] \leq \exp(O(\c^{2.01}))/\omega^4$ for $\mu$ of weight $k+1$ (and $k \leq \c$).  But this is immediate from Lemma~\ref{lem:psharp-var}.
\end{proof}

We now work towards the proof of Theorem~\ref{thm:pinsker-not-sharp}.  Adding Theorem~\ref{thm:Delta-bound-variance} and the square of~\eqref{eqn:remind-me} we obtain
\begin{equation} \label{eqn:monkey}
    \E_{\blambda \sim \SWunif{n}{\r}}[\L_\c(\blambda)^2] \leq \exp(O(\c^{2.01})) \cdot \frac{1}{\omega^4}.
\end{equation}
We would now like to similarly claim that
\newcommand{\Lp}{L^{\!\scriptscriptstyle{+}\!}}
\begin{equation} \label{eqn:monkey2}
    \E_{\blambda \sim \SWunif{n}{\rp}}[\Lp_\c(\blambda)^2] \leq \exp(O(\c^{2.01})) \cdot \frac{1}{\omega^4},
\end{equation}
where we are writing
\[
    \Lp_\c(\lambda) \coloneqq \sum_{k=2}^\c\frac{(-1)^k \pstar{k}(\lambda)}{k(\r+1)^k}.
\]
To obtain this, it suffices to repeat all of the arguments beginning with Section~\ref{sec:L-psharp} until this point; the only thing that really changes is that $\omega = \omega_r$ needs to be replaced with $\omega_{r+1}$, but this has a negligible effect on the bounds (and indeed usually very slightly improves them).

Next, we claim that Lemma~\ref{lem:down-with-infinity} continues to hold if we replace $\L_\c(\lambda)$ with the analogous $\Lp_\c(\lambda)$.  The key change to the proof comes in the last main inequality, where we need to observe that the 
\[
    \left(\frac{1}{\r^k} - \frac{1}{(\r+\frac12)^k}\right) \leq \frac{k}{2\r^{k+1}}
\]
continues to hold if the left-hand side is replaced with
\[
    \left(\frac{1}{(\r+\frac12)^k} - \frac{1}{(\r+1)^k}\right).
\]

We need one more definition for the proof of Theorem~\ref{thm:pinsker-not-sharp}.
\begin{definition}
    Say that $\lambda \vdash n$ is \emph{usual$^+$} if it is usual and if furthermore $|\L^*_\infty(\lambda)| \leq 2$.
\end{definition}
\begin{lemma}                                       \label{lem:usual+}
    Both for $\blambda \sim \SWunif{n}{\r}$ and $\blambda \sim \SWunif{n}{\rp}$ it holds that
    \[
        \Pr[\blambda \textnormal{ not usual$^+$}] \leq \exp(O(\c^{2.01})) \cdot \frac{1}{\omega^4}.
    \]
\end{lemma}
\begin{proof}
    For $\blambda \sim \SWunif{n}{\r}$, Lemma~\ref{lem:lambda-usual} tells us that 
    \[
        \Pr[\blambda \text{ not usual}] \leq 2^{-20\r/\omega} \leq 2^{-\Omega(\sqrt{\r})} \ll\exp(O(\c^{2.01})) \cdot \frac{1}{\omega^4}
    \]
    and it's easy to check that this is also true with plenty of room to spare for $\blambda \sim \SWunif{n}{\rp}$.  Thus it suffices to verify for both distributions on~$\blambda$ that the probability of $|\L^*_\infty(\lambda)| \leq 2$ satisfies the same upper bound. By applying Markov's inequality to~\eqref{eqn:monkey},~\eqref{eqn:monkey2} we get
    \[
        \Pr_{\blambda \sim \SWunif{n}{\r}}[\L_\c(\blambda)^2 \geq 1], \Pr_{\blambda \sim \SWunif{n}{\rp}}[\Lp_\c(\blambda)^2 \geq 1] \leq \exp(O(\c^{2.01})) \cdot \frac{1}{\omega^4}.
    \]
    Finally, when $\blambda$ is usual and $|\L_\c(\blambda)^2| \not \geq 1$, it follows that necessarily $|\L^*_\infty(\blambda)| \leq 2$, in light of Lemma~\ref{lem:down-with-infinity} and the fact that~$\frac{201}{\omega^3} \leq 1$.  As noted earlier, the $\rp$-analogue of Lemma~\ref{lem:down-with-infinity} holds, and hence we may draw the same conclusion concerning $\Lp_\c(\blambda)^2$.
\end{proof}

Finally we are ready to complete the proof of Theorem~\ref{thm:pinsker-not-sharp}. We begin with
\begin{multline*}
    \dtv{\SWunif{n}{\r}}{\SWunif{n}{\rp}}
    \leq \frac12 \Pr_{\blambda \sim \SWunif{n}{\r}}[\blambda \text{ not usual$^+$}] + \frac12\Pr_{\blambda \sim \SWunif{n}{\rp}}[\blambda \text{ not usual$^+$}] \\
        + \frac12\sum_{\text{usual$^+$~$\lambda$}} \ABS{\SWunif{n}{\rp}[\lambda] - \SWunif{n}{\r}[\lambda]}.
\end{multline*}
We can bound the first two terms above using Lemma~\ref{lem:usual+}. Indeed there is room to spare, as the bound we get is the square of what we can tolerate.  Thus it remains to bound the third term by $\exp(O(\c^{2.01})) \cdot \frac{1}{\omega^2}$.  For it we use
\begin{align}                                    
    \sum_{\text{usual$^+$~$\lambda$}} \ABS{\SWunif{n}{\rp}[\lambda] - \SWunif{n}{\r}[\lambda]}
    &= \E_{\blambda \sim \SWunif{n}{\rp}}\left[1_{\{\blambda\text{ usual$^+$}\}} \cdot \ABS{1-\frac{\SWunif{n}{\r}[\blambda]}{\SWunif{n}{\rp}[\blambda]}}\right] \nonumber\\
    &= \E_{\blambda \sim \SWunif{n}{\rp}}\left[1_{\{\blambda\text{ usual$^+$}\}} \cdot \ABS{1-\exp(u(\blambda))}\right] \label{eqn:im-tired}
\end{align}
where
\begin{equation}                            \label{eqn:hey}
    u(\blambda) = \ln\left(\frac{\SWunif{n}{\r}[\lambda]}{\SWunif{n}{\rp}[\lambda]}\right)
                = n \ln\left(1+\frac{1}{\r}\right) -\frac{n}{\r+\frac12} + \L^*_\infty(\blambda),
\end{equation}
the last equality holding from~\eqref{eqn:dkl2} (see also the sentence after~\eqref{eqn:dkl-prest}) under the assumption that $\blambda$ is usual (which we can indeed assume, since we're multiplying against $1_{\{\blambda\text{ usual$^+$}\}}$).  As we noted after~\eqref{eqn:dkl-prest}, the first two quantities in~\eqref{eqn:hey} sum to a positive quantity not exceeding $\frac{n}{12\r^3} \leq \frac{1}{\omega^2}$.  Furthermore, because of the presence of the usual$^+$-indicator in~\eqref{eqn:im-tired} we may assume in analyzing~\eqref{eqn:hey} that $|\L^*_\infty(\blambda)| \leq 2$.  Thus we may use the bound $u(\blambda) \leq 2 + \frac{1}{\omega^2} \leq 2.01$.  Since $|1-\exp(u)| \leq 4|u|$ for $u \in [-2.01,2.01]$, we may conclude that
\[
    \eqref{eqn:im-tired} \leq 4\E_{\blambda \sim \SWunif{n}{\rp}}\left[1_{\{\blambda\text{ usual$^+$}\}} \cdot \left(\frac{1}{\omega^2} + \abs{\L^*_\infty(\blambda)}\right)\right].
\]
Thus to complete the proof of Theorem~\ref{thm:pinsker-not-sharp} it remains to show
\[
    \E_{\blambda \sim \SWunif{n}{\rp}}\left[\abs{\L^*_\infty(\blambda)}\right]\leq  \exp(O(\c^{2.01})) \cdot \frac{1}{\omega^2}.
\]
By the $\rp$-analogue of Lemma~\ref{lem:down-with-infinity}, it suffices to prove this with $\Lp_\c(\blambda)$ in place of $\L^*_\infty(\blambda)$, because $201/\omega^3 \ll \exp(O(\c^{2.01}))/\omega^2$.  But finally
\[
    \E_{\blambda \sim \SWunif{n}{\rp}}\left[\abs{\Lp_\c(\blambda)}\right] \leq \sqrt{\E_{\blambda \sim \SWunif{n}{\rp}}\left[\Lp_\c(\blambda)^2\right]} \leq \exp(O(\c^{2.01})) \cdot \frac{1}{\omega^2},
\]
using Cauchy--Schwarz and~\eqref{eqn:monkey2}.  The proof of Theorem~\ref{thm:pinsker-not-sharp}---and hence also the testing lower bound~\eqref{eqn:precise-rank-lower-bound}---is therefore complete.

\section{Quantum rank testing}
\subsection{Testers with one-sided error}\label{sec:one-sided}

In this section, we prove the first part of Theorem~\ref{thm:rank-intro}, that $\Theta(r^2/\eps)$ copies are necessary and sufficient to test whether or not a state has rank~$r$ with one-sided error.
We will show this by analyzing the following algorithm.
\begin{nameddef}{Rank Tester}
Given $\rho^{\otimes n}$,
\begin{enumerate}
\item Sample $\blambda \sim \SWdens{n}{\rho}$.
\item Accept if $\ell(\blambda) \leq r$.  Reject otherwise.
\end{enumerate}
\end{nameddef}
\noindent
Our primary tool in analyzing this tester will be the RSK correspondence.
Suppose $\rho$'s nonzero eigenvalues are $\eta = \{\eta_1, \ldots, \eta_d\}$, and let~$\calD$ be the distribution over $[d]$ induced by~$\eta$.
By Remark~\ref{rem:quantum-free}, $\SWdens{n}{\rho}$ has the same distribution as the process which first samples $\bw \sim \calD^{\otimes n}$ and outputs $\blambda = \rsk(\bw)$.  Write $\lds(\bw)$ for the length of the longest strongly decreasing subsequence in~$\bw$.
By Theorem~\ref{thm:greene}, $\ell(\blambda) = \lds(\bw)$.

The key property we will need of the Rank Tester is the following:

\begin{proposition}\label{prop:rank-tester-is-optimal}
The Rank Tester is the optimal algorithm for testing whether or not a state has rank~$r$ with one-sided error.
\end{proposition}
\begin{proof}
To show this, we need to show (i) that every $\lambda$ satisfying $\ell(\lambda) \leq r$ occurs with nonzero probability in $\SWdens{n}{\rho}$ for some $\rho$ of rank $r$ and (ii) that no $\lambda$ satisfying $\ell(\lambda) > r$ occurs in $\SWdens{n}{\rho}$ for any $\rho$ of rank $r$.  The first follows because if $\rho$ has $r$ nonzero eigenvalues, then the word
\begin{equation*}
w\coloneqq
\underbrace{r, \dots, r}_{\lambda_r\ \text{letters}},
\underbrace{(r-1), \dots, (r-1)}_{\lambda_{r-1}\ \text{letters}},
\ldots,
\underbrace{1, \dots, 1}_{\lambda_1\ \text{letters}}
\end{equation*}
occurs in $\calD^{\otimes n}$ with nonzero probability.  It is easy to check that $\lambda = \rsk(w)$.

To show that (ii) holds, if $\rho$ is rank~$r$, then $\eta$ has at most $r$ nonzero entries.
Thus, any word~$w$ in the support of $\calD^{\otimes n}$ will always satisfy $\lds(w) \leq r$ because $w$ will contain at most $r$ distinct letters.  As $\ell(\lambda) = \lds(w)$, we are done.
\end{proof}

As a result of Proposition~\ref{prop:rank-tester-is-optimal}, Theorem~\ref{thm:rank-intro} follows from the following lemma.

\begin{lemma}
The Rank Tester tests whether or not a state has rank $r$ with $\Theta(r^2/\eps)$ copies.
\end{lemma}
\begin{proof}
If $\rho$ is $\eps$-far from having rank~$r$, then $\eta$ is $\eps$-far in TV distance from having support size~$r$.
Thus, we can show the lemma by showing the following two facts about probability distributions.
\begin{enumerate}[label=(\roman*)]
\item For every probability distribution $\calD= (p_1, \ldots, p_d)$ which is $\eps$-far from having support size~$r$, a random word $\bw \sim \calD^{\otimes n}$ satisfies $\lds(\bw) \geq r+1$ with probability at least~$2/3$ for some $n = O(r^2/\eps)$.\label{item:rank-upper}
\item There exists an integer~$d$ and a probability distribution $\calD = (p_1, \ldots, p_d)$ which is $\eps$-far from having support size~$r$ such that, for a random word $\bw \sim \calD^{\otimes n}$, $\lds(\bw) \leq r$ with probability greater than~$1/3$ whenever $n = o(r^2/\eps)$.\label{item:rank-lower}
\end{enumerate}
\paragraph{Proof of statement~\ref{item:rank-upper}:}
We will need the following concentration bound for sums of geometric random variables.
\begin{proposition}[\cite{Bro}]\label{prop:geom-conc}
Write $X = X_1 + \ldots + X_n$, where the $X_i$'s are iid geometric random variables with expectation~$\mu$.
For any $k >1$, $\Pr[X > k n \mu] \leq \exp\left(-\frac{1}{2} k n (1-1/k)^2\right)$.
\end{proposition}
\noindent
We note that Proposition~\ref{prop:geom-conc} also holds with the weaker hypothesis that the $X_i$'s are independent (and not necessarily identically distributed), each with expectation at most~$\mu$.

We may assume that $p_1 \geq \ldots \geq p_d$.
We will split into two cases, handled below: \ref{item:case-big} $p_{r+1} \geq \eps/4r$ and \ref{item:case-small} $p_{r+1} < \eps/4r$.
\begin{enumerate}[label=(\arabic*)]
\item Because the probabilities are sorted, $p_1, \ldots, p_{r+1} \geq \eps/4r$.
For the infinite random word $\bw \sim \calD^{\otimes \infty}$,
consider the number of letters one has to traverse through before finding $(r+1),r, \ldots, 1$ as a subsequence.
This number is distributed as $\bX = \bX_{r+1} + \ldots + \bX_1$, where $\bX_i$ is a geometric random variable with success probability $p_i$.

By assumption, $p_i \geq \eps/4r$, and therefore $\E \bX_i \leq 4r/\eps$, for each $i \in [r+1]$.
By Proposition~\ref{prop:geom-conc}, $\bX$ is at most $24r^2/\eps$ with probability at least~$2/3$.
Thus, if $n = 24r^2/\eps$, then $\bw \sim \calD^{\otimes n}$ has a strictly decreasing subsequence of size $r+1$ with high probability.\label{item:case-big}

\item
Because the probabilities are sorted, $p_{r+1},\ldots, p_d < \eps/4r$.  Place the letters from $\{r+1, \ldots, d\}$ into buckets as follows:
starting from letter $(r+1)$ and proceeding in order, add each letter to the current bucket until it contains at least $\eps/4r$ weight.
At this point, move to the next bucket and repeat this process starting with the current letter until all letters have been bucketed.

Because these letters have weight $\leq \eps/4r$, each bucket has total weight in the interval $[\eps/4r, \eps/2r)$ (except possibly the final bucket).
There must be at least $2r+1$ buckets with nonzero total weight, as otherwise $p_{r+1} + \ldots + p_d < \eps$, contradicting the fact that $\calP$ is $\eps$-far from having support size~$r$.  This gives us at least $2r\geq r + 1$ buckets each of which contains at least $\eps/4r$ total weight.

Now  we can use an argument similar to case~\ref{item:case-big} to show that when $n = 24r^2/\eps$, a random $\bw \sim \calD^{\otimes n}$ will with probability $\geq 2/3$ have a strictly decreasing subsequence in which the first letter comes from bucket $r+1$, the second letter comes from bucket~$r$, and so on (ending in a letter from the first bucket).  This is a strictly decreasing subsequence of size $r+1$.\label{item:case-small}
\end{enumerate}

\paragraph{Proof of statement~\ref{item:rank-lower}:}
For $d \gg r$, define the probability distribution
\begin{equation*}
\calP = \left(1-2\eps, \frac{2\eps}{d - 1}, \ldots , \frac{2\eps}{d-1}\right).
\end{equation*}
Because $d \gg r$, $\calP$ is $\eps$-far from having support size~$r$.
For a string $w \in [d]^n$, let $\widetilde{w}$ be the substring of $w$ formed by deleting all occurrences of the letter ``1'' from~$w$.
It is easy to see that $\lds(\widetilde{w})\leq\lds(w)\leq\lds(\widetilde{w})+1$.

For a randomly drawn $\bw \sim \calP^{\otimes n}$, let us condition on $\widetilde{\bw}$  having a certain fixed length $m$.
The value of $\lds(\widetilde{\bw})$ is distributed as the length of the longest decreasing subsequence in a uniformly random word drawn from $[d-1]^m$.  By Theorem~\ref{thm:greene}, this is distributed as $\blambda_1'$ for $\blambda \sim \SWunif{m}{d-1}$.
Setting $B = \left\lceil100\sqrt{m}\right\rceil$,
let us show that $\Pr[\blambda_1' \geq B]$ is small.
If $B \geq d$, then surely $\blambda_1' < B$ always, as $\blambda \sim \SWunif{m}{d-1}$ will always have height at most~$d-1$.
On the other hand, if $B < d$, then by Proposition~\ref{prop:increasing-seq},
\begin{equation*}
\Pr[\blambda_1' \geq B] \leq \left(\frac{2e^2 m}{B^2}\right)^B \leq \frac{2e^2}{10000}.
\end{equation*}
In summary, conditioned on $\widetilde{\bw}$ having a certain fixed length~$m$, $\lds(\widetilde{\bw}) \leq O(\sqrt{m})$
with all but the above probability.

In expectation, for a random $\bw \sim \calP^{\otimes n}$, $\widetilde{\bw}$ has length $2\eps d$.
By Markov's inequality, the probability that the length of $\widetilde{\bw}$ is greater than $200 \eps d$ is at most $1/100$.
Conditioned on the length of $\widetilde{\bw}$ being at most $200 \eps d$,
the above paragraph tells us that $\lds(\widetilde{\bw})\leq O(\sqrt{\eps d})$  with probability $1-2e^2/10000$.
Thus, when $\bw \sim \calP^{\otimes n}$, we have with probability greater than~$1/3$ that $\lds(\bw)\leq O(\sqrt{\eps d})$, which is $o(r)$ unless $d = \Omega(r^2/\eps)$.
\end{proof}

For our last result of this section, we will show that the copy complexity of the Rank Tester can be improved in certain interesting cases.
In particular, the Rank Tester matches the upper bound of the Uniform Distribution Distinguisher from Section~\ref{sec:uniform} for the case of $r$ v. $r+1$, and does so with one-sided error.
\begin{proposition}
The Rank Tester can distinguish between the case when $\rho$'s spectrum is uniform on either $r$ or $r+1$ eigenvalues with $O(r^2)$ copies of $\rho$.
\end{proposition}
\begin{proof}
If $\rho$'s spectrum is uniform on~$r$ eigenvalues, then it is rank~$r$ and so the Rank Tester never rejects.
Thus, we need only show that the Rank Tester rejects with probability $\geq 2/3$ when $\rho$'s spectrum is uniform on $r+1$ eigenvalues for some $n = O(r^2)$.
We will follow the analysis in the proof of statement~\ref{item:rank-upper} above and show that a random word $\bw \sim \calD^{\otimes n}$ has $\lds(\bw) = r+1$ with high probability.
The gain will come from the fact that $\eta = (1/(r+1), \ldots, 1/(r+1))$.

For the infinite random word $\bw \sim \calD^{\otimes \infty}$,
consider the number of letters one has to traverse through before finding $(r+1),r, \ldots, 1$ as a subsequence.
This number is distributed as $\bX = \bX_{r+1} + \ldots + \bX_1$, where $\bX_i$ is a geometric random variable with success probability $1/(r+1)$ and expectation~$r+1$.
By Proposition~\ref{prop:geom-conc}, $\bX$ is at most $6r^2$ with probability at least~$2/3$.
Thus, if $n = 6r^2$, then $\bw \sim \calD^{\otimes n}$ has a strictly decreasing subsequence of size $r+1$ with high probability.
\end{proof}

\subsection{A lower bound for testers with two-sided error}

In this section, we prove the second part of Theorem~\ref{thm:rank-intro}, that $\Omega(r/\eps)$ copies are necessary to test whether or not a state has rank $r$ with two-sided error.
\begin{proof}
Let $d \gg r$.
In this proof, we will take the viewpoint of a density matrix as a probability distribution over pure states.
Let $\rho$ and $\sigma$ be maximally mixed on subspaces of dimension $(r-1)$ and $(d-1)$, respectively.
Consider the following process for generating a product state $|\Psi\rangle = |\Psi_1\rangle \otimes \cdots \otimes |\Psi_n \rangle$:
\begin{enumerate}
\item Let $x \in \{0, 1\}^n_{2\eps}$ be a uniformly random $2\eps$-biased string, meaning each coordinate is selected independently according to $\Pr[x_i = 1] = 2\eps$.
\item For each $i \in [n]$ such that $x_i = 0$, set $|\Psi_1\rangle \coloneqq  |d\rangle$.
\item Let $b$ be an arbitrary $\{0,1\}$-bit.  For each $i \in [n]$ such that $x_i = 1$,
	\begin{enumerate}
	\item if $b = 0$, then set $|\Psi_i \rangle$ to be a state vector sampled from $\rho$.
	\item if $b = 1$, then set $|\Psi_i \rangle$ to be a state vector sampled from $\sigma$.
	\end{enumerate}
\end{enumerate}
If~$b$ is~$0$, then the mixed state output by this procedure has spectrum $(1-2\eps, \frac{2\eps}{r-1}, \ldots, \frac{2\eps}{r-1})$, which is rank~$r$.
On the other hand, if~$b$ is~$1$, then the mixed state output by this procedure has spectrum $(1-2\eps, \frac{2\eps}{d-1}, \ldots, \frac{2\eps}{d-1})$, which because $d \gg r$ is $\eps$-far from having rank~$r$.

Let us consider the choice of $x$ in the first step, and set $\weight(\bx)$ to be the number of $1$'s in $x$.
In expectation, $\weight(\bx)$ will be $2\eps n$, and so by Markov's inequality $\weight(\bx)$ will be at most $200\eps n$ with probability at least $99/100$.
There must exist an $x$ with $\weight(x) \leq 200\eps n$ conditioned on which the algorithm succeeds with probability at least $3/5$, as otherwise it will succeed in total with probability at most $1/100 + 99/100 \cdot 3/5 < 2/3$.

Fix any such $x$.
The job of the algorithm is reduced to distinguishing between the cases when those $|\Psi_i\rangle$'s for which $x_i = 1$ came from $\rho$ which is maximally mixed on a subspace of dimension $(r-1)$ (when $b = 0$) or from $\sigma$ which is maximally mixed on a subspace of dimension $(d-1)$ (when $b = 1$).  Because $d \gg r$, we have by Theorem~\ref{thm:q-bday} that this requires at least $\Omega(r)$ copies to succeed with probability at least $3/5$.
Thus, we must have $200 \eps n \geq \Omega(r)$, in which case $n = \Omega(r/\eps)$.
\end{proof}

\section{The EYD lower bound (continued)}\label{sec:eyd-cont}

\newcommand{\unifpart}[1]{\mathsf{unif}_{#1}}

In this section, we prove Theorem~\ref{thm:small-n}.
\begin{named}{Theorem~\ref{thm:small-n} restated}
For every constant $C > 0$, there are constants $\delta, \eps >0$ such that
\begin{equation*}
\Pr_{\blambda \sim \SWunif{n}{d}}[\dtv{\underline{\blambda}}{\unif{d}} > \eps] \geq \delta
\end{equation*}
when $n \leq C d^2$ and $d$ is sufficiently large.
\end{named}

\begin{proof}
To prove Theorem~\ref{thm:small-n}, we show, at a high level, that when $n \leq C d^2$,
Biane's law of large numbers kicks in and $\overline{\blambda}$ approaches the limiting curve $\Omega_\theta$, for $\theta \coloneqq \frac{\sqrt{n}}{d}$.
Each of these curves is constantly far from the curve produced by the uniform partition, and the lower bound follows.
However, carrying out this proof involves some subtle argumentation and splitting of hairs which we will go into.

There is one regime where $\overline{\blambda}$ certainly does not approach $\Omega_\theta$: when $n$ is a fixed value independent of the value of~$d$, then $\overline{\blambda}$ will be always be constantly far from $\Omega_\theta$.
However, we can rule this case out by noting that when $n$ is too small as a function of~$d$, then any $\lambda = (\lambda_1, \ldots, \lambda_d)$ with $n$ boxes will have most of its~$\lambda_i$'s zero, and so $\underline{\lambda}$ will be far from uniform.
In particular, when $n = o(d)$, then we have that $\dtv{\underline{\blambda}}{\unif{d}} \rightarrow 1$ as $d \rightarrow \infty$.
As a result, for sufficiently large~$d$ we can immediately assume that $n \geq f(d)$, where $f(d)$ is any function which is both $\omega_d(1)$ and $o(d)$.
For concreteness, we will take $f(d) \coloneqq  \sqrt{d}$.

We are now in the regime where Biane's law of large numbers holds.
Theorem~\ref{thm:biane} tells us that if $\frac{\sqrt{n}}{d} \sim c$ for~$c$ some absolute constant, then there is some constant $d(c)>0$ such that for a random $\blambda \sim \SWunif{n}{d}$, $\overline{\blambda}$ is $\eps$-close (in $L^\infty$ distance) to $\Omega_{c}$ whenever $d \geq d(c)$.
The main difficulty we have in applying Biane's law of large numbers directly is that the function $d(c)$ is left unspecified and, for example, could be wildly different even for two close values of $c$.
This is problematic in our case, because for each value of~$d$, the ratio $\theta = \frac{\sqrt{n}}{d}$ may be any real number in the interval $[\sqrt{f(d)}/d, \sqrt{C}]$, and so $\theta$ may jump around and never converge to a fixed value~$c$.
In particular, an adversary could potentially choose~$n$ (and therefore~$\theta$) as a function of~$d$ cleverly so that for each~$d$, we have that $d < d(\theta)$, and so Biane's law of large numbers never applies.
Though seemingly unlikely, this possibility is not ruled out by the statements of known theorems.

Our goal now is to show that the convergence to the limiting shapes guaranteed by Biane's theorem happens at roughly the same rate for all values of $\theta$ in our interval.
First we will need a definition.

\begin{definition}
Given continual diagrams $f, g:\R\rightarrow \R$, the \emph{$L^1$ distance} between them is
\begin{equation*}
\dcurve{f}{g}:=\int_{\R} \left|f(x) - g(x)\right| \mathrm{d}x.
\end{equation*}
This defines a metric on the set of continual diagrams, and it is well-defined because $f(x)-g(x) = 0$ whenever $|x|$ is sufficiently large.
If $\lambda, \mu$ are both partitions of $n$, then $\dcurve{\overline{\lambda}}{\overline{\mu}} = 4\cdot \dtv{\underline{\lambda},\underline{\mu}}$.
\end{definition}

We will prove the following result:

\begin{theorem}\label{thm:all-converge}
Let $C > 0$ be an absolute constant, and let $f(d):\N\rightarrow \N$ be  $\omega_d(1)$.
Then for any constant $0 < \delta < 1$, if $f(d) \leq n \leq C d^2$, then
\begin{equation*}
\Pr_{\blambda \sim \SWunif{n}{d}}\left[\dcurve{\overline{\blambda}}{\Omega_\theta} \geq \delta\right] \leq \delta,
\end{equation*}
for sufficiently large~$d$, where $ \theta= \frac{\sqrt{n}}{d}$.
\end{theorem}

Let us now complete the argument assuming Theorem~\ref{thm:all-converge}.
For $\kappa >0$,
define the following continual diagram:
\begin{equation}\label{eq:unif}
\overline{\unifpart{\kappa}}(x)
	:=\left\{
		\begin{array}{cl}
			x+\frac{2}{\kappa}&\text{if $x \in (-\frac{1}{\kappa},\kappa-\frac{1}{\kappa}]$}\\
			-x + 2\kappa & \text{if $x \in (\kappa-\frac{1}{\kappa}, \kappa)$},\\
			|x| & \text{otherwise}.
		\end{array}\right.
\end{equation}
To see how such a function arises, consider the uniform ``partition'' $\left(\frac{n}{d}, \ldots, \frac{n}{d}\right)$ (``partition'' being in quotation marks because $\frac{n}{d}$ may not be integral).
Drawing this in the French notation gives a rectangle of width $\frac{n}{d}$ and height $d$ whose bottom-left corner is the origin.
Drawing this in the Russian notation and dilating by a factor of $1/\sqrt{n}$ therefore gives the curve $\overline{\unifpart{\theta}}(x)$.
One consequence of this is that if $\lambda$ is a partition of~$n$, then
$\dcurve{\overline{\lambda}}{\overline{\unifpart{\theta}}} = 4\cdot \dtv{\underline{\lambda},\unif{d}}$.

Define the function $\Delta:(0, \sqrt{C}]\rightarrow \R^{\geq 0}$
by $\Delta(\kappa) := \dcurve{\overline{\unifpart{\kappa}}}{\Omega_\kappa}$.
When $\kappa < .3$, $\Delta(\kappa) > .5$ for all $c$.
This is because $\Omega_\kappa(x) = -x$ for all $x \leq -2$ regardless of $\kappa$,
whereas $\overline{\unifpart{\kappa}}(x)=-x + 2\kappa$ in $(\kappa-\frac{1}{\kappa}, -2]$.
Because $\kappa < .3$,
\begin{equation*}
\dcurve{\overline{\unifpart{\kappa}}}{\Omega_\kappa}
= \int_{\R} \left|\overline{\unifpart{\kappa}}(x) - \Omega_\kappa(x)\right| \mathrm{d}x
\geq 2\kappa \cdot \left(\tfrac{1}{\kappa} - 2 - \kappa\right) \geq 0.5.
\end{equation*}
Now, let us lower-bound $\Delta(\kappa)$ when $\kappa \geq .3$.
Write $I$ for the interval $[.3, \sqrt{C}]$.  (If $.3 > \sqrt{C}$ then this step can be skipped.)
To begin, we note that $\Delta(\kappa)$ is continuous on $I$.
By comparing~\eqref{eq:unif} with Theorem~\ref{thm:biane}, it is easy to see that $\Delta(\kappa) > 0$ for all $\kappa> 0$.
We can now apply the extreme value theorem, which implies that $\Delta$ achieves its minimum on $I$ at some fixed point $\kappa^* \in I$.
We therefore have that $\Delta(\kappa) \geq \Delta(\kappa^*) > 0$ for all $\kappa \in I$.

Combining the last two paragraphs, we now know that there is some value
\begin{equation*}
\delta := \min\{0.5, \Delta(\kappa^*)\} > 0
\end{equation*}
such that $\Delta(\kappa) > \delta$ for all $\kappa \in (0, \sqrt{C}]$.
Crucially, $\delta$ is an absolute constant which depends only on the constant $C$ and is independent of $n$ and~$d$.
Now, let us apply Theorem~\ref{thm:all-converge} with the values $f(d) = \sqrt{d}$, $C$, and $\frac{\delta}{2}$.
Then with probability at least $1 - \frac{\delta}{2}$,
$\dcurve{\overline{\blambda}}{\Omega_\theta} < \frac{\delta}{2}$.  When this occurs,
\begin{equation*}
\dtv{\underline{\blambda}}{\unif{d}}
= \frac{1}{4} \dcurve{\overline{\blambda}}{\overline{\unifpart{\theta}}}
\geq \frac{1}{4} \left(\dcurve{\Omega_\theta}{\overline{\unifpart{\theta}}} - \dcurve{\overline{\blambda}}{\Omega_\theta}\right)
\geq \frac{\delta}{8},
\end{equation*}
where the second step follows from the triangle inequality,
and the third step uses the fact that $\dcurve{\Omega_\theta}{\overline{\unifpart{\theta}}} = \Delta(\theta) \geq \delta$.
This proves the theorem with the parameters $1-\frac{\delta}{2}$ and $\frac{\delta}{8}$.
\end{proof}

It remains to prove Theorem~\ref{thm:all-converge}, and this is done in the next subsection.
\subsection{Proof of Theorem~\ref{thm:all-converge}}
Our goal is to give a rate of convergence of $\overline{\blambda}$ to $\Omega_\theta$ which depends only on~$d$ and is independent of~$n$.
To do this, we will show that standard law of large numbers arguments give convergence rates of this form.
Biane's~\cite{Bia01} proof of the law of large numbers for the Schur-Weyl distribution does not use Kerov's algebra of observables.
Instead, we will follow the proof of the law of large numbers (second form) for the Plancherel distribution in~\cite[Theorem~$5.5$]{IO02}
and use results from~\cite{Mel10a} to extend this proof to the Schur--Weyl distribution.
We emphasize that our proof contains no ideas not already found in~\cite{IO02,Mel10a},
and that our goal is just to show that proper bookkeeping of their arguments yields our Theorem~\ref{thm:all-converge}.
(Finally, we note that Meliot~\cite{Mel10a} also sketches a proof the law of large numbers for the Schur--Weyl distribution using Kerov's algebra of observables at the beginning of his Section~$3$.)

Write $\Delta_\blambda(x) \coloneqq \overline{\blambda}(x) - \Omega_\theta(x)$.
Because $\overline{\blambda}$ and $\Omega_\theta$ are both continual diagrams,
we know that $\Delta_\blambda$ is supported (i.e., nonzero) on a finite interval.
We will need a stronger property, which is that the width of this interval does not grow with~$d$ (or, equivalently, with~$n$).
To show this, note that $\Delta_\blambda(x)$ is zero when both $\Omega_\theta(x) = |x|$ and $\overline{\blambda}(x) = |x|$.
For the first of these, we can consult Theorem~\ref{thm:biane} and see that $\Omega_\theta(x) = |x|$ outside the interval $[-2, \theta+2]$.
On the other hand, $\overline{\blambda}(x)$ does not equal $|x|$ outside a constant-width interval
for all $\blambda \sim \SWunif{n}{d}$.
(For example, with nonzero probability $\blambda = (n)$, in which case $\overline{\blambda}(x) = |x|$ only outside the interval $(-1/\sqrt{n},\sqrt{n})$.)
However, the next proposition shows that our desired property occurs with high probability.

\begin{proposition}\label{prop:constant-width-interval}
With probability $1-\frac{\delta}{2}$,  $\overline{\blambda}(x) \neq |x|$ only on an interval of width $w = O_\delta(1)$.
\end{proposition}
\begin{proof}
We will show that $\blambda_1$ and $\blambda_1' \leq \beta\sqrt{n}$, each with probability $1-\delta/4$, for some constant $\beta$ which depends only on $\delta$ (and $C$).
The proposition will then follow from the union bound,
as $\overline{\blambda} = |x|$ outside the interval $[-\blambda_1/\sqrt{n}, \blambda_1/\sqrt{n}]$.
By Proposition~\ref{prop:increasing-seq},
\begin{equation*}
\Pr[\blambda_1 \geq \beta\sqrt{n}], \Pr[\blambda_1' \geq \beta\sqrt{n}] \leq
\left(\frac{(1+\beta\theta) e^2}{\beta^2}\right)^{\beta\sqrt{n}}
\leq \frac{(1+\beta\theta) e^2}{\beta^2}
\leq \frac{(1+\beta\sqrt{C}) e^2}{\beta^2}.
\end{equation*}
This can be made less than $\delta/4$ by choosing $\beta$ to be a sufficiently large function of $C$ and $\delta$.
\end{proof}

Let $I'$ be the constant-width interval guaranteed by Proposition~\ref{prop:constant-width-interval}.
Clearly, $I'$ contains the point zero.  Thus, if we define
\begin{equation*}
I := [-2, \theta+2] \cup I'
\end{equation*}
then this is a single interval of width $w = O_\delta(1)$.
This motivates the following definition:
\begin{definition}
We say that $\lambda$ is \emph{usual} if $\Delta_\lambda$ is supported on $I$.
By the previous discussion, a random $\blambda$ is usual with probability $1-\delta/2$.
\end{definition}

Let us condition $\blambda$ on it being usual,
and let us suppose that $\dcurve{\overline{\blambda}}{\Omega_\theta} \geq \delta$.
Then there is some point $x \in I$ such that $\left|\Delta_\blambda(x)\right| \geq \frac{\delta}{w}$.
Now we will use the fact that $\Omega_\theta$ and $\overline{\blambda}$ are continual diagrams,
which implies that they are both $1$-Lipschitz, and therefore $\Delta_\blambda$ is $2$-Lipschitz.
Then if we consider the subinterval $I_x \subseteq I$ defined as $I_x := [x-\frac{\delta}{4w}, x + \frac{\delta}{4w}]$,
this Lipschitz property implies that $\left|\Delta_\blambda(y) \right|\geq \frac{\delta}{2w}$ for all $y \in I_x$.
(That $I_x$ is contained in $I$ follows from the fact that $\Delta_\blambda$ is nonzero on $I_x$ and $\blambda$ is usual.)
We note that the width of $I_x$ is $\frac{\delta}{2w}$.

Let $\calJ$ be a set of $\lceil\frac{4w^2}{\delta}\rceil$ closed intervals of width $\frac{\delta}{4w}$ which cover $I$.
\ignore{(Concretely, $\calJ$ may be taken to be XXX.)}
These intervals are chosen to have half the width of $I_x$,
the result being that there is some interval $J^* \in \calJ$ which is completely contained in $I_x$.
For each interval $J \in \calJ$, let $\Psi_J:\R\rightarrow\R^{\geq 0}$ be a continuous function supported on $J$
which satisfies $\int \Psi_J(y) dy = 1$ (such functions are known to exist; e.g., bump functions).
Then
\begin{equation*}
\left|\int_{-\infty}^\infty \Delta_{\blambda}(y) \Psi_{J^*}(y) dy\right|
\geq \min_{y \in I_x} \left|\Delta_{\blambda}(y) \right|\cdot \int_{-\infty}^\infty \Psi_{J^*}(y) dy
\geq \frac{\delta}{2w}.
\end{equation*}

By the Weierstrass approximation theorem, we can approximate each $\Psi_J$
with a polynomial function $\widetilde{\Psi}_J$ such that for each $x \in I$, $|\Psi_J(x) - \widetilde{\Psi}_J(x)| \leq \frac{\delta}{8w^3}$.
(Outside of $I$, $\widetilde{\Psi}_J$ can---and will---be an arbitrarily bad approximator for $\Psi_J$.)
Because $\Delta_\blambda$ is $2$-Lipschitz and~$\blambda$ is usual,
$|\Delta_\blambda(x)|\leq 2w$ for all $x \in I$ and is zero everywhere else.  As a result, for the interval $J^*$,
\begin{equation*}
\left|\int_{-\infty}^\infty \Delta_{\blambda}(y) \widetilde{\Psi}_{J^*}(y) dy\right|
\geq
\left|\int_{-\infty}^\infty \Delta_{\blambda}(y) \Psi_{J^*}(y) dy\right|
-
\left|\int_{-\infty}^\infty \Delta_{\blambda}(y) \left(\Psi_{J^*}(y)-\widetilde{\Psi}_{J^*}(y) \right) dy\right|
\geq
\frac{\delta}{4w}.
\end{equation*}
The first inequality uses the triangle inequality, and the second inequality uses crucially the fact that $\Delta_{\blambda}$ is zero outside $I$.

In summary, we have
\begin{equation}\label{eq:dunno-what-to-call-this}
\Pr_{\blambda \sim \SWunif{n}{d}}\left[\dcurve{\overline{\blambda}}{\Omega_\theta} \geq \delta\right]
\leq
\Pr_{\blambda \sim \SWunif{n}{d}}\left[\exists J \in \calJ : \left|\int_{-\infty}^\infty \Delta_{\blambda}(y) \widetilde{\Psi}_J(y) dy\right| \geq \frac{\delta}{4w}\right] + \frac{\delta}{2},
\end{equation}
where the $\delta/2$ comes from the event that $\blambda$ is not usual.
We will therefore show that $\left|\int \Delta_\blambda(y) \widetilde{\Psi}_J(y) dy\right|$ is at most $\frac{\delta}{4w}$ for all $J \in \calJ$ with probability at least $1-\frac{\delta}{2}$.
By the union bound, it suffices to show that for each $J \in \calJ$,  $\left|\int \Delta_\blambda(y) \widetilde{\Psi}_J(y) dy\right| \leq \frac{\delta}{4w}$ with probability at least $1-\frac{\delta}{2\cdot |\calJ|}$.

Let $m$ be the maximum degree of the $\widetilde{\Psi}_J$ functions, for all $J \in \calJ$.
Fix an interval $J \in \calJ$.
Then we can write
\begin{equation}\label{eq:sum-of-moments}
\widetilde{\Psi}_J(x) = \sum_{k=0}^m a_J^{(k)} x^k
\quad
\text{and}
\quad
\int_{-\infty}^\infty \Delta_{\lambda}(y) \widetilde{\Psi}_J(y) dy = \sum_{k=0}^m a_J^{(k)} \int_{-\infty}^\infty x^k \Delta_{\lambda}(x) dx,
\end{equation}
where the $a_J^{(k)}$'s are constants.
The following proposition, found in~\cite[Lemma 7]{Mel10a}, gives a nice expression for the integrals on the right-hand side.
\begin{proposition}\label{prop:meliot-integral}
Let $k \geq 1$.  Then
\begin{equation*}
\int_{-\infty}^{\infty} x^k \Delta_{\lambda}(x) dx = \frac{2 \cdot\widetilde{q}_{k+1}(\lambda)}{(k+1)\sqrt{n}},
\end{equation*}
where $\widetilde{q}_k(\lambda)$ is the quantity defined as
\begin{equation*}
\widetilde{q}_k(\lambda)
\coloneqq  \frac{\widetilde{p}_{k+1}(\lambda)}{(k+1) n^{k/2}}
- \sum_{\ell = 1}^{\lfloor \frac{k+1}{2} \rfloor} \frac{k^{\downarrow 2\ell - 1}}{(k+1-\ell) \ell! (\ell-1)!}\cdot \frac{n^{k/2 + 1 - \ell}}{d^{k+1 - 2\ell}}.
\end{equation*}
\end{proposition}
The key fact we will use is that we can upper bound the right-hand side of Equation~\eqref{eq:sum-of-moments} by a quantity which decays with~$d$, independent of the value of~$n$.  This is the subject of the following lemma.
\begin{lemma}\label{lem:not-very-big}
The random variable $\left|\frac{\widetilde{q}_k(\blambda)}{\sqrt{n}}\right|$, for $\blambda \sim \SWunif{n}{d}$, has mean $o_d(1)$, for all $f(d) \leq n \leq C d^2$.
\end{lemma}
Applying Proposition~\ref{prop:meliot-integral} and Lemma~\ref{lem:not-very-big} to Equation~\eqref{eq:sum-of-moments}, we see that
$\E_{\blambda \sim \SWunif{n}{d}}\left|\int \Delta_{\blambda}(x) \widetilde{\Psi}_J(x) dx\right|$ is $o_d(1)$.
We may take $d$ large enough to make this quantity arbitrarily small.
Thus, select $d_J$ so that for all $d \geq d_J$, this expectation is at most $\frac{\delta^2}{8w\cdot|\calJ|}$.
Then by Markov's inequality, $\left|\int \Delta_{\blambda}(x) \widetilde{\Psi}_J(x) dx\right| \leq \frac{\delta}{4w}$
with probability at least $1 - \frac{\delta}{2\cdot|\calJ|}$.
If we set $d_0$ to be the max of $d_J$ over all $J \in \calJ$,
then by Equation~\eqref{eq:dunno-what-to-call-this},
$\Pr_{\blambda \sim \SWunif{n}{d}}\left[\dcurve{\overline{\blambda}}{\Omega_\theta} \geq \delta\right] \leq \delta$
so long as $d\geq d_0$, and we are done.

Now we turn to the proof of Lemma~\ref{lem:not-very-big}.
\begin{proof}[Proof of Lemma~\ref{lem:not-very-big}]
Define
\begin{equation*}
X_k(\lambda) \coloneqq  \sum_{\mu : \weight(\mu)=k} \frac{k^{\downarrow \ell(\mu)}}{m(\mu)} \cdot \psharp{\mu}(\lambda)
\end{equation*}
and
\begin{equation}\label{eq:what-the-heck}
\qsharp{k}(\lambda)
\coloneqq  \frac{X_{k+1}(\lambda)}{(k+1) n^{k/2}}
- \sum_{\ell = 1}^{\lfloor \frac{k+1}{2} \rfloor} \frac{k^{\downarrow 2\ell - 1}}{(k+1-\ell) \ell! (\ell-1)!}\cdot \frac{n^{k/2 + 1 - \ell}}{d^{k+1 - 2\ell}}.
\end{equation}
Then by Proposition~\ref{prop:p-to-psharp}, $\widetilde{q}_k(\lambda)$ and $\qsharp{k}(\lambda)$ differ from each other by
$n^{-k/2}$ times an observable  $\calO(\lambda)$ of weight~$k$.  Thus,
\begin{equation*}
\E_{\blambda \sim \SWunif{n}{d}} \left|\frac{\widetilde{q}_k(\blambda)}{\sqrt{n}}\right|
\leq
\E_{\blambda \sim \SWunif{n}{d}} \left|\frac{\qsharp{k}(\blambda)}{\sqrt{n}}\right|
+
\E_{\blambda \sim \SWunif{n}{d}} \left|\frac{\calO(\blambda)}{n^{(k+1)/2}}\right|.
\end{equation*}
By Cauchy--Schwarz, $\E | \calO(\blambda)/n^{(k+1)/2}| \leq \sqrt{\E\calO(\blambda)^2/n^{k+1}}$.
Because $\calO$ has weight~$k$, $\calO^2$ has weight~$2k$.
As a result, we can use the next proposition to bound the contribution from this term by $o_d(1)$.
\begin{proposition}\label{prop:pretty-small}
Let $\calO(\lambda)$ be an observable of weight at most~$2k$.  Then
\begin{equation*}
\E_{\blambda \sim \SWunif{n}{d}} \left[\frac{\calO(\blambda)}{n^{k+1}}\right] = o_d(1).
\end{equation*}
\end{proposition}
\begin{proof}
\ignore{
Because $\calO$ has weight at most~$2k$, we can write it as
\begin{equation*}
\calO(\lambda) = \sum_{\mu : |\mu| + \ell(\mu)\leq 2k} \alpha_\mu \cdot \psharp{\mu}(\lambda),
\end{equation*}
for some coefficients $\alpha_\mu$.
This follows from the fact that the $\psharp{}$ polynomials of weight~$2k$ span the set of observables of weight~$2k$, i.e.\ XXX.
Thus, the expectation we care about is
\begin{equation}\label{eq:expectation-expansion}
\E_{\blambda \sim \SWunif{n}{d}} \left[\frac{\calO(\blambda)}{n^{k+1}}\right]
= \sum_{\mu : |\mu| + \ell(\mu)\leq 2k} \alpha_\mu \cdot \E_{\blambda \sim \SWunif{n}{d}} \left[\frac{\psharp{\mu}(\blambda)}{n^{k+1}}\right].
\end{equation}
Fix a partition $\mu$ satisfying $|\mu| +\ell(\mu) \leq 2k$.
By XXX we can bound its contribution to this sum by
\begin{equation*}
0 \leq \E_{\blambda \sim \SWunif{n}{d}} \left[\frac{\psharp{\mu}(\blambda)}{n^{k+1}}\right]
\leq \frac{1}{n^{k+1}}\cdot n^{|\mu|} d^{\ell(\mu) - |\mu|}
= \frac{n^{|\mu|}}{n^{k+1}}\cdot\frac{d^{\weight(\mu)}}{d^{2|\mu|}}.
\end{equation*}
If $|\mu|<k+1$, then this expression is at most $1/n$, which is $o_d(1)$ because $n \geq f(d) = \omega_d(1)$.
On the other hand, if $|\mu|\geq k+1$, then for all $n \leq Cd^2$ this expression is at most
\begin{equation*}
\frac{(Cd^2)^{|\mu|}}{(Cd^2)^{k+1}}\cdot\frac{d^{\weight(\mu)}}{d^{2|\mu|}}
= C^{|\mu| - (k+1)}\cdot\frac{d^{\weight(\mu)}}{d^{2(k+1)}},
\end{equation*}
which is $o_d(1)$ as $\weight(\mu) \leq 2k$.  As a result, the sum in Equation~\eqref{eq:expectation-expansion} is~$o_d(1)$, and we are done.
}
As in the proof of Lemma~\ref{lem:errL}, this reduces to showing that
$\E_{\blambda \sim \SWunif{n}{d}} \left[\psharp{\mu}(\blambda)/n^{k+1}\right] = o_d(1)$,
where $\mu$ is a partition of weight~$2k$, i.e.\ $|\mu| + \ell(\mu) \leq 2k$.
By Corollary~\ref{cor:sw-expect},
\begin{equation*}
\E_{\blambda \sim \SWunif{n}{d}} \left[\frac{\psharp{\mu}(\blambda)}{n^{k+1}}\right]
=  \frac{\falling{n}{|\mu|}}{n^{k+1}}\cdot  \frac{d^{\ell(\mu)}}{ d^{|\mu|}}
\leq \frac{n^{|\mu|}}{n^{k+1}}\cdot  \frac{d^{\ell(\mu)}}{ d^{|\mu|}}
= \frac{n^{|\mu|}}{n^{k+1}}\cdot\frac{d^{\weight(\mu)}}{d^{2|\mu|}}.
\end{equation*}
If $|\mu|<k+1$, then this expression is at most $1/n$, which is $o_d(1)$ because $n \geq f(d) = \omega_d(1)$.
On the other hand, if $|\mu|\geq k+1$, then for all $n \leq Cd^2$ this expression is at most
\begin{equation*}
\frac{(Cd^2)^{|\mu|}}{(Cd^2)^{k+1}}\cdot\frac{d^{\weight(\mu)}}{d^{2|\mu|}}
\leq C^{|\mu| - (k+1)}\cdot\frac{d^{\weight(\mu)}}{d^{2(k+1)}},
\end{equation*}
which is $o_d(1)$ as $\weight(\mu) \leq 2k$.
\end{proof}

It remains to bound $\E|\qsharp{k}(\blambda)/\sqrt{n}|$ by $o_d(1)$.
First, we will show that $\qsharp{k}(\blambda)$ can be viewed as (approximately) computing the deviation of a certain random variable from its mean.
To do this, let us compute the mean of the first term on the right-hand side of Equation~\eqref{eq:what-the-heck}.
\begin{align*}
\E_{\blambda \sim \SWunif{n}{d}} \frac{X_{k+1}(\blambda)}{(k+1)n^{k/2}}
&= \frac{1}{(k+1)n^{k/2}}\cdot
	\sum_{\mu:\weight(\mu) = k+1}\frac{(k+1)^{\downarrow \ell(\mu)}}{m(\mu)}\cdot \frac{n^{\downarrow |\mu|}}{d^{|\mu|-\ell(\mu)}}\\
&= \frac{1}{(k+1)n^{k/2}}\cdot
	\sum_{\ell=1}^{\lfloor \frac{k+1}{2}\rfloor} \frac{(k+1)^{\downarrow \ell} n^{\downarrow k+1 - \ell}}{d^{k+1-2\ell}}
	\sum_{\mu:\weight(\mu) = k+1}\frac{1}{m(\mu)}\\
&= \frac{1}{(k+1)n^{k/2}}\cdot
	\sum_{\ell=1}^{\lfloor \frac{k+1}{2}\rfloor} \frac{(k+1)^{\downarrow \ell} n^{\downarrow k+1 - \ell}}{d^{k+1-2\ell}}
 	\cdot \frac{1}{\ell!}\binom{k-\ell}{\ell-1}\\
&= \sum_{\ell = 1}^{\lfloor \frac{k+1}{2} \rfloor} \frac{k^{\downarrow 2\ell - 1}}{(k+1-\ell) \ell! (\ell-1)!}
		\cdot \frac{n^{\downarrow k + 1 - \ell}}{n^{k/2}\cdot d^{k+1 - 2\ell}},
\end{align*}
where the third equality follows from~\cite[Lemma~$11$]{Mel10a}.
As a result, the difference
\begin{equation*}
\E_{\blambda \sim \SWunif{n}{d}} \frac{X_{k+1}(\blambda)}{(k+1)n^{k/2}}
-
\sum_{\ell = 1}^{\lfloor \frac{k+1}{2} \rfloor} \frac{k^{\downarrow 2\ell - 1}}{(k+1-\ell) \ell! (\ell-1)!}\cdot \frac{n^{k/2 + 1 - \ell}}{d^{k+1 - 2\ell}}
\end{equation*}
can be written as a sum over terms of the form $a\cdot n^b/d^{k+1-2\ell}$, where $a$ is a constant coefficient, $1\leq b \leq k/2 - \ell$, and $1 \leq \ell \leq \lfloor \frac{k+1}{2}\rfloor$.  Given that $n \leq C d^2$, each of these terms if $\pm o_d(1)$.
Thus, if we set
\begin{equation*}
q_k(\lambda) \coloneqq  \frac{X_{k+1}(\lambda)}{(k+1)n^{k/2}} - \E_{\blambda \sim \SWunif{n}{d}}  \frac{X_{k+1}(\blambda)}{(k+1)n^{k/2}},
\end{equation*}
then
\begin{equation*}
\E_{\blambda \sim \SWunif{n}{d}} \left|\frac{\qsharp{k}(\blambda)}{\sqrt{n}}\right|
\leq
\E_{\blambda \sim \SWunif{n}{d}} \left| \frac{q_k(\blambda)}{\sqrt{n}}\right| + o_1(d).
\end{equation*}

Finally, we show that $\E \left|q_k(\blambda)/\sqrt{n}\right| = o_d(1)$.
By Cauchy--Schwarz,
\begin{equation*}
\E \left|\frac{q_k(\blambda)}{\sqrt{n}}\right| \leq \sqrt{\E \left(\frac{q_k(\blambda)}{\sqrt{n}}\right)^2},
\end{equation*}
so it suffices to show that $\E \left(q_k(\blambda)/\sqrt{n}\right)^2 = o_d(1)$.
This expectation is simply the variance of the random variable $X_{k+1}(\blambda)/(k+1)n^{(k+1)/2}$,
which itself is a weighted sum of a constant number of random variables of the form $\psharp{\mu}(\blambda)/n^{(k+1)/2}$,
where $\weight(\mu) = k+1$.
An easy application of Cauchy--Schwarz shows that the variance of a weighted sum of a constant number of random variables
is $o_d(1)$ if the variance of each random variables is $o_d(1)$.  Thus, we will show that
$\Var[\psharp{\mu}(\blambda)/n^{(k+1)/2}] = o_d(1)$ for all $\weight(\mu) = k+1$.

Fix a partition $\mu$ of weight $k+1$.  Then
\begin{equation*}
\Var\left[\frac{\psharp{\mu}(\blambda)}{n^{(k+1)/2}}\right]
 = \E_{\blambda \sim \SWunif{n}{d}}\left[\frac{1}{n^{(k+1)/2}}\left(\psharp{\mu}(\blambda) \psharp{\mu}(\blambda) - \E[\psharp{\mu}]^2 \right)\right]
\end{equation*}
By Proposition~\ref{prop:should-be-in-a-structure-constants-discussion-or-theorem}, $\psharp{\mu}(\lambda) \cdot \psharp{\mu}(\lambda) = \psharp{\mu \cup \mu}(\lambda) + \calO(\lambda)$,
where $\calO(\lambda)$ is an observable of weight at most $2\cdot \weight(\psharp{\mu}) - 2 = 2k$.
Then
\begin{equation*}
\Var\left[\frac{\psharp{\mu}(\blambda)}{n^{(k+1)/2}}\right]
= \E_{\blambda \sim \SWunif{n}{d}}\left[\frac{1}{n^{k+1}} \cdot \left(\psharp{\mu \cup \mu}(\blambda)-\E[\psharp{\mu}]^2\right)\right] + \E_{\blambda \sim \SWunif{n}{d}}\left[\frac{1}{n^{k+1}} \cdot \calO(\blambda)\right].
\end{equation*}
The second term is $\pm o_d(1)$ by Proposition~\ref{prop:pretty-small}.  As for the first term, Corollary~\ref{cor:sw-expect}, shows that it equals
\begin{equation}\label{eq:almost-there}
\frac{1}{n^{k+1}} \cdot \left(n^{\downarrow 2 |\mu|} d^{2\ell(\mu)-2|\mu|} - n^{\downarrow |\mu|} n^{\downarrow |\mu|}d^{2\ell(\mu)-2|\mu|}\right)
= \frac{1}{d^{4|\mu| - 2(k+1)}} \cdot \left(\frac{n^{\downarrow 2|\mu|} - (n^{\downarrow |\mu|})^2}{n^{k+1}}\right),
\end{equation}
where we used the fact that $\ell(\mu) = \weight(\mu) - |\mu| = k+1 - |\mu|$.
The highest-degree term of both $n^{\downarrow 2|\mu|}$ and $(n^{\downarrow |\mu|})^2$ is $n^{2|\mu|}$, so we can write
\begin{equation*}
\eqref{eq:almost-there}= \frac{1}{d^{4|\mu| - 2(k+1)}} \cdot \sum_{b = -(k+1)}^{2|\mu| - (k + 2)} \alpha_b \cdot n^b
\end{equation*}
for some constants $\alpha_b$.
When $b < 0$, $n^b/d^{4|\mu| - 2k - 2}\leq 1/n$, which is $o_d(1)$ because $n \geq f(d) = \omega_d(1)$.
On the other hand, when $b \geq 0$, then this term is $o_d(1)$ because $n \leq C d^2$.
\end{proof}

\ignore{
\subsection{Some algebra calculations}\label{sec:tedious-sec}

In this section, we prove Fact~\ref{fact:tedious-fact}.
\begin{named}{Fact~\ref{fact:tedious-fact} restated}
Let $c \in \R$.  Then
\begin{itemize}
\item $\pstarc{2}{c} = (-2c-1)\psharp{(1)} + \psharp{(2)}$, and
\item $\pstarc{4}{c} = (-4c^3-6c^2-4c-1)\psharp{(1)}+(6c^2+6c+4)\psharp{(2)}+(-6c-3)\psharp{(1, 1)}+(-4c-2)\psharp{(3)}+4\psharp{(2,1)}+\psharp{(4)}$.
\end{itemize}
\end{named}
\begin{proof}
Let $k\geq 1$.  We begin by expanding $\pstarc{k}{c}$ as a linear combination of $\pstar{j}$'s.
\begin{align*}
\pstarc{k}{c}(\lambda)
&= \sum_{i=1}^\infty (\lambda_i - i - c)^k - (-i-c)^k\\
&= \sum_{i=1}^\infty (\lambda_i + \tfrac{1}{2}+(-c-\tfrac{1}{2}))^k - (-i +\tfrac{1}{2}+(-c-\tfrac{1}{2}))^k\\
&= \sum_{i=1}^\infty \sum_{j=1}^k \binom{k}{j}\left((\lambda_i-i+\tfrac{1}{2})^j(-c-\tfrac{1}{2})^{k-j} - (-i+\tfrac{1}{2})^j(-c-\tfrac{1}{2})^{k-j}\right)\\
& = \sum_{j=1}^k \binom{k}{j}(-c-\tfrac{1}{2})^{k-j} \sum_{i=1}^\infty (\lambda_i-i+\tfrac{1}{2})^j - (-i + \tfrac{1}{2})^j\\
& = \sum_{j=1}^k \binom{k}{j}(-c - \tfrac{1}{2})^{k-j} \pstar{j}(\lambda),
\end{align*}
where in the third step the sum begins at $j =1$ rather than $j= 0$ because the summand is zero when $j = 0$.
In our case, we care about $k = 2,4$.
\begin{align*}
\pstarc{2}{c} &= 2 (-c-\tfrac{1}{2}) \pstar{1} + \pstar{2},\\
\pstarc{4}{c} &= 4 (-c-\tfrac{1}{2})^3 \pstar{1} + 6 (-c-\tfrac{1}{2})^2 \pstar{2}
				+ 4 (-c-\tfrac{1}{2}) \pstar{3} + \pstar{4}.
\end{align*}
\rnote{XXXThis can be rewritten, because we already have a section about pstars to psharps.  The only annoyance is going from products of psharps to single psharps; i.e., the stupid structure constant stuff}Next, we appeal to \cite[Proposition~$2.7$]{IO02}, which shows how to express the $\widetilde{p}_k$ polynomials as linear combinations of the $\pstar{j}$ polynomials.  For the values of $k = 2, 3, 4, 5$, their proposition gives
\begin{align*}
\widetilde{p}_2 = 2 \pstar{1},\quad
\widetilde{p}_3  = 3\pstar{2},\quad
\widetilde{p}_4  = 4 \pstar{3} + \pstar{1},\quad
\widetilde{p}_5  = 5\pstar{4} + \frac{5}{2}\pstar{2}.
\end{align*}
(The expansion given in their Proposition~$2.7$ includes $\pstar{0}$ terms, which are identically zero.  Thus, we can omit these.)
Rearranging these yields
\begin{equation*}
\pstar{1} = \frac{1}{2} \widetilde{p}_2,\quad
\pstar{2} = \frac{1}{3} \widetilde{p}_3,\quad
\pstar{3} = \frac{1}{4} \widetilde{p}_4 - \frac{1}{8} \widetilde{p}_2,\quad
\pstar{4} = \frac{1}{5} \widetilde{p}_5 - \frac{1}{6} \widetilde{p}_3.
\end{equation*}
Finally, the end of Section~$2$ in~\cite{Mel10a} gives the following explicit formulas.
\begin{equation*}
\widetilde{p}_2 = 2 \psharp{(1)},\quad
\widetilde{p}_3 = 3 \psharp{(2)},\quad
\widetilde{p}_4 = 4 \psharp{(3)} + 6 \psharp{(1,1)} + 2 \psharp{(1)},\quad
\widetilde{p}_5 = 5 \psharp{(4)} + 20 \psharp{(2,1)} + 15 \psharp{(2)}.
\end{equation*}
Fact~\ref{fact:tedious-fact} now follows by performing appropriate substitutions.
\end{proof}

}

\ignore{

\newpage \textbf{XXXXXXXXXX old stuff}
\subsection{Preliminaries}

We begin with some background material.  The \emph{shifted Schur functions} are defined as
\begin{definition}
Let $\lambda = (\lambda_1, \ldots, \lambda_n)$ be a partition.  Then
\begin{equation*}
s^*_\lambda(x_1, \ldots, x_n) \coloneqq  \frac{\det[(x_i + n - i)^{\downarrow \lambda_j + n -j}]}{\det[(x_i+n-i)^{\downarrow n-j}]}.
\end{equation*}
\end{definition}
\noindent
The shifted Schur polynomials were introduced and studied in the work of Okounkov and Olshanski~\cite{OO98b}.
They are a shifted analogue of the so-called \emph{factorial Schur functions} of~[XXX] (see also~[XXX]).
For us, they are important mainly because they appear in the following formula for ratios of Schur functions.
\begin{fact}\label{fact:binomial-formula}
Let $\lambda = (\lambda_1, \ldots, \lambda_n)$ be a partition.  Then
\begin{equation*}
\frac{s_\lambda(1 + x_1,\ldots, 1 + x_n)}{s_\lambda(1, \ldots, 1)}
= \sum_{\mu} \frac{s^*_\mu(\lambda) s_\mu(x)}{n^{\downarrow\mu}}.
\end{equation*}
\end{fact}
\noindent
This formula appears as Theorem~$5.1$ in~\cite{OO98b} and is known as the \emph{binomial formula}.
It can be thought of as the Taylor expansion for Schur polynomials around the point $(1,\ldots, 1)$.
(See~\cite{OO98a} for a paper which uses the binomial formula extensively.)

From Section~$15$ of~\cite{OO98b}, there is a linear isomorphism $\psi:\Lambda\rightarrow\Lambda^*$
which satisfies $\psi(s_\mu) = s^*_\mu$ and $\psi(p_\mu) = \psharp{\mu}$.
Applying this to Fact~\ref{fact:looks-like-projector-formula}, we get the following formula for shifted Schur functions.
\begin{fact}\label{fact:shifted-formula}
Let $\lambda = (\lambda_1, \ldots, \lambda_n)$ be a partition of~$k$.  Then
\begin{equation*}
s^*_\lambda(x_1, \ldots, x_n) = \E_{\pi \in \symm{k}}\left[\chi_\lambda(\pi)\cdot \psharp{\pi}(x)\right].
\end{equation*}
\end{fact}
\noindent
One interesting consequence of this formula is that it allows us to calculate an explicit value for the expectation of $s^*_\mu$ with respect to the Schur--Weyl distribution.
\begin{proposition}\label{prop:shifted-expectation}
Let $\mu$ be a partition of~$k$.  Then
\begin{equation*}
\E_{\lambda\sim\SWunif{n}{d}} s^*_\mu(\lambda) = \frac{\dim(\mu) n^{\downarrow k}d^{\downarrow\mu}}{d^k k!}.
\end{equation*}
\end{proposition}
\begin{proof}
By Facts~XXX and~\ref{fact:shifted-formula},
\begin{align*}
\E_{\lambda\sim\SWunif{n}{d}} s^*_\mu(\lambda)
&= \E_{\pi \in \symm{k}}\left[\chi_\mu(\pi)\cdot \E_{\lambda\sim\SWunif{n}{d}}\left[\psharp{\pi}(\lambda)\right]\right]\\
&=\frac{n^{\downarrow k}}{d^k}\E_{\pi \in \symm{k}}\left[\chi_\mu(\pi)\cdot d^{\ell(\pi)}\right].
\end{align*}
MORE PROOF HERE.
\end{proof}
\noindent
Though we won't use Proposition~\ref{prop:shifted-expectation} per se, we will use the following corollary.
\begin{corollary}\label{cor:pi-expectation}
Let $\mu$ be a partition of~$k$.  Then
\begin{equation*}
\E_{\pi \in \symm{k}}\left[\chi_\lambda(\pi)\cdot d^{\ell(\pi)}\right]= \frac{\dim(\mu) d^{\downarrow\mu}}{k!}.
\end{equation*}
\end{corollary}

\subsection{In need of a name}

\begin{theorem}\label{thm:ratio}
Given $x = (x_1, \ldots, x_d)$ satisfying $x_1 + \ldots + x_d = 0$, \rnote{do we need all $x_i \geq -1$?}let
\begin{equation*}
\schurrat{x}(\lambda)\coloneqq \frac{s_\lambda(1+x_1, \ldots, 1+x_d)}{s_\lambda(1, \ldots, 1)}.
\end{equation*}
Then over a random $\blambda \sim \SWunif{n}{d}$, $\schurrat{x}(\blambda)$ is a random variable with mean~$1$ and variance
\begin{equation*}
\Var[\schurrat{x}]
= \sum_{|\mu| > 0}\frac{n^{\downarrow |\mu|}}{d^{|\mu|}} \cdot\frac{s_\mu(x)^2}{d^{\downarrow \mu}}.
\end{equation*}
\end{theorem}
\begin{proof}
By Proposition~\ref{prop:schur-probability}, $\dim(\lambda)\cdot s_\lambda(\eta_1, \ldots, \eta_d)$ is the probability that the weak Schur sampling algorithm outputs~$\lambda$ given a state with spectrum~$\eta$.  Since these probabilities sum to one, we have that
\begin{equation}
0 = \sum_{\lambda \vdash n} \dim(\lambda)\cdot s_\lambda((1 + x_1)/d, \ldots, (1+x_d)/d)
	- \dim(\lambda)\cdot s_\lambda(1/d, \ldots, 1/d).\label{eq:zero-diff}
\end{equation}
By definition, $\dim(\lambda)\cdot s_\lambda(1/d, \ldots, 1/d)$ gives the probability that the Schur--Weyl distribution outputs $\lambda$.  As a result,
\begin{align*}
& \sum_{\lambda \vdash n} \dim(\lambda)\cdot s_\lambda((1 + x_1)/d, \ldots, (1+x_d)/d) - \dim(\lambda)\cdot s_\lambda(1/d, \ldots, 1/d)\\
 =& \sum_{\lambda \vdash n} \dim(\lambda)\cdot s_\lambda(1/d, \ldots, 1/d)
		\cdot \left(\frac{s_\lambda((1 + x_1)/d, \ldots, (1+x_d)/d)}{s_\lambda(1/d, \ldots, 1/d)}-1\right)\\
 =& \sum_{\lambda \vdash n} \dim(\lambda)\cdot s_\lambda(1/d, \ldots, 1/d)
		\cdot \left(\frac{s_\lambda(1 + x_1, \ldots, 1+x_d)}{s_\lambda(1, \ldots, 1)}-1\right)\\
=& \E_{\blambda \sim \SWunif{n}{d}}[\schurrat{x}(\blambda) - 1].
\end{align*}
Here the second equality follows from the fact that $s_\lambda$ is a homogeneous polynomial.  By Equation~\eqref{eq:zero-diff}, this expectation is zero, meaning that the expectation of $\schurrat{x}$ is one.

Now we compute the variance of $\schurrat{x}$.  By Fact~\ref{fact:binomial-formula},
\begin{equation*}
\schurrat{x}(\lambda) = \sum_{\mu} \frac{s^*_\mu(\lambda) s_\mu(x)}{n^{\downarrow\mu}}.
\end{equation*}
The constant term in this expansion corresponds to the empty partition $\mu = ()$.
In this case, $s^*_\mu(\lambda) s_\mu(x)/n^{\downarrow\mu} = 1$ for all $\lambda$ and $x$.
Thus,
\begin{equation*}
\schurrat{x}(\lambda)-1 = \sum_{|\mu|>0} \frac{s^*_\mu(\lambda) s_\mu(x)}{d^{\downarrow\mu}}.
\end{equation*}
As a result,
\begin{equation}\label{eq:var-of-schurrat}
\Var[\schurrat{x}]
= \sum_{|\mu|,|\rho|>0} \frac{s_\mu(x) s_\rho(x)}{d^{\downarrow\mu}\cdot d^{\downarrow\rho}}
\E_{\blambda\sim \SWunif{n}{d}}\left[s^*_\mu(\blambda) s^*_\rho(\blambda)\right].
\end{equation}
Our proof strategy will be to show that certain appropriately chosen subexpressions within this summation have expectation zero, and what remains is the expression in the theorem statement.  Our proof structure will begin with the expectation of simple subexpressions and scale these upwards until at the end we have reached $\Var[\schurrat{x}]$.

We begin with the following lemma.
\begin{lemma}\label{lem:lemon-pi}
Let $\pi \in S_k$ and let $\lambda\vdash k$ be a partition.  Then
\begin{equation*}
\E_{\sigma \in S_k}\left[\chi_{\lambda}(\sigma) \cdot d^{\ell(\pi\cdot\sigma)}\right] = \frac{\chi_{\lambda}(\pi)\cdot d^{\downarrow \lambda}}{k!}.
\end{equation*}
\end{lemma}
\begin{proof}
The bulk of the work in proving Lemma~\ref{lem:lemon-pi} goes towards proving the following identity.
\begin{proposition}\label{prop:character-identity}
Let $\lambda,\mu \vdash k$ and $\alpha, \beta \in S_k$.  Then
\begin{equation*}
\sum_{\sigma \sim \mu}  \chi_\lambda(\alpha)\cdot \chi_\lambda(\beta\cdot \sigma)
=\sum_{\sigma \sim \mu} \chi_\lambda(\beta)\cdot \chi_\lambda(\alpha \cdot \sigma).
\end{equation*}
\end{proposition}
\begin{proof}
By definition, $\psharp{\rho}(\lambda)= |\lambda|! \cdot \chi_\lambda(\rho)/\dim \lambda$ when $|\rho| = |\lambda|$.
Thus, the statement of Proposition~\ref{prop:character-identity} is equivalent to
\begin{equation}
\sum_{\sigma \sim \mu}  \psharp{\alpha}(\lambda) \cdot \psharp{\beta \cdot \sigma}(\lambda)
=\sum_{\sigma \sim \mu} \psharp{\beta}(\lambda) \cdot \psharp{\alpha \cdot \sigma}(\lambda).\label{eq:left-right}
\end{equation}
Let us expand the left-hand side using the structure constants of the algebra.
\begin{equation*}
\sum_{\sigma \sim \mu}  \psharp{\alpha}(\lambda) \cdot \psharp{\beta \cdot \sigma}(\lambda)
=
\sum_{\sigma \sim \mu} \sum_{\rho} f^{\rho}_{\alpha(\beta\cdot\sigma)} \cdot \psharp{\rho}(\lambda).
\end{equation*}
From Proposition~\ref{prop:structure-coefficients}, $f^{\rho}_{\alpha(\beta\cdot\sigma)}$ is zero whenever $|\rho| < k$
(this follows from $|X| \geq |X_1| = |\alpha| = k$).
Furthermore, $\psharp{\rho}(\lambda)$ is by definition zero whenever $|\rho| > |\lambda| = k$.  Thus,
\begin{equation*}
\sum_{\sigma \sim \mu}  \psharp{\alpha}(\lambda) \cdot \psharp{\beta \cdot \sigma}(\lambda)
=
\sum_{\sigma \sim \mu} \sum_{\rho \vdash k} f^{\rho}_{\alpha(\beta\cdot\sigma)} \cdot \psharp{\rho}(\lambda).
\end{equation*}
By Corollary~\ref{cor:structure-coefficients-rewrite}, the coefficient on the polynomial $\psharp{\rho}$ in this expression is
\begin{align*}
\sum_{\sigma \sim \mu}f^{\rho}_{\alpha(\beta\cdot\sigma)}
&= \sum_{\sigma \sim \mu}k! \Pr_{\substack{s_1\sim \alpha\\ s_2\sim \beta \sigma}}[s_1 s_2 \sim \rho]\\
&= \frac{k!k!}{z_\mu} \cdot \Pr_{\substack{\sigma \sim \mu\\s_1\sim \alpha\\ s_2\sim \beta \sigma}}[s_1 s_2 \sim \rho].
\end{align*}
Here the second equality uses the fact that $k!/z_\mu$ counts the number of permutations of shape $\mu$.

Now, I'm going to make two claims.  First, I claim that the distribution on permutations given by drawing a random $\sigma \sim \mu$ and then outputting a random $s_2 \sim \beta \sigma$ is the same as the distribution given by drawing random $\sigma \sim \mu$ and $s_2 \sim \beta$ and outputting $s_2\cdot\sigma$.  This implies that
\begin{equation}
\sum_{\sigma \sim \mu}f^{\rho}_{\alpha(\beta\cdot\sigma)}
= \frac{k!k!}{z_\mu} \cdot \Pr_{\substack{\sigma \sim \mu\\s_1\sim \alpha\\ s_2\sim \beta}}[s_1 s_2\sigma \sim \rho].\label{eq:alpha-beta}
\end{equation}
Second, I claim that the distribution on permutations given by drawing random $s_1 \sim \alpha$ and $s_2 \sim \beta$ and outputting $s_1 s_2$ is the same as the distribution given by drawing random $s_1 \sim \alpha$ and $s_2 \sim \beta$ and outputting $s_2 s_1$.  This implies that
\begin{equation}
\sum_{\sigma \sim \mu}f^{\rho}_{\alpha(\beta\cdot\sigma)}
= \frac{k!k!}{z_\mu} \cdot \Pr_{\substack{\sigma \sim \mu\\s_1\sim \alpha\\ s_2\sim \beta}}[s_2 s_1\sigma \sim \rho].\label{eq:beta-alpha}
\end{equation}
Before defending these claims, I'll show how to finish the argument from here.  Equation~\eqref{eq:alpha-beta} gives the coefficient of $\psharp{\rho}$ in the left-hand side of Equation~\eqref{eq:left-right}.  Equation~\eqref{eq:beta-alpha} is just Equation~\eqref{eq:alpha-beta} with the roles of $\alpha$ and~$\beta$ swapped.  As a result, it gives the coefficient of $\psharp{\rho}$ in the right-hand side of Equation~\eqref{eq:left-right}.  Thus, the two sides are equal, proving the lemma.

I think the first claim is obvious, though I can't think of a simple proof of it off the top of my head.  As for the second claim, I believe this follows from the fact that the sums $S_1=\sum_{s_1\sim \alpha} s_1$ and $S_2=\sum_{s_2\sim \beta}s_2$ are both centers of the group algebra, meaning that $S_1 S_2 = S_2 S_1$.
\end{proof}
Now we can complete the proof of Lemma~\ref{lem:lemon-pi}.
\begin{align*}
\E_{\sigma \in S_k}\left[\chi_{\lambda}(\sigma) \cdot d^{\ell(\pi\cdot\sigma)}\right]
&= \E_{\sigma \in S_k}\left[\chi_{\lambda}(\pi^{-1}\sigma) \cdot d^{\ell(\sigma)}\right]\\
&= \frac{1}{k!}\sum_{\mu \vdash k} d^{\ell(\mu)}\sum_{\sigma \sim \mu} \chi_{\lambda}(\pi^{-1}\sigma).
\end{align*}
Now we apply Propositition~\ref{prop:character-identity} with $\beta = \pi^{-1}$ and $\alpha = e$, the identity permutation.
\begin{align*}
\frac{1}{k!}\sum_{\mu \vdash k} d^{\ell(\mu)}\sum_{\sigma \sim \mu} \chi_{\lambda}(\pi^{-1}\sigma)
& = \frac{1}{k!}\sum_{\mu \vdash k} d^{\ell(\mu)}\cdot\frac{\chi_{\lambda}(\pi^{-1})}{\chi_{\lambda}(e)}\sum_{\sigma \sim \mu} \chi_{\lambda}(\sigma)\\
& = \frac{\chi_{\lambda}(\pi^{-1})}{\chi_{\lambda}(e)}\cdot\E_{\sigma \in S_k}\left[\chi_{\lambda}(\sigma) \cdot d^{\ell(\sigma)}\right]\\
&=\frac{\chi_{\lambda}(\pi^{-1})}{\chi_{\lambda}(e)}\cdot\frac{\dim(\lambda)\cdot d^{\downarrow \lambda}}{k!},
\end{align*}
by Corollary~\ref{cor:pi-expectation}.
The lemma now follows because $\dim(\lambda) = \chi_{\lambda}(e)$ and $\chi_\lambda(\pi^{-1}) = \chi_\lambda(\pi)$.
\end{proof}

Next, we extend Lemma~\ref{lem:lemon-pi} to handle a more general type of permutations.
To do this, we will need some notation.
For a fixed $k \leq m \leq 2k$,
we will be considering permutations $\sigma$ on $[1,\ldots, k]$ and permutations $\pi$ on $[m-k+1, m]$.
We will write $\overline{\sigma}$ and $\overline{\pi}$ to extend these to permutations over $[1,\ldots, m]$, as in Proposition~\ref{prop:structure-coefficients}
(i.e., $\overline{\sigma}(i) = i$ for $i > k$ and $\overline{\pi}(i) = i$ for $i < m-k+1$).

Next, we will write $\pi_{\cup [1,k]}$ for the permutation on $[1,k]$ which is produced by writing $\overline{\pi}$ in cycle notation and deleting any element which is greater than $k$.  For example, if $\pi = (5,8) (4,6,7) (9)$ (so $\pi$ is a permutation on $[4, 9]$),
then $\overline{\pi} = (1)(2)(3)(5,8)(4,6,7)(9)$,
and so $\pi_{\cup [1,6]} = (1)(2)(3)(5) (4,6)$.
Finally, $\pi_{\cap [k+1, m]}$ is the permutation which is produced by writing $\pi$ in cycle notation and deleting any cycle which contains an element in $[1,k]$.  For example, using the previous example for $\pi$, we have that $\pi_{\cap [5,9]} = (5,8) (9)$.  In this case, we consider $\pi_{\cap[5,9]}$ to be a permutation only on the elements $\{5,8,9\}$ and more generally, $\pi_{\cap[k+1,m]}$ is a permutation only on those elements not deleted when producing it.

\begin{lemma}\label{lem:general-case}
For $k \leq m \leq 2k$, let $\pi$ be a permutation on $[m-k+1,m]$, and let $\lambda\vdash k$ be a partition.
Then
\begin{equation*}
\E_{\sigma \in S_k}\left[\chi_{\lambda}(\sigma) \cdot d^{\ell(\overline{\pi}\cdot\overline{\sigma})}\right]
= \frac{\chi_{\lambda}(\pi_{\cup [1,k]})\cdot d^{\downarrow \lambda}\cdot d^{\ell(\pi_{\cap[k+1,m]})}}{k!}.
\end{equation*}
\end{lemma}
\begin{proof}
We will first show that
$\ell(\overline{\pi}\cdot\overline{\sigma}) = \ell(\pi_{\cup[1,k]}\cdot \sigma)+\ell(\pi_{\cap[k+1,m]})$.
This follows because the permutation $\pi_{\cup[1,k]}\cdot \sigma$ can be formed by writing $\overline{\pi}\cdot\overline{\sigma}$ in cycle notation and deleting any element which is greater than $k$.  The only cycles which are fully deleted when forming $\pi_{\cup[1,k]}\cdot \sigma$ are those in $\overline{\pi}\cdot\overline{\sigma}$ which consist entirely of elements in $[k+1, m]$, of which there are $\ell(\pi_{\cap[k+1,m]})$ cycles.

This means that
\begin{align*}
\E_{\sigma \in S_k}\left[\chi_{\lambda}(\sigma) \cdot d^{\ell(\overline{\pi}\cdot\overline{\sigma})}\right]
& = \E_{\sigma \in S_k}\left[\chi_{\lambda}(\sigma) \cdot d^{\ell(\pi_{\cup[1,k]}\cdot \sigma)+\ell(\pi_{\cap[k+1,m]})}\right]\\
& = d^{\ell(\pi_{\cap[k+1,m]})}\cdot\E_{\sigma \in S_k}\left[\chi_{\lambda}(\sigma) \cdot d^{\ell(\pi_{\cup[1,k]}\cdot \sigma)}\right]\\
&= d^{\ell(\pi_{\cap[k+1,m]})}\cdot\frac{\chi_{\lambda}(\pi_{\cup[1,k]})\cdot d^{\downarrow \lambda}}{k!},
\end{align*}
where the last step follows from Lemma~\ref{lem:lemon-pi}.  This proves the lemma.
\end{proof}

\begin{proposition}\label{prop:cancellation}
Let $i \leq j$ and $j \leq k \leq i+j$.  Let $\pi \in S_j$ have no fixed points, and let $\alpha \in S_{[k-i+1, k]}$.  Then
\begin{equation*}
\sum_{|\rho|=j} \frac{\chi_\rho(\pi)}{d^{\downarrow \rho}} \E_{\beta \in S_j}\left[\chi_\rho(\beta) d^{\ell(\overline{\alpha}\cdot\overline{\beta})}\right]
=
\left\{
\begin{array}{cl}
z_\pi/k!&\text{if $i=j=k$ and $\alpha \sim \pi$},\\
0&\text{otherwise.}
\end{array}\right.
\end{equation*}
\end{proposition}
\begin{proof}
By Lemma~\ref{lem:general-case},
\begin{align*}
\sum_{|\rho|=j} \frac{\chi_\rho(\pi)}{d^{\downarrow \rho}} \E_{\beta \in S_j}\left[\chi_\rho(\beta) d^{\ell(\overline{\alpha}\cdot\overline{\beta})}\right]
& = \sum_{|\rho|=j} \frac{\chi_\rho(\pi)}{d^{\downarrow \rho}} \cdot
	\frac{\chi_{\rho}(\alpha_{\cup [1,j]})\cdot d^{\downarrow \rho}\cdot d^{\ell(\alpha_{\cap[j+1,k]})}}{j!}\\
&= \frac{d^{\ell(\alpha_{\cap[j+1,k]})}}{j!} \sum_{|\rho|=j} \chi_{\rho}(\pi)\cdot \chi_{\rho}(\alpha_{\cup[1,j]}).
\end{align*}
Now, if $k > j$, then $\alpha_{\cup[1,j]}$ has at least one fixed point.
The same holds if $k =j$ but $i < j$.
This means that in both of these cases,
$\alpha_{\cup[1,j]}$ has a different cycle structure than $\pi$,
and so by the Schur orthogonality relations (i.e., Fact~XXX),
the sum $\sum_{|\rho|=j} \chi_{\rho}(\pi)\cdot \chi_{\rho}(\alpha_{\cup[1,j]})$ is zero.

This leaves the case when $i = j = k$.
In this case, $\alpha_{\cup[1,j]} = \alpha$,
and so the sum $\sum_{|\rho|=j} \chi_{\rho}(\pi)\cdot \chi_{\rho}(\alpha)$
is nonzero only when $\alpha \sim \pi$, in which case it equals $z_\pi$.  This completes the lemma.
\end{proof}

Now we give an expression for the expected value of $s^*_\mu(\blambda) s^*_\rho(\blambda)$.
There are several ways one could imagine evaluating this expectation.
One is to use the Littlewood-Richardson rule for factorial Schur functions given in~\cite{MS99}
(in particular, their Proposition~$4.2$)
to write $s^*_\mu s^*_\rho$ as a linear combination of $s^*$ polynomials, and then compute the expectation of these polynomials using our Proposition~\ref{prop:shifted-expectation}.  Unfortunately,
the coefficients given by~\cite{MS99} are somewhat difficult to work with.
Instead, we expand both $s*_\mu$ and $s^*_\rho$
into linear combinations of $\psharp{}$ polynomials, and then multiply these using Proposition~\ref{prop:structure-coefficients}.
\begin{proposition}\label{prop:schur-squared}
Let $\mu \vdash i$ and $\rho \vdash j$, where $i \leq j$.  Then
\begin{equation*}
\E_{\blambda\sim \SWunif{n}{d}}\left[s^*_\mu(\blambda) s^*_\rho(\blambda)\right]
= \sum_{k=j}^{i+j} \left(C_{i j}^k \frac{i! j!}{k!}\right)\cdot \frac{n^{\downarrow k}}{d^k}\cdot
\E_{\alpha, \beta}\left[\chi_{\mu}(\alpha) \chi_\rho(\beta) d^{\ell(\overline{\alpha}\cdot\overline{\beta})}\right],
\end{equation*}
where $C_{i j}^k$ is as defined in Corollary~\ref{cor:structure-coefficients-rewrite}.
Here $\alpha$ is a uniformly random permutation on $[k-i+1, \ldots, k]$ and $\beta$ is a uniformly random permutation on $[1,j]$.
\end{proposition}
\begin{proof}
\end{proof}

\begin{proposition}\label{prop:almost-last-prop}
Let $i \leq j$, and let $\mu \vdash i$, $\rho\vdash j$.  Let $\pi \in S_j$ have no fixed points. Then
\begin{equation*}
\sum_{|\rho|=j} \frac{\chi_\rho(\pi)}{d^{\downarrow \rho}} \E_{\blambda\sim \SWunif{n}{d}}\left[s^*_\mu(\blambda) s^*_\rho(\blambda)\right]
=
\left\{
\begin{array}{cl}
\chi_\mu(\pi) \cdot \frac{n^{\downarrow i}}{d^i} & \text{if i = j},\\
0 & \text{otherwise.}
\end{array}
\right.
\end{equation*}
\end{proposition}
\begin{proof}
By Proposition~\ref{prop:schur-squared},
\begin{align*}
\sum_{|\rho|=j} \frac{\chi_\rho(\pi)}{d^{\downarrow \rho}} \E_{\blambda\sim \SWunif{n}{d}}\left[s^*_\mu(\blambda) s^*_\rho(\blambda)\right]
&= \sum_{|\rho|=j} \frac{\chi_\rho(\pi)}{d^{\downarrow \rho}}
	\sum_{k=j}^{i+j} \left(C_{i j}^k \frac{i! j!}{k!}\right)\cdot \frac{n^{\downarrow k}}{d^k}\cdot
		\E_{\alpha, \beta}\left[\chi_{\mu}(\alpha) \chi_\rho(\beta) d^{\ell(\overline{\alpha}\cdot\overline{\beta})}\right]\\
&= \sum_{k=j}^{i+j} \left(C_{i j}^k \frac{i! j!}{k!}\right)\cdot \frac{n^{\downarrow k}}{d^k}\cdot
	\E_\alpha \left[\chi_{\mu}(\alpha) \sum_{|\rho|=j} \frac{\chi_\rho(\pi)}{d^{\downarrow \rho}}
		\E_{\beta}\left[\chi_\rho(\beta) d^{\ell(\overline{\alpha}\cdot\overline{\beta})}\right]\right].
\end{align*}
By Proposition~\ref{prop:cancellation}, the inner summation is zero unless $i = j = k$, in which case it equals $z_\pi/i!$ when $\alpha \sim \pi$ and~$0$ otherwise.  Because $\alpha$ is a uniformly random permutation, it is conjugate with $\pi$ with probability $z_\pi^{-1}$, meaning that
this expression is $\chi_\mu(\pi) \cdot n^{\downarrow k}/d^k$, as guaranteed.
\end{proof}

\begin{proposition}
Let $i, j > 0$.  Then
\begin{equation*}
\sum_{\substack{|\mu|=i\\|\rho|=j}} \frac{s_\mu(x) s_\rho(x)}{d^{\downarrow\mu}\cdot d^{\downarrow\rho}}
\E_{\blambda\sim \SWunif{n}{d}}\left[s^*_\mu(\blambda) s^*_\rho(\blambda)\right]
=
\left\{
\begin{array}{cl}
\frac{n^{\downarrow i}}{d^i}
\sum_{|\mu|=i} s_\mu(x)^2/d^{\downarrow \mu} & \text{if $i = j$},\\
0 & \text{otherwise}.
\end{array}
\right.
\end{equation*}
\end{proposition}
\begin{proof}
Assume without loss of generality that $j \geq i$.  Then
\begin{align*}
\sum_{\substack{|\mu|=i\\|\rho|=j}} \frac{s_\mu(x) s_\rho(x)}{d^{\downarrow\mu}\cdot d^{\downarrow\rho}}
\E_{\blambda\sim \SWunif{n}{d}}\left[s^*_\mu(\blambda) s^*_\rho(\blambda)\right]
=&
\sum_{|\mu|=i} \frac{s_\mu(x)}{d^{\downarrow \mu}}
	\sum_{|\rho|=j} \frac{s_\rho(x)}{d^{\downarrow \rho}}\E_{\blambda\sim \SWunif{n}{d}}\left[s^*_\mu(\blambda) s^*_\rho(\blambda)\right]\\
=&
\sum_{|\mu|=i} \frac{s_\mu(x)}{d^{\downarrow \mu}}
	\E_{\pi \in S_j}\left[p_{\pi}(x)
	\sum_{|\rho|=j} \frac{\chi_{\rho}(\pi)}{d^{\downarrow \rho}}\E_{\blambda\sim \SWunif{n}{d}}\left[s^*_\mu(\blambda) s^*_\rho(\blambda)\right]\right].
\end{align*}
Now, we note that because $x_1 + \ldots + x_d = 0$, $p_\pi(x) = 0$ whenever $\pi$ has a fixed point.
On the other hand, when $\pi$ has no fixed point, the inner summation is zero by Proposition~\ref{prop:almost-last-prop} unless $i=j$.
Otherwise, when $i=j$, Proposition~\ref{prop:almost-last-prop} shows that this expression equals
\begin{equation*}
\sum_{|\mu|=i} \frac{s_\mu(x)}{d^{\downarrow \mu}}
	\E_{\pi \in S_i}\left[p_{\pi}(x)
	\chi_\mu(\pi) \cdot \frac{n^{\downarrow i}}{d^i}\right]
= \frac{n^{\downarrow i}}{d^i}
\sum_{|\mu|=i} \frac{s_\mu(x)^2}{d^{\downarrow \mu}},
\end{equation*}
where here we have used the formula for Schur functions given in Fact~XXX.
\end{proof}
This completes the theorem, as by Equation~\eqref{eq:var-of-schurrat},
\begin{align*}
\Var[\schurrat{x}]
&= \sum_{|\mu|,|\rho|>0} \frac{s_\mu(x) s_\rho(x)}{d^{\downarrow\mu}\cdot d^{\downarrow\rho}}
\E_{\blambda\sim \SWunif{n}{d}}\left[s^*_\mu(\blambda) s^*_\rho(\blambda)\right]\\
&= \sum_{i, j > 0} \sum_{\substack{|\mu| = i\\|\rho|=j}}\frac{s_\mu(x) s_\rho(x)}{d^{\downarrow\mu}\cdot d^{\downarrow\rho}}
\E_{\blambda\sim \SWunif{n}{d}}\left[s^*_\mu(\blambda) s^*_\rho(\blambda)\right]\\
& = \sum_{i > 0}\frac{n^{\downarrow i}}{d^i}
\sum_{|\mu|=i} \frac{s_\mu(x)^2}{d^{\downarrow \mu}}.\qedhere
\end{align*}
\end{proof}
}

\bibliographystyle{alpha}
\bibliography{wright}

\end{document}